\PassOptionsToPackage{unicode}{hyperref}
\PassOptionsToPackage{hyphens}{url}
\PassOptionsToPackage{dvipsnames,svgnames,x11names}{xcolor}
\documentclass[
  12pt]{article}
\usepackage{dsfont}
\usepackage{amsmath,amssymb,amsthm}
\usepackage[pagebackref]{hyperref}
\usepackage{iftex}
\ifPDFTeX
  \usepackage[T1]{fontenc}
  \usepackage[utf8]{inputenc}
  \usepackage{textcomp} 
\else 
  \usepackage{unicode-math}
  \defaultfontfeatures{Scale=MatchLowercase}
  \defaultfontfeatures[\rmfamily]{Ligatures=TeX,Scale=1}
\fi
\usepackage{lmodern}
\ifPDFTeX\else  
\fi
\IfFileExists{upquote.sty}{\usepackage{upquote}}{}
\IfFileExists{microtype.sty}{
  \usepackage[]{microtype}
  \UseMicrotypeSet[protrusion]{basicmath} 
}{}
\makeatletter
\@ifundefined{KOMAClassName}{
  \IfFileExists{parskip.sty}{%
    \usepackage{parskip}
  }{
    \setlength{\parindent}{0pt}
    \setlength{\parskip}{6pt plus 2pt minus 1pt}}
}{
  \KOMAoptions{parskip=half}}
\makeatother
\usepackage{xcolor}
\makeatletter
\ifx\paragraph\undefined\else
  \let\oldparagraph\paragraph
  \renewcommand{\paragraph}{
    \@ifstar
      \xxxParagraphStar
      \xxxParagraphNoStar
  }
  \newcommand{\xxxParagraphStar}[1]{\oldparagraph*{#1}\mbox{}}
  \newcommand{\xxxParagraphNoStar}[1]{\oldparagraph{#1}\mbox{}}
\fi
\ifx\subparagraph\undefined\else
  \let\oldsubparagraph\subparagraph
  \renewcommand{\subparagraph}{
    \@ifstar
      \xxxSubParagraphStar
      \xxxSubParagraphNoStar
  }
  \newcommand{\xxxSubParagraphStar}[1]{\oldsubparagraph*{#1}\mbox{}}
  \newcommand{\xxxSubParagraphNoStar}[1]{\oldsubparagraph{#1}\mbox{}}
\fi
\makeatother

\newtheorem{theorem}{Theorem}

\usepackage{longtable,booktabs,array}
\usepackage{calc} 
\usepackage{etoolbox}
\makeatletter
\patchcmd\longtable{\par}{\if@noskipsec\mbox{}\fi\par}{}{}
\makeatother
\IfFileExists{footnotehyper.sty}{\usepackage{footnotehyper}}{\usepackage{footnote}}
\makesavenoteenv{longtable}
\usepackage{graphicx}
\makeatletter
\def\maxwidth{\ifdim\Gin@nat@width>\linewidth\linewidth\else\Gin@nat@width\fi}
\def\maxheight{\ifdim\Gin@nat@height>\textheight\textheight\else\Gin@nat@height\fi}
\makeatother
\setkeys{Gin}{width=\maxwidth,height=\maxheight,keepaspectratio}
\makeatletter
\def\fps@figure{htbp}
\makeatother

\makeatletter
\@ifpackageloaded{caption}{}{\usepackage{caption}}
\AtBeginDocument{%
\ifdefined\contentsname
  \renewcommand*\contentsname{Table of contents}
\else
  \newcommand\contentsname{Table of contents}
\fi
\ifdefined\listfigurename
  \renewcommand*\listfigurename{List of Figures}
\else
  \newcommand\listfigurename{List of Figures}
\fi
\ifdefined\listtablename
  \renewcommand*\listtablename{List of Tables}
\else
  \newcommand\listtablename{List of Tables}
\fi
\ifdefined\figurename
  \renewcommand*\figurename{Figure}
\else
  \newcommand\figurename{Figure}
\fi
\ifdefined\tablename
  \renewcommand*\tablename{Table}
\else
  \newcommand\tablename{Table}
\fi
}
\@ifpackageloaded{float}{}{\usepackage{float}}
\floatstyle{ruled}
\@ifundefined{c@chapter}{\newfloat{codelisting}{h}{lop}}{\newfloat{codelisting}{h}{lop}[chapter]}
\floatname{codelisting}{Listing}

\makeatother
\makeatletter
\@ifpackageloaded{caption}{}{\usepackage{caption}}
\@ifpackageloaded{subcaption}{}{\usepackage{subcaption}}
\makeatother

\ifLuaTeX
  \usepackage{selnolig}  
\fi
\usepackage[]{natbib}
\usepackage{bookmark}

\usepackage{xr-hyper}

\usepackage{units}            
\usepackage[ruled,vlined]{algorithm2e}  
\newcommand{\MF}{\mathcal{M}} 
\newcommand{\bx}{X}
\renewcommand{\P}{\mathbb{P}} 

\IfFileExists{xurl.sty}{\usepackage{xurl}}{} 
\hypersetup{
  pdftitle={Computer code validation via mixture model estimation},
  pdfauthor={Author 1; Author 2},
  pdfkeywords={3 to 6 keywords, that do not appear in the title},
  colorlinks=true,
  linkcolor={blue},
  filecolor={Maroon},
  citecolor={Blue},
  urlcolor={Blue},
  pdfcreator={LaTeX via pandoc}}

\begin{document}

\def\spacingset#1{\renewcommand{\baselinestretch}%
{#1}\small\normalsize} \spacingset{1}


  \title{\bf Computer code validation via mixture model estimation}

  \author{
Negar Soleimani\thanks{Université Paris-Saclay, AgroParisTech, INRAE, UMR MIA Paris-Saclay, 91120, Palaiseau, France.}
\and
Pierre Barbillon\footnotemark[1]
\and
Kaniav Kamary\thanks{INSA Lyon, CNRS, École Centrale Lyon, Université Claude Bernard Lyon 1, ICJ UMR 5208, France.}
\and
Merlin Keller\thanks{EDF R\&D PRISME, Chatou, France.}
}

  \maketitle

\bigskip
\begin{abstract}
When computer codes model complex physical systems, calibration alone is insufficient; one must also assess whether a discrepancy term is needed. In this paper, we study computer code validation through a Bayesian mixture approach that compares a pure-code model with a discrepancy-corrected model. The method relies on the posterior distribution of a mixture weight, which measures the relative support of the two competing distributions. Under the assumption that the code is linear in the calibration parameters, or can be well approximated by a linear surrogate, we show that mixture component-shared parameters can be used to combine flexible modeling, even when noninformative priors are assigned to some common parameters. Inference is performed using a Metropolis-within-Gibbs algorithm.
In addition, we introduce a thresholded allocation rule that complements the global mixture weight by providing a local diagnostic of where the discrepancy-corrected component is truly needed along the input domain. Beyond global model comparison, the proposed approach is also able to perform local model discrimination by identifying where the discrepancy-corrected component is truly needed along the input domain.
\end{abstract}


\noindent%

\vfill

\newpage

\section{Introduction}\label{sec-intro}
Since field experiments are often impractical or economically expensive, computer codes have been developed as substitutes for many complex physical systems \citep{gramacy2020}.  
At the core of such codes are mathematical equations describing physical phenomena.
Two fundamental challenges must then be addressed: verification and validation. Verification \citep{Roache1998} consists of ensuring that the code provides a sufficiently accurate approximation of the underlying physical model. Validation, however, addresses a different question, namely whether the physical model itself adequately represents the real phenomenon of interest. In practice, validation is closely linked to calibration, whose purpose is to reduce uncertainty in unknown model parameters by combining the code with experimental data \citep{Roy2011}. Throughout this work, we assume that a verified code is available.
In the Bayesian calibration framework introduced by \citet{kennedy2001,DHAKJCCJAC2004}, the computer code is represented by a function $f(x,\pmb{\theta})$, where $\pmb{\theta}$ denotes unknown calibration parameters. To account for possible systematic departures between the code and the physical system, a discrepancy function $\delta(x)$ is introduced, typically endowed with a Gaussian process prior. From this perspective, validation amounts to determining whether such a discrepancy term is required or whether the computer model alone provides an adequate representation of reality.
A major difficulty, however, is that the joint estimation of calibration parameters and discrepancy may lead to confounding and identifiability issues, resulting in biased inference on $\pmb{\theta}$ and degraded predictive performance when discrepancy is neglected \citep{brynjarsdottir2014}. These difficulties have motivated an extensive literature, including discrepancy priors that constrain the bias to remain small \citep{gu2018scaled}, orthogonal Gaussian-process constructions \citep{plumlee2017,huang2025local}, modular or sequential strategies \citep{wong2017,MAUPIN2020106818}, and more recent adaptive formulations of model discrepancy \citep{leoni2024adaptive}.
Several approaches have also been proposed specifically for model validation. For instance, \citet{MJBJOBPRSJA2007} assessed predictive accuracy through tolerance bounds around the true process, whereas \citet{GDMKPBAPEP2016} formulated validation as a Bayesian model selection problem between a pure-code model and a discrepancy-corrected model. From a Bayesian perspective, model comparison is most commonly performed through Bayes factors or posterior model probabilities 
\citep{JOBLRP1996,STBKPBNHA1997,Robert2007,JMMNSPCPRJR2014}, while predictive criteria such as leave-one-out cross-validation and WAIC are also widely used 
\citep{vehtari2017,yao2018,GayaKetz2024}.
Although powerful, Bayes-factor-based procedures based on marginal likelihoods \citep{Knuth2014,naderi2026approximating} can become delicate in complex hierarchical settings, especially when weakly informative or improper priors are desirable. As a result, alternative approaches have gained increasing attention and Bayesian model averaging has been developed as a way to account for model uncertainty rather than selecting a single model \citep{Fragoso2015}. \citet{KKKMCPRJR2014} proposed a paradigm for Bayesian hypothesis testing in which the competing models are embedded within an encompassing mixture model. Rather than relying on Bayes factors, inference is based on the posterior distribution of the mixture weights, which quantifies the relative support for the competing components. This perspective is particularly attractive because it accommodates shared parameters across models, permits the use of improper priors under suitable conditions, and yields an interpretable measure of support. Subsequent work has further developed this framework from both methodological and computational viewpoints. In particular, \citet{Newman2026} considered efficient sampling strategies and model selection in high-dimensional mixture settings. Related ideas have also been adapted to specific application domains, including Bayesian model selection for generalized linear mixed models \citep{Xu2023}, while alternative criteria incorporating geometric or predictive information have recently been proposed to improve robustness in complex settings \citep{Ramirez2026}.
In this paper, we adapt the mixture-estimation approach of \citet{KKKMCPRJR2014} to computer code validation by comparing a pure-code model with a discrepancy-corrected model. We work under the assumption that the code is linear in the calibration parameters, or can be well approximated by a linear surrogate. The presence of a Gaussian-process discrepancy makes the present setting more challenging than the original mixture-testing framework, since it induces dependence among observations and raises additional identifiability issues.  
To handle this structure, we embed the pure-code model $\MF_0$ and the discrepancy-corrected model $\MF_1$ into a single encompassing mixture $\MF_\alpha$ whose two components share the calibration parameters $\pmb{\theta}$ and the noise variance $\lambda^2$ and differ only through the Gaussian-process discrepancy $\delta(X)$.
We show that, under suitable conditions, the parameter-sharing structure allows improper priors on the common calibration parameters while preserving posterior propriety. Since the posterior distribution has no standard form, we develop a Metropolis-within-Gibbs algorithm for the posterior inference. The resulting framework performs calibration and validation simultaneously. In addition, we introduce a thresholded allocation device that complements the global mixture weight $\alpha$ with a local diagnostic, indicating where along the input domain the discrepancy-corrected component is effectively needed.
The paper is organized as follows. Section~\ref{sec-meth} summarizes statistical modeling and methodologies starting with an introduction to the competing models, the mixture formulation, 
and then the Bayesian analysis of the model. The MCMC algorithm is described in Section~\ref{sec:algo}.
~Section~\ref{sec:case} illustrates the performance of the approach on synthetic experiments and a real dataset called {\it ball-drop dataset} and the paper ends with a conclusion in Section~\ref{sec-conc}.

\section{Statistical modeling framework}\label{sec-meth}
{\bf \large Models in competition.} The computer code is represented by a parametric function $f(x,\pmb{\theta})$ that aims at approximating the physical quantity of interest $r(x)\in\mathbb{R}$, where $x\in\mathcal{X}\subset\mathbb{R}^{p}$ is the controllable input and $\pmb{\theta}\in\Theta\subset\mathbb{R}^{d}$ is the vector of calibration parameters. When the code is costly, it is common to replace $f$ by a surrogate obtained via Gaussian-process emulation \citep{sacks1989}, polynomial expansions \citep{an2001quasi}, or a linearization approximation \citep{Bachoc2014}. We assume here that the code is linear in $\pmb{\theta}$ \citep{GDMKPBAPEP2016}, or that a linear surrogate $f(x, \pmb{\theta})=g(x)\pmb{\theta}$ is adequate; richer surrogate models \citep{kennedy2001} are beyond the scope of this paper.

To account for systematic differences between $r(x)$ and $g(x)\pmb{\theta}$, \cite{kennedy2001} introduced a discrepancy term
$\delta(x)=r(x)-g(x)\pmb{\theta}^*,$
where $\pmb{\theta}^*$ is the true value of $\pmb{\theta}$. 
If $X=(x_1,\ldots,x_n)^\top$ and $Y=(y_1,\ldots,y_n)^\top$ denote the design of 
observable
inputs and the corresponding vector of field measurements, respectively, we consider two competing models for the whole data set: for all $i\in\{1,\ldots,n\}$ the data follows either
\begin{align}\label{eq:2}
\MF_0: y_i &= g(x_i)\pmb{\theta} + \epsilon_i, \text{ or } \MF_1: y_i = g(x_i)\pmb{\theta} + \delta(x_i) + \epsilon_i,\quad\text{with } \epsilon_i\sim\mathcal{N}(0,\lambda^2).
\end{align}
The model $\MF_0$ treats the code as adequate up to measurement error, whereas $\MF_1$ accounts for possible model inadequacy by introducing a discrepancy function $\delta(X)$.
Assessing whether the discrepancy is needed to be added to the code, can thus be considered as a model-comparison problem between $\MF_0$ and $\MF_1$ \citep{GDMKPBAPEP2016}.
On the other hand, if we associate with each $y_i$ a latent variable allocation $\zeta_i\in\{0,1\}$, then including the discrepancy in the model can be represented as  
$
y_i = g(x_i)\pmb{\theta} + \zeta_i\,\delta(x_i) + \epsilon_i,
$
where $\zeta_i=0$ delivers $\MF_0$ and $\zeta_i=1$ issues the model $\MF_1$. The model selection procedure in \eqref{eq:2} is then translated to the testing hypothesis problem between $\mathcal{H}_0 : \forall i\in\{1,\ldots,n\},\zeta_i=0 $ and
$\mathcal{H}_1 :\forall i\in\{1,\ldots,n\}, \zeta_i=1$. Note that intermediate configurations are also permitted, quantified globally by $\alpha\in[0,1]$ and locally by the allocation rule of Section~\ref{sec:threshold}.

{\bf \large Mixture of $\MF_0$ and $\MF_1$.} Following \citet{KKKMCPRJR2014}, the models in \eqref{eq:2} are embedded into an encompassing mixture defined as: for all $i\in\{1,\ldots,n\}$
\begin{equation}\label{eq:1}
\MF_\alpha: y_i\sim \alpha \ell_{\MF_0}(\pmb{\theta}, \lambda^2;y_i, x_i)+(1-\alpha)\ell_{\MF_1}(\pmb{\theta}, \lambda^2, \delta;y_i, x_i), \quad \alpha\in[0,1]
\end{equation}
where $\ell_{\MF_0}$ and $\ell_{\MF_1}$ denote the following Gaussian densities 
\begin{align}\label{eq:7}
\ell_{\MF_0}(\pmb{\theta}, \lambda^2; y_i, x_i)
&=
\frac{1}{\sqrt{2\pi\lambda^2}}\exp\!\left(-\tfrac{1}{2\lambda^2}(y_i - g(x_i)\pmb{\theta})^2 \right),
\\
\ell_{\MF_1}(\pmb{\theta}, \lambda^2, \delta; y_i, x_i)
&=
\frac{1}{\sqrt{2\pi\lambda^2}}\exp\!\left(-\tfrac{1}{2\lambda^2}(y_i - g(x_i)\pmb{\theta} - \delta(x_i))^2 \right).\nonumber
\end{align}
By treating the discrepancy function $\delta(X)$ as an unknown quantity to be estimated jointly with the calibration parameter $\pmb{\theta}$ and the noise variance $\lambda^2$, and assigning a Gaussian-process prior to $\delta$, we obtain conditional independence of the observations $y_i$ and can perform standard Bayesian inference for the model \eqref{eq:1}.
\citet{KKKMCPRJR2014} demonstrated that, under suitable conditions, the posterior of $\alpha$ concentrates near 1 when the data are generated from $\MF_0$. In the mixture model \eqref{eq:1}, model inference is based on the posterior of $\alpha$ which governs the global competition between the two components : mass near $\alpha=1$ favours $\MF_0$, whereas mass near $\alpha=0$ favours $\MF_1$. 

{\bf \large Prior modeling.} We assign to the discrepancy vector $\delta(X) = (\delta(x_1),\ldots,\delta(x_n))^\top$ a zero-mean Gaussian process prior \citep{CERCKIW2006,kennedy2001} as follows
\begin{equation}\label{dltp}
\delta(X)\sim \mathcal{GP}(0, \Sigma_{\delta}), \quad \Sigma_{\delta}=\sigma_\delta^2\,\text{Corr}_{\gamma_\delta}(x_i, x_{i}'),
\end{equation}
where $\sigma_\delta^2>0$ is a scale parameter and $\text{Corr}_{\gamma_\delta}$ is a stationary correlation function with correlation length $\gamma_\delta$. In the simulation study section, we use the exponential correlation
$\text{Corr}_{\gamma_\delta}(x_i, x_{i}')
=
\exp\!\left(-||x_i-x_{i}'||_1/\gamma_\delta \right)$. 
The $x_i$s are re-scaled to $[0,1]^p$, so that $\gamma_\delta\in(0,1)$ governs the smoothness and effective range of dependence.
For the shared parameters $\pmb{\theta}$ and $\lambda^2$ in \eqref{eq:1}, we 
use the noninformative Jeffreys prior $\pi(\pmb{\theta}, \lambda^2)\,\propto\, \nicefrac{1}{(\lambda^2)^{d/2+1}}$
consisting of the flat distribution on $\pmb{\theta}$ and the scale-invariant density on $\lambda^2$ \citep{OJBVDBS2001}. For brevity, we
postpone to the appendix \ref{sup:propriety} the proof that the posterior is proper.
The discrepancy scale parameter is reparameterized as $\sigma_\delta^2 =\nicefrac{\lambda^2}{k}$, so that the parameter $k\in(\underline{k},1)$  (where $\underline{k}>0$) controls the ratio between measurement-noise and discrepancy scale parameter.
Note that the constraint $\sigma_\delta^2 > \lambda^2$ avoids degenerate cases in which the discrepancy is 
negligible.
We therefore consider $k \sim \mathcal{U}(\underline{k}, 1)$ with $\underline{k}=0.015$, a symmetric Beta prior on the mixture weight $\alpha \sim \mathcal{B}(a_0, a_0)$, and a weakly informative prior on the correlation length $\gamma_\delta \sim \mathcal{U}(0.1,1)$. Note that the lower bound in $\mathcal{U}(0.1,1)$ excludes extremely small correlation lengths for which the discrepancy becomes indistinguishable from white noise.

{\bf \large Identifiability.} A fundamental difficulty in calibration is the confounding between the calibration parameters $\pmb{\theta}$ and the discrepancy $\delta(X)$: the two can compensate for each other, since only their sum $g(x)\pmb{\theta}+\delta(x)$ is informed by the data. A possible remedy is the orthogonal Gaussian-process (OGP) construction of \citet{plumlee2017}, which enforces $\int_{\mathcal{X}} g(\xi)\,\delta(\xi)\,d\xi = 0$ and thereby restricts the discrepancy to the subspace orthogonal to the code response, removing this particular source of confounding. We implemented this construction and compared it with the classical GP prior; the details and the corresponding numerical results are reported in the Supplementary Material.
In our experiments, the orthogonal prior leads to posterior results very similar to those obtained with the classical Gaussian-process prior, in particular for the mixture weight $\alpha$. This does not imply that the $\pmb{\theta}$--$\delta(X)$ confounding has been fully removed. Rather, it suggests that, in our examples, enforcing orthogonality alone does not substantially alter the posterior allocation mechanism or the behaviour of $\alpha$. The remaining inferential difficulties may therefore involve both residual calibration--discrepancy confounding and, when the correlation length is very small, an additional ambiguity between a rough discrepancy and the measurement error.

{\bf \large Parameter posterior distributions.} As the joint posterior distribution of $(\delta,\pmb{\theta},\lambda^2,\alpha,k,\gamma_{\delta})$ has no closed form, we proceed by working with the corresponding full conditional distributions. Recall that $\zeta_i\in\{0,1\}$ denotes a latent variables indicating which mixture component in \eqref{eq:1} the observation $y_i$ belongs to, then the mixture model \eqref{eq:1} leads to the conditional distribution of $\zeta_i$ given by
\begin{equation}\label{zeta}
\P(\zeta_i=j\mid\cdot)\propto \left(\alpha\,f_{\MF_0}(y_i\mid x_i,\pmb{\theta},\lambda^2)\right)^{1-j}\left((1-\alpha)\,f_{\MF_1}(y_i\mid x_i,\delta,\pmb{\theta},\lambda^2,k,\gamma_\delta)\right)^{j}.
\end{equation}
Let $m=\sum_{i=1}^n \zeta_i$ be the number of observations allocated to $\MF_1$, with corresponding inputs $\pmb{x_m}=(x_{i_1},\ldots,x_{i_m})^\top$ and the $m-$length vector of observations $\pmb{y_m}$; similarly denote the remaining $n-m$ indices by $\pmb{x_{n-m}},\pmb{y_{n-m}}$. Under the GP prior \eqref{dltp}, the joint distribution of $(\delta,\pmb{y_m})$ 
is indeed Gaussian (see e.g.\ \citealp{CERCKIW2006}) and the conditional distribution of $\delta$ can be derived as follows
\begin{align}\label{musigdlt}
\delta \mid \pmb{y}, \zeta, \pmb{\theta}, \lambda^2, k,& \gamma_\delta, \alpha \sim \mathcal{N}_n(\hat{\mu}_\delta, \hat{\Sigma}_\delta),\nonumber\\
\hat{\mu}_\delta
= \Sigma_{\pmb{\delta},\pmb{y_m}}\Sigma_{\pmb{y_m},\pmb{y_m}}^{-1}\bigl[\pmb{y_m} - g(\pmb{x_m})\pmb{\theta}\bigr],& \qquad
\hat{\Sigma}_\delta
= \Sigma_{\pmb{\delta},\pmb{\delta}} - \Sigma_{\pmb{\delta},\pmb{y_m}}\Sigma_{\pmb{y_m},\pmb{y_m}}^{-1}\Sigma_{\pmb{y_m},\pmb{\delta}},
\end{align}
where $[\Sigma_{\pmb{\delta},\pmb{\delta}}]_{k,k'}=\sigma_\delta^2\,\text{Corr}_\delta(x_k,x_{k'})$, $[\Sigma_{\pmb{\delta},\pmb{y_m}}]_{k,i_l}=\sigma_\delta^2\,\text{Corr}_\delta(x_k,x_{i_l})$, and $[\Sigma_{\pmb{y_m},\pmb{y_m}}]_{i_l,i_{l'}}=\lambda^2\,\mathds{I}_{l=l'}+\sigma_\delta^2\,\text{Corr}_\delta(x_{i_l},x_{i_{l'}})$, for all $k,k'=1,\ldots,n$; $l,l'=1, \ldots, m$. 
The conditional posterior of the parameters $\pmb{\theta},\lambda^2, \alpha$, are given as
\begin{align}\label{eq:full-conditionals-main}
\pmb{\theta}\mid\pmb{y},\zeta,\delta,\lambda^2&,k,\alpha,\gamma_\delta
\sim \mathcal{N}_d(\hat{\mu}_{\pmb{\theta}}, \hat{\Sigma}_{\pmb{\theta}}),\nonumber\\
\hat{\Sigma}_{\pmb{\theta}}=\lambda^2 A^{-1}; \quad \hat{\mu}_{\pmb{\theta}}=A^{-1}\!&\left[g^\top(\pmb{x}_{n-m})\,\pmb{y}_{n-m}+g^\top(\pmb{x}_{m})\bigl(\pmb{y}_{m}-\delta(\pmb{x}_{m})\bigr)\right]\nonumber\\
A=g^\top(\pmb{x}_{n-m})&g(\pmb{x}_{n-m})+g^\top(\pmb{x}_{m})g(\pmb{x}_{m})\nonumber\nonumber\\
\lambda^2\mid\pmb{y},\zeta,\delta,\pmb{\theta}&,k,\alpha,\gamma_\delta
\sim \mathcal{IG}(\hat{a}_\lambda, \hat{b}_\lambda),\nonumber\nonumber\\
\hat{a}_\lambda = n+\frac{d}{2},\quad \hat{b}_\lambda
=
\frac12\sum_{i:\zeta_i=0}\bigl(y_i-g(x_i)\pmb{\theta}\bigr)^2
&+
\frac12\sum_{i:\zeta_i=1}\bigl(y_i-g(x_i)\pmb{\theta}-\delta(x_i)\bigr)^2
+
\frac{k}{2}\,
\delta(\pmb{x})^\top
\mathrm{Corr}_{\pmb{x}}^{-1}
\delta(\pmb{x})\nonumber\\
\alpha\mid\pmb{y},\zeta,\delta,\pmb{\theta},\lambda^2,k,\gamma_\delta
&\sim \mathcal{B}\mathrm{eta}(n-m+a_0,\, m+a_0),
\end{align}
and for the variance-ratio parameter $k$, the full conditional posterior is 
\begin{equation}\label{kpost}
\pi(k\mid\cdot)\propto
k^{n/2}\exp\left(
-\frac{1}{2\lambda^2}
k\,
\delta(\pmb{x})^\top
\mathrm{Corr}_{\pmb{x}}^{-1}
\delta(\pmb{x})
\right)\mathds{I}_{(\underline{k},1)}(k),
\end{equation}
that is, a Gamma distribution truncated to $(\underline{k},1)$.
The full conditional of $\gamma_\delta$ is not available in closed form and is therefore updated by a random-walk Metropolis--Hastings step, with proposal scale tuned to obtain a reasonable acceptance rate.

\section{MCMC algorithm}\label{sec:algo}         
The parameter space of the mixture model $\MF_\alpha$, has dimension $d+2n+4$, corresponding to $\pmb{\theta}$, $\delta(X)$,$(\zeta_j)_{1\le j\le n}$, $\lambda^2$, $\alpha$, $k$,$\gamma_\delta$. 
Since all full conditional distributions are not available in closed form, we therefore adopt a Metropolis-within-Gibbs sampler, see Algorithm~\ref{algo:mwg}, that updates blocks of parameters sequentially, using closed-form full conditionals where available and Metropolis--Hastings steps otherwise.

\begin{algorithm}[H]
\DontPrintSemicolon
\textbf{Initialization:} Choose initial values $(\delta^{(0)}, \pmb{\theta}^{(0)}, \lambda^{(0)}, \alpha^{(0)}, k^{(0)}, \gamma_\delta^{(0)})$.\\
\For{$t=1,\ldots,T$}{
  \textbf{(a)} Sample 
  $\pmb{\zeta}^{(t)}=(\zeta_1^{(t)},\ldots,\zeta_n^{(t)})$
  ~according to $ \Pr\!\left(\zeta_i=j \mid y_i,x_i,\delta^{(t-1)},\pmb{\theta}^{(t-1)},(\lambda^2)^{(t-1)},k^{(t-1)},\gamma_\delta^{(t-1)},\alpha^{(t-1)}\right),
    \qquad j\in\{0,1\}$.\;
  \textbf{(b)} Sample 
  $\delta^{(t)}|\pmb{y}, X,\pmb{\zeta}^{(t)},\pmb{\theta}^{(t-1)},(\lambda^2)^{(t-1)},k^{(t-1)},\gamma_\delta^{(t-1)} \sim
  \mathcal{N}_n\!\left(\hat{\mu}_\delta,\hat{\Sigma}_\delta\right)$.\;

  \textbf{(c)} Sample
  $\pmb{\theta}^{(t)}|\pmb{y}, X,\pmb{\zeta}^{(t)},\delta^{(t)},(\lambda^2)^{(t-1)},k^{(t-1)},\alpha^{(t-1)} \sim \mathcal{N}_d(\hat{\mu}_{\pmb{\theta}},\hat{\Sigma}_{\pmb{\theta}})$.\;

  \textbf{(d)} Sample
  $(\lambda^2)^{(t)}|\pmb{y}, X,\pmb{\zeta}^{(t)},\delta^{(t)},\pmb{\theta}^{(t)},k^{(t-1)},\alpha^{(t-1)},\gamma_\delta^{(t-1)} \sim \mathcal{IG}(\hat{a}_\lambda,\hat{b}_\lambda)$
  .\;

  \textbf{(e)} Sample $\alpha^{(t)}|\pmb{y}, X,\pmb{\zeta}^{(t)},\delta^{(t)},(\lambda^2)^{(t)},\pmb{\theta}^{(t)},k^{(t-1)},\gamma_\delta^{(t-1)} \sim \mathcal{B}(n-m^{(t)}+a_0; m^{(t)}+a_0)$
  .\;

  \textbf{(f)} Sample $k^{(t)}|\pmb{y}, X,\pmb{\zeta}^{(t)},\delta^{(t)},\pmb{\theta}^{(t)},(\lambda^2)^{(t)},\alpha^{(t)},\gamma_\delta^{(t-1)}$ from its full conditional distribution given in \eqref{kpost}
    .\;

  \textbf{(g)} Update $\gamma_\delta^{(t)}|\pmb{y}, X,\zeta^{(t)},\delta^{(t)},\pmb{\theta}^{(t)},(\lambda^2)^{(t)},\alpha^{(t)},k^{(t)}$ by a random-walk Metropolis--Hastings step 
  .\;
}
\caption{Metropolis-within-Gibbs algorithm}
\label{algo:mwg}
\end{algorithm}


We note here that when the competing models are well separated, the mixture approach ensures that the posterior distribution of the mixture weight $\alpha$ concentrates near 0 or 1, consistently selecting the true data-generating model \citep{KKKMCPRJR2014}. However, in the case of embedded models, where one model is a special case of the other, this separation becomes less clear. In particular, in our setting, $\MF_0$ is nested within $\MF_1$, corresponding to the special case where $\delta(X)=0$. As a result, the two mixture components are not well separated, and the larger model can approximate the smaller one arbitrarily well, leading to weak identifiability and confounding between components. Consequently, the posterior distribution of $\alpha$ may fail to concentrate at the extremes and instead exhibit boundary behavior or slower convergence, depending on the prior specification (see Section~\ref{sec:case}). This difficulty is therefore intrinsic to the model structure rather than an artifact of the inference algorithm, and should not be confused with label-switching issues encountered in symmetric mixture models. On the other hand, in many code calibration problems, the discrepancy may only affect a subset of the observations, for instance over a restricted region of the response domain. In such situations, the mixture weight $\alpha$ may fail to reflect the local nature of the discrepancy, as it aggregates information over the entire dataset. In such cases, $\alpha$ may again remain away from the extremes $\{0,1\}$ even when a local discrepancy is present. To make this local structure explicit, we introduce below a simple thresholding device in Algorithm \ref{algo:mwg} that forces the component $\MF_0$ to be selected whenever the discrepancy is too small to be practically distinguishable from measurement noise.

\subsection{Threshold-based allocation mechanism}\label{sec:threshold}
Conditional on the discrepancy vector $\delta(X)$, we define the following threshold indicator $\tau_i(s)=\mathds{I}\{|\delta(x_i)|>s\},$
where $s>0$ is a user-chosen threshold. We then replace the standard allocation prior $\P(\zeta_i=1\mid \alpha)=1-\alpha$ by 
\begin{equation*}
\P(\zeta_i=1\mid \delta,\alpha,s) = (1-\alpha)\,\tau_i(s),\qquad\P(\zeta_i=0\mid \delta,\alpha,s) = 1-(1-\alpha)\,\tau_i(s).
\end{equation*}
Hence, if $|\delta(x_i)|\le s$ the allocation is deterministic ($\zeta_i=0$ almost surely), whereas if $|\delta(x_i)|>s$ the usual mixture mechanism in \eqref{eq:1} applies and
so the threshold acts only through the allocation mechanism. 
This mechanism is implemented in the allocation step (a) of the algorithm \ref{algo:mwg} by setting the posterior probability of $\zeta_i=1$ to zero whenever $|\delta(x_i)|\le s$, while keeping the usual expression otherwise. 
This threshold has to be set with respect to the order of magnitude of the data at hand and a prior belief on the measurement noise.
Moreover, to avoid letting the deterministic region dominate inference on $\alpha$, the Beta update of $\alpha$ is computed using only the latent allocations $\zeta_i$ such that $\tau_i(s)=1$. Observations satisfying $|\delta(x_i)|\le s$ are excluded from the sufficient statistics entering the full conditional distribution of $\alpha$.
This preserves the interpretation of $\alpha$ as a global weight governing the competition between $\MF_0$ and $\MF_1$ \emph{where the discrepancy is potentially active}, while the sequence $\{\P(\zeta_i=1\mid Y)\}_{i=1}^n$ provides the desired point-wise information. 
Our thresholding mechanism is related in spirit to the latent
Gaussian-process approach for hidden simulator constraints proposed by \citep{MENZ2025102607}. However, in our method, the threshold is applied to the model discrepancy rather than to a latent process representing simulator failure.
In practice, this modification provides a local diagnostic: posterior summaries of $\P(\zeta_i=1\mid Y)$ highlight the regions where the discrepancy-corrected component is effectively required, even when $\alpha$ (as a global quantity) is less decisive.

\section{Numerical experiments}\label{sec:case}
In this section, we consider a collection of real ball-drop experiments \citep{de2020discovery}, in which the vertical trajectory of a falling ball is recorded using a high-speed camera. 
 Multiple ball types (Baseball, Blue Basketball and $\cdots$) have been used each with uniformly sampled time--height measurements on an interval $[t_0,t_n]$. The time variable will be re-scaled to $[0,1]$ to improve numerical conditioning and to match the domain of the GP prior. Under the assumption of idealized free fall without air resistance with zero initial velocity, the physical model reads
  $ f(t,\theta)
   \,=\,g(t) \theta=
   \begin{pmatrix}
   1 & -\tfrac{1}{2} t^{2}
   \end{pmatrix}
   \begin{pmatrix}
   h_0 \\ g_e
  \end{pmatrix},$
where $h_0$ and $g_e$ are the initial height and the gravitational acceleration, respectively. Real-world imperfections such as air drag, measurement noise, ball-specific aerodynamics, and small timing offsets, produce mild but meaningful departures from the idealized physics, making this collection well suited to assess our mixture-based approach. Before turning to the real data application, we first assess the performance of the proposed approach on synthetic datasets that closely resemble the structure and variability of the real data. This enables us to evaluate the ability of the model to recover known parameters in a controlled yet realistic setting, and to gain insight into its accuracy and robustness before applying it to real observations.

\subsection{Synthetic data}\label{sec:Syntheticdata}
We suppose that $\pmb{\theta}^*=(46.45, 9.8)^\top$, ${\lambda^2}^*=0.01$, $k^*=0.1$, and in the following, we proceed with some datasets simulated once from the model $\MF_0$, then from $\MF_1$ and finally from a mixture of both $\MF_0$ and $\MF_1$. Regarding the implementation of the algorithm \ref{algo:mwg}, we note that for each dataset, the posterior inference is based on 10,000 MCMC iterations, with the first 2,000 iterations discarded as burn-in.\\
\textbf{Data simulation under $\MF_0$:}
    For 50 datasets of size $n=45$ simulated from $\MF_0: y_i=g(x_i)\pmb{\theta}^*+\epsilon_i$ 
. Figure~\ref{fig:theta_alpha-classic} shows that the posterior means of $\pmb{\theta}$ concentrate around the true values used for simulating datasets which ensure the accuracy of the estimation procedure. 
\begin{figure}[h!]
    \centering
    \includegraphics[width=\textwidth, height=0.4\textheight]{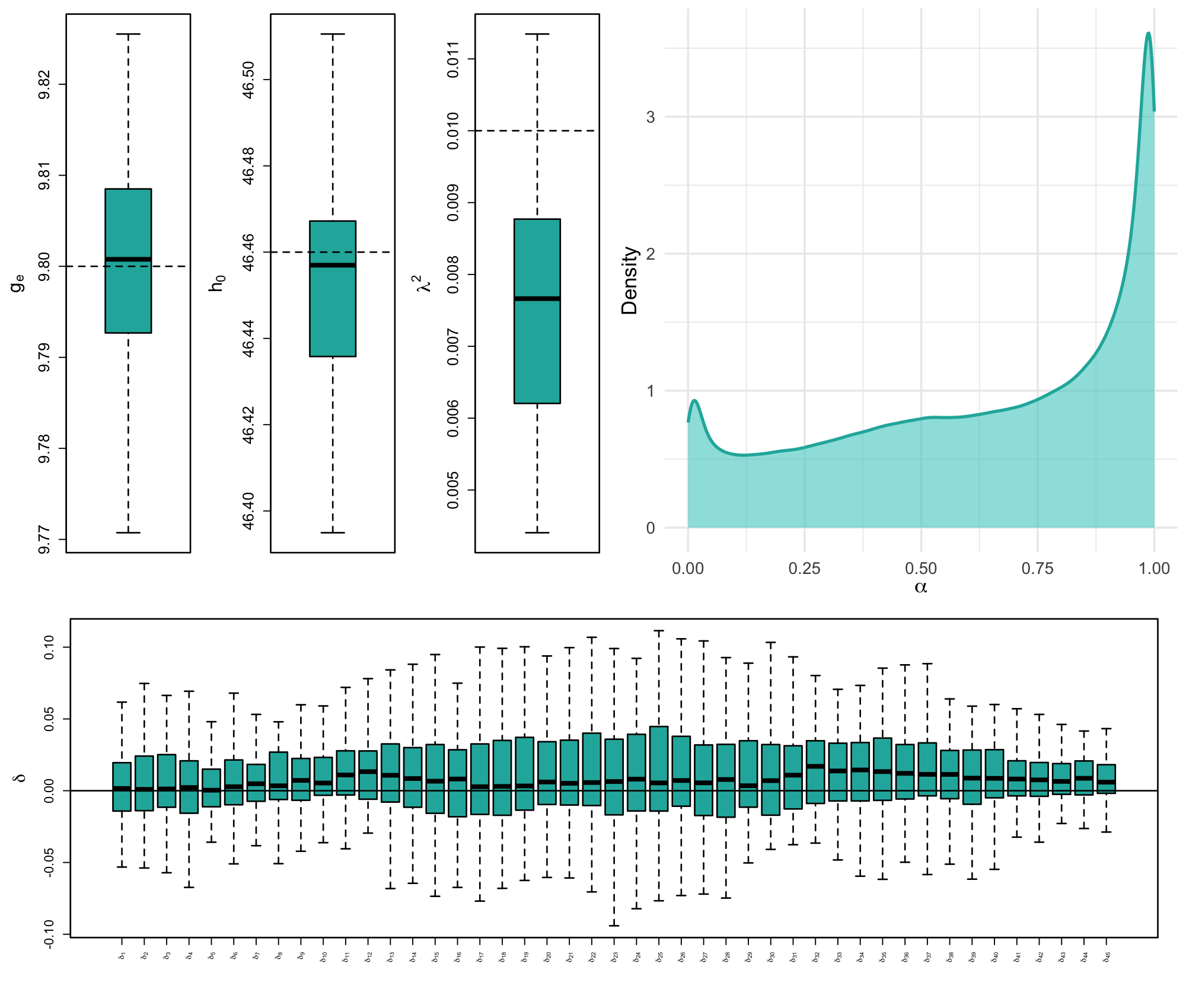}
    \caption{\small 50 data simulations under $\MF_0$. (Top Left) Distributions of the parameter posterior means compared to the true values (horizontal dotted lines). (Top Right) Pooled posterior densities of $\alpha$ across the 50 datasets. (Bottom) Distributions of the posterior means of $\delta$ computed across the 50 datasets.}
    \label{fig:theta_alpha-classic}
\end{figure}
Regarding the pooled posterior densities of $\alpha$, we observe a high concentration near 1, correctly favouring the no-discrepancy model. However, the bimodal behavior can be explained by the embedded nature of the competing models as we mentioned before. In other words, since $\MF_0$ is nested within 
$\MF_1$ with posterior $\delta$ values centered on zero as illustrated on the bottom plot of Figure~\ref{fig:theta_alpha-classic}, the model $\MF_1$ approximates $\MF_0$ arbitrarily well. Although the data are generated from $\MF_0$, the mixture formulation allows two competing explanations: either selecting 
$\MF_0$ directly (leading to $\alpha$ close to 1), or favoring 
$\MF_1$ with a discrepancy term $\delta$ close to zero. This lack of identifiability between the two components results in a posterior distribution of $\alpha$ that does not fully concentrate at 1, but instead exhibits a slight bimodality. Additional diagnostics for this experiment such as a posterior-predictive check against the simulated trajectories and the empirical coverage of $95\%$ credible intervals for $(g_e,h_0,\lambda^2)$ across the $50$ replicated datasets are reported in Section \ref{sup:m0_extra} 
~of the appendix.
On the other hand, the posterior means of $\lambda^2$ exhibit a slight downward bias relative to the true value, though small compared to the across-dataset variability—a behavior consistent with the well-known difficulty of estimating variance components in GP discrepancy models, where part of the residual variance is absorbed by $\delta$ even when its posterior concentrates near zero. This bias does not propagate to the calibration parameters or to the global model selection through $\alpha$.
Two complementary refinements of this baseline experiment are reported in Section \ref{sup:m0_sensitivity} of the appendix. First, applying the threshold-based allocation mechanism 
attenuates the residual bimodality of $\alpha$ by preventing negligible posterior discrepancy values from being assigned to $\MF_1$, thereby producing a sharper concentration around $\alpha=1$. Second, increasing the sample size to $n=100$ 
further reduces the across-dataset variability of the posterior means of $(g_e,h_0,\lambda^2)$ and tightens the pooled posterior of $\alpha$, in line with standard Bayesian asymptotic behavior. Combining both refinements yields more concentrated estimates and confirms the robustness of the proposed approach as the amount of information grows.

\textbf{Data simulation under $\MF_1$ :}
Since the correlation length $\gamma_\delta$ directly controls the smoothness of the Gaussian process discrepancy, its value can substantially affect the model posterior inference. In general, for inputs normalized to $[0,1]$, small values of $\gamma_\delta$ (typically $\gamma_\delta \lesssim 0.1$) generate highly irregular discrepancy realizations that may become nearly indistinguishable from additional white noise, thereby inducing confounding with the measurement noises. Intermediate values ($\gamma_\delta \in [0.2,0.5]$) generally produce smoother and more structured discrepancy functions, allowing meaningful departures from the mean model while remaining sufficiently identifiable. In contrast, larger values of $\gamma_\delta$ (typically $\gamma_\delta \gtrsim 0.5$) yield increasingly smooth discrepancy functions varying over broader regions of the input domain. When $\gamma_\delta$ approaches $1$, the discrepancy may behave as a slowly varying trend, which can in turn become difficult to distinguish from the mean structure $g(x)\pmb{\theta}$, leading to additional identifiability issues with the calibration parameters. Investigating a range of values of $\gamma_\delta$ is therefore important to assess the robustness of the proposed approach across different discrepancy regimes. We therefore study how the mixture model behaves when data are generated from the discrepancy-corrected model $\mathcal{M}_1$, with particular emphasis on the role of the discrepancy correlation length $\gamma_\delta$ in the posterior inference. 
For each $\gamma_\delta^\ast \in \{0.01,0.1,0.2,\ldots,0.9\}$, we generate $50$ datasets of size $n=45$ from the model $\MF_1$
~on the normalized time instants of the ball-drop experiment.
\begin{figure}[ht]
    \centering
    \includegraphics[width=0.9\textwidth]{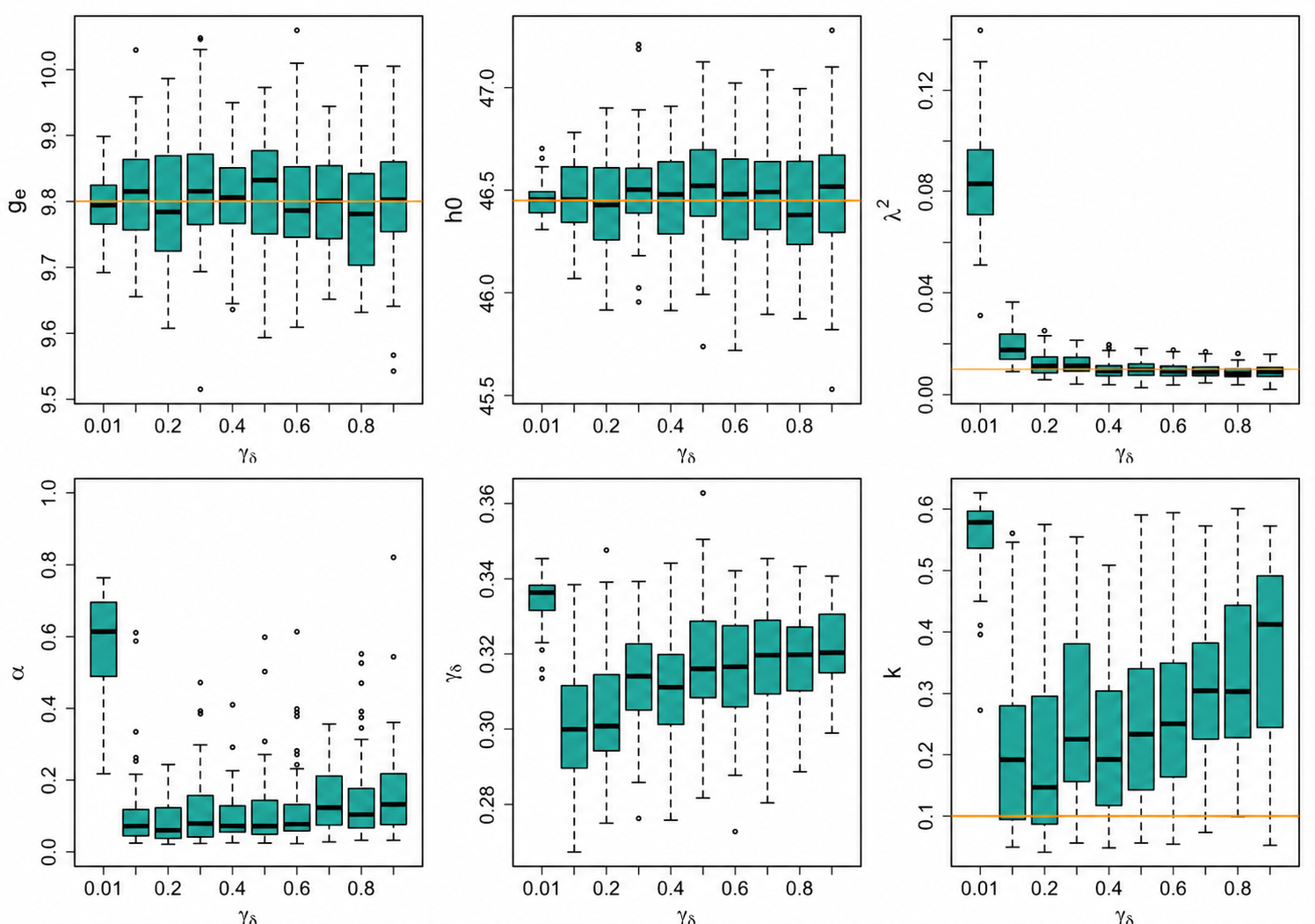}
    \caption{\small 50 data simulations under $\MF_1$. Each boxplot represents the posterior mean distribution obtained across the data replications and compared to the true values (horizontal orange lines).}
    \label{fig:m1-classicalGP}
\end{figure}
Figure~\ref{fig:m1-classicalGP} displays that the posterior means of $\alpha$ (bottom-left) concentrate near small values for all $\gamma_\delta^\ast \ge 0.1$, correctly favouring the discrepancy-corrected component. When $\gamma_\delta^\ast=0.01$, however, $\alpha$ shifts upward and becomes much more dispersed. This is because the discrepancy arising from a Gaussian process with a very small correlation length behaves like additional white noise, making it difficult to distinguish from the measurement error and leading again to identifiability issues. The calibration parameters $(h_0,g_e)$ are recovered well across the whole range. By contrast, the variance-related quantities $(\lambda^2,k,\gamma_\delta)$ vary more markedly, reflecting the usual trade-off in GP discrepancy models between amplitude and correlation length. Despite this partial non-identifiability, the posterior of $\alpha$ remains stable and continues to discriminate correctly.

\textbf{Simulation under a mixture of $\MF_0$ and $\MF_1$.}
\label{sec:sim_m1_threshold}
We now illustrate the threshold-based local diagnostic introduced in Section~\ref{sec:threshold} on a design in which the discrepancy is not present over a first part of the dataset ($i=1, \ldots, n_0$) and it is added afterwards. We keep the same observation grid as previous data simulations but modify first the simulation of the discrepancy $\delta^\ast$ from \eqref{dltp} by considering $k^*=0.01$, $\gamma^*_\delta=0.5$,
fixing then $n_0=20$ and simulate datasets from $
y_i = g(x_i)\pmb{\theta}^* + \mathds{I}\{i>n_0\}\,\delta^*(x_i) + \epsilon_i; i=1,\dots,n
$.
~To obtain a more interpretable thresholding experiment, we retain only simulated discrepancy paths such that, on the active region $i>n_0$, the discrepancy is sign-coherent and positive $\delta^\ast(x_i)\ge 0; \forall i>n_0$. We do not impose any minimum amplitude on the active values: some of them may fall below the threshold level $s=0.3$ and others may lie above it. This yields a design in which the first part of the trajectory is locally compatible with $\MF_0$, while the second part contains a positive discrepancy of mixed amplitude across the threshold, providing a realistic test bed for the threshold-based local diagnostic of Section~\ref{sec:threshold}.
We evaluate different configurations combining two aspects: fixing $g_e=g_e^\ast$ versus updating $g_e$ within the MCMC algorithm, and using thresholding versus no thresholding. The cases with fixed $g_e$ are reported in the main text, while those with updated $g_e$ are provided in the Supplementary Material.
We note here that by fixing $g_e=g_e^\ast$, we aim at eliminating uncertainty in the physical model component and therefore isolate the effect of the mixture allocation mechanism from calibration--discrepancy confounding.
In the case of Bayesian inference without thresholding, the latent allocations are driven by the global mixture mechanism through $\alpha$. 

\begin{figure}[h!]
        \centering
        \includegraphics[width=0.5\textwidth]{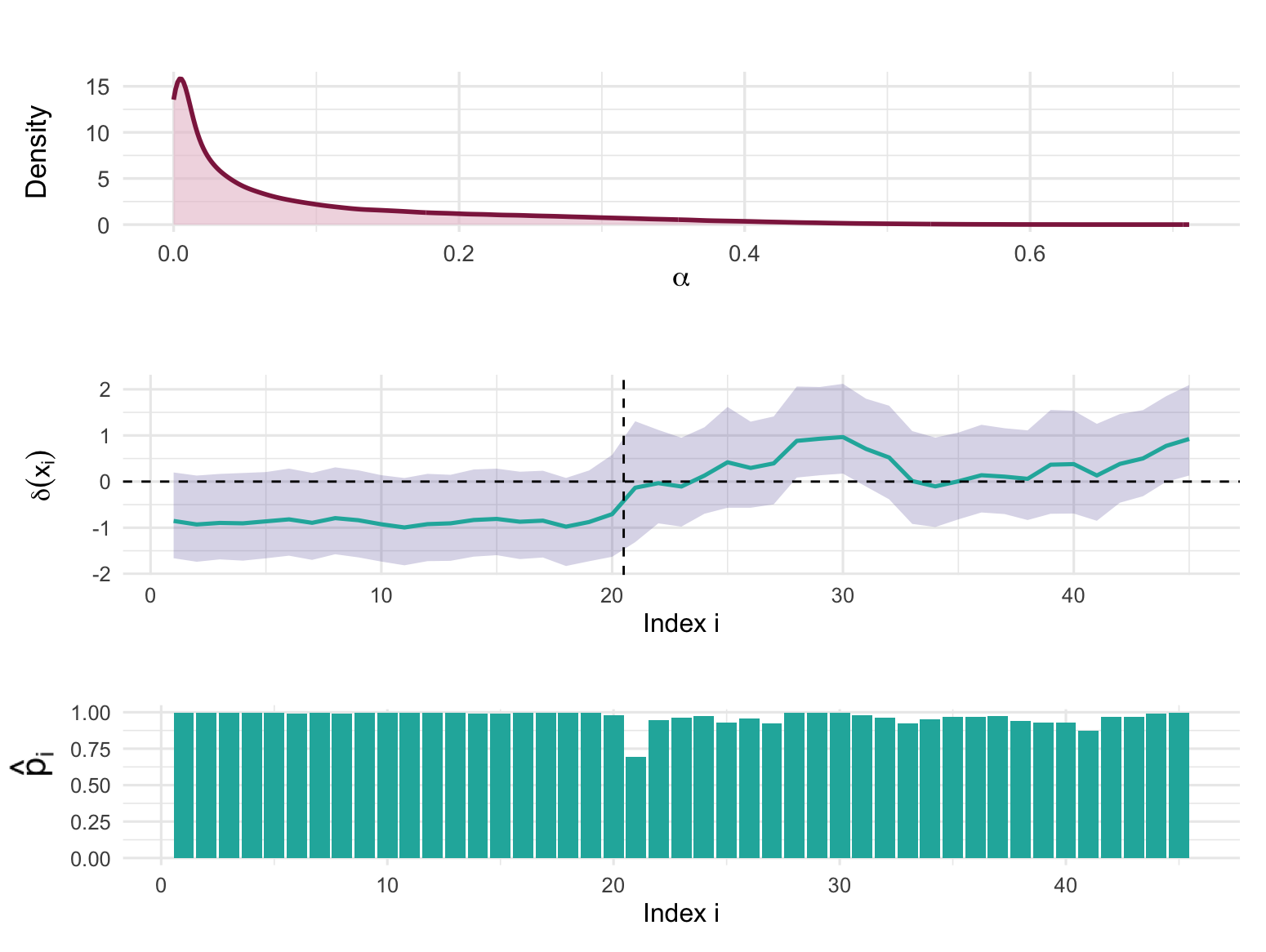}\includegraphics[width=0.5\textwidth]{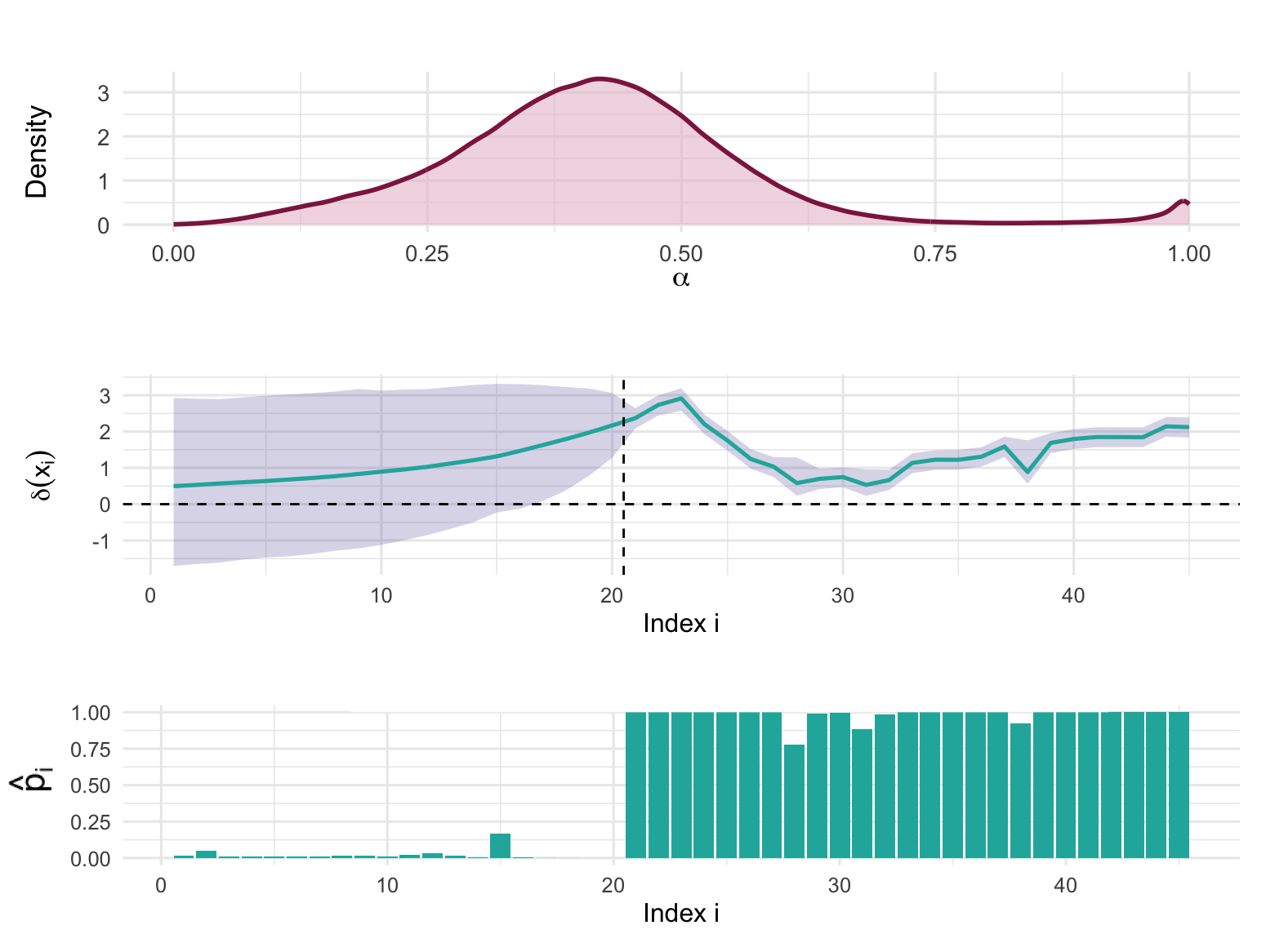}
        \caption{\small Bayesian inference : (Left panel) without thresholding and (Right panel) with thresholding, when $g_e$ is fixed during MCMC implementation but $h_0,\delta,\lambda^2,\alpha$ are estimated. (Top) Pooled posterior density of $\alpha$. (Middle) Posterior means of $\delta(x_i)$ with $95\%$ credible bands plotted for one representative dataset. The vertical dashed line marks the true activation boundary at $i=n_0+0.5$. (Bottom) Pointwise posterior inclusion probabilities $\hat p_i$.}
        \label{fig:seuil_senario2_4}
\end{figure}
Figure~\ref{fig:seuil_senario2_4} (left panel) shows that fixing $g_e$ alone is insufficient to produce a selective local diagnostic. The middle plots of the left panel illustrate the posterior mean of $\delta(x_i)$ (solid line) lying between approximately $-1.5$ and $2$ across the whole input domain, with $95\%$ pointwise credible bands of width roughly $\pm 2$. Note that, in this configuration, the discrepancy is present \emph{only} on the active region $i > n_0$ (vertical dashed line); on the inactive region $i \le n_0$, the true discrepancy is exactly zero. Despite this, the posterior of $\delta(x_i)$ remains markedly positive on the inactive region as well, which is a first indication that the unthresholded mixture allocation does not localize the discrepancy. The credible band shows a mild contraction as $i$ enters the active region: as we move into $i > n_0$, the observations $y_i$ start carrying direct information about $\delta(x_i)$, and the GP correlation lets the posterior at neighboring sites borrow strength from this informative region, so the band tightens. 
For a representative replicated dataset, we examine pointwise posterior ``inclusion" probabilities that we define as $\hat p_i = \P\bigl(\zeta_i = 1 \mid \pmb{y}, X, \pmb{\theta}, \delta, \lambda^2, k, \gamma_\delta, \alpha\bigr); i=1,\dots,n$. The inclusion probabilities are not available in closed-form and are therefore estimated from the MCMC outputs produced by Algorithm~\ref{algo:mwg} (illustrated on the bottom of Figure~\ref{fig:seuil_senario2_4}) using
\begin{align*}
    \hat p_i 
    \;=\;\frac{1}{T}\sum_{t=1}^{T}\mathds{I}_{\{\zeta_i^{(t)}=1\}}.
\end{align*}

where $T$ is the number of the MCMC iterations retained after discarding the burn-in period.
Averaging over the retained iterations therefore yields a sample-based estimate of the posterior probability that observation $i$ comes from the discrepancy-corrected component. In the bottom plot of 
In Figure~\ref{fig:seuil_senario2_4} (bottom-left panel), the pointwise inclusion probabilities $\hat p_i$ are almost equal to $1$ across the entire domain, including the discrepancy-inactive region, and the global mixture weight $\alpha$ concentrates near zero, so the mixture globally favors the discrepancy-corrected component. These results show that eliminating the $g_e$--$\delta$ confounding alone does not enable accurate localization of the discrepancy: the unthresholded mixture allocation continues to favor the discrepancy-corrected model everywhere, including where no discrepancy is present. This highlights the need for a threshold-based local diagnostic.

We therefore applied the thresholded allocations described in Section~\ref{sec:threshold} and Figure~\ref{fig:seuil_senario2_4} (right-panel) displays that the pointwise inclusion probabilities $\hat p_i$ stay close to zero before the true change-point and jump sharply to values close to one afterwards, in agreement with the data-generating design. The posterior mean of $\delta(x_i)$ increases substantially after the change-point, with much tighter credible bands in the active region. The posterior of $\alpha$ itself is centred at an intermediate value rather than at an extreme, confirming that $\alpha$ is best interpreted as a global summary of the dataset while local interpretation is carried by the $\hat p_i$'s. For the same simulated datasets, we also considered two other cases where $g_e$ is unknown and estimated and Algorithm \ref{algo:mwg} is implemented twice with and without thresholded allocations. 
Without thresholding, the standard mixture model tends to use the discrepancy-corrected component almost everywhere once it is globally preferred. While applying thresholded allocations leaves the pointwise allocation pattern weakly identified. Comparison with the case where $g_e=g_e^\ast$, this lack of identifiability is genuine confounding between the physical calibration and the discrepancy, not a failure of the thresholding device: once this confounding is reduced, the thresholded allocation map localizes the active region very sharply. Details and figures are reported in the appendix.

\subsection{Application to real datasets}
We now apply the mixture framework to the real ball-drop data, focusing on the  \emph{Blue Basketball} experiment. Results for all eleven balls are provided in the 
\href{https://github.com/negar-soleimani/estimate_mixture_models}{code repository}.

The normalized time domain is used for a stable inference; the initial height $h_0$ is inferred from the data, and, in the results shown below, we fix $g_e$ at its nominal value $9.8\,\mathrm{m/s^2}$ to reduce the well-known confounding between the physical calibration and the discrepancy term (Results with $g_e$ left free are qualitatively similar and are provided in the appendix). We compare two configurations: the Bayesian mixture model with the standard (unthresholded) allocation mechanism defined in Section \ref{sec-meth}, and the thresholded mixture of Section~\ref{sec:threshold} with $s=0.3$.
The two posterior summaries are displayed side by side in Figure~\ref{fig:real_blue_basketball}: panel~(a) shows the unthresholded fit, panel~(b) the thresholded one.
\begin{figure}[h!]
    \centering
    \begin{subfigure}[t]{0.49\textwidth}
        \centering
        \includegraphics[width=\linewidth]{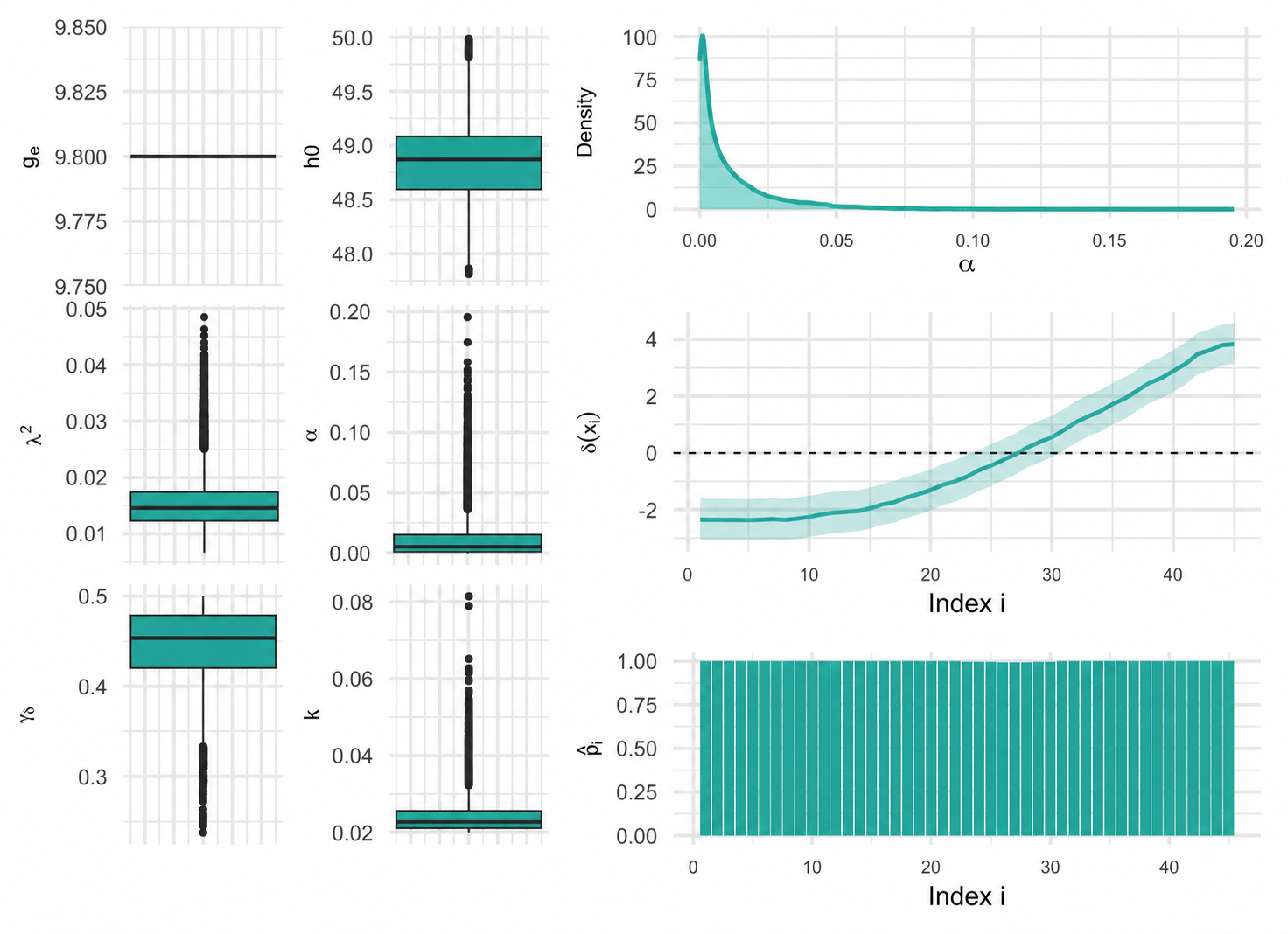}
        \caption{Unthresholded mixture.}
        \label{fig:real_blue_basketball_a}
    \end{subfigure}\hfill
    \begin{subfigure}[t]{0.49\textwidth}
        \centering
        \includegraphics[width=\linewidth]{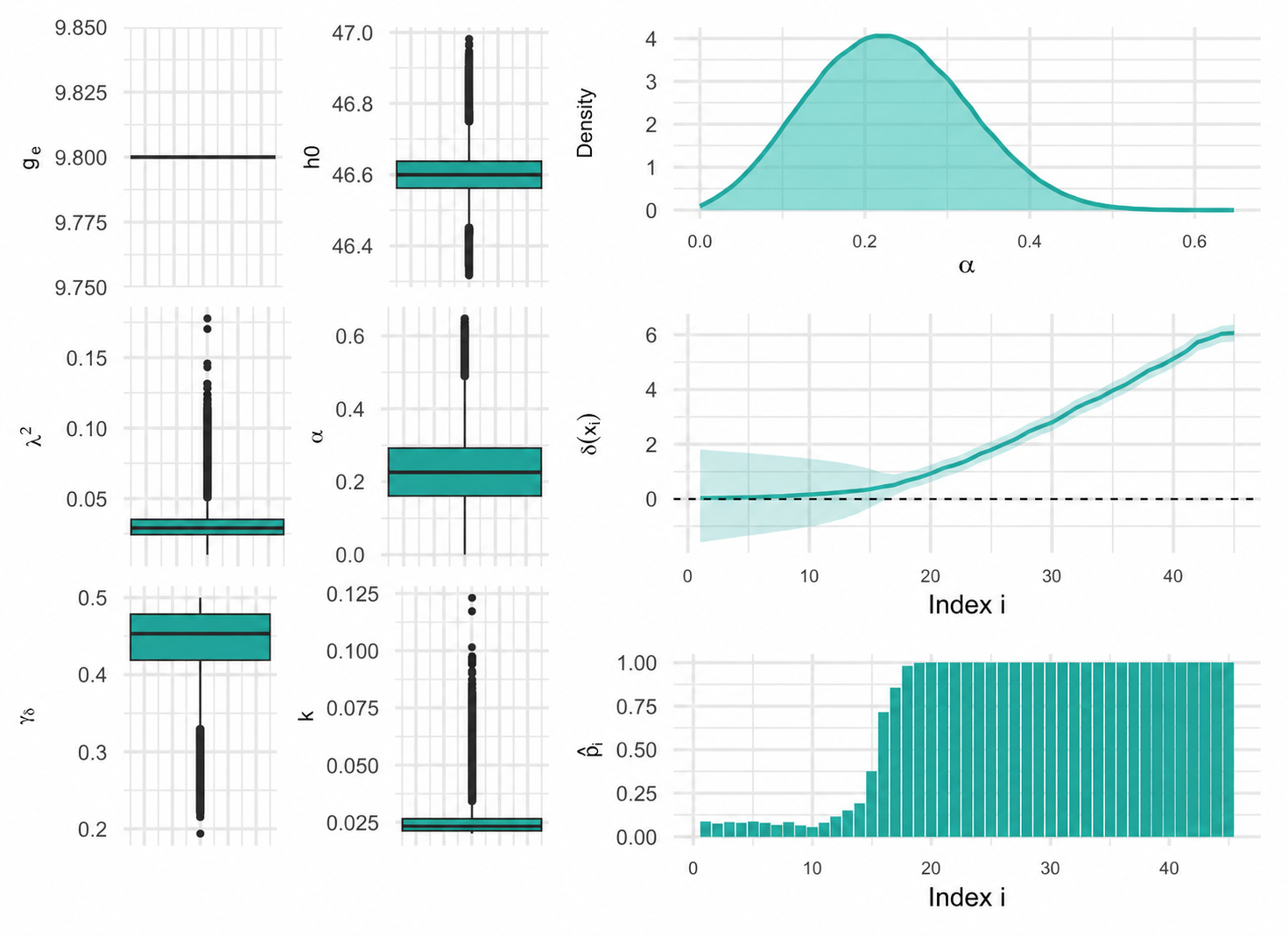}
        \caption{Thresholded mixture, $s=0.3$.}
        \label{fig:real_blue_basketball_b}
    \end{subfigure}
    \caption{\small Real data application (Blue Basketball). Posterior summaries when  $g_e^\ast = 9.8$ : (Left-Boxplots) Posterior distribution of the MCMC draws; (Top-Right-panel) Posterior density of $\alpha$ samples; (Middle-Right-panel) Posterior mean (solid-line) of $\delta(x_i)$ with pointwise $95\%$ credible bands; (Bottom-Right-panel) Pointwise posterior inclusion probabilities $\hat p_i$.}
    \label{fig:real_blue_basketball}
\end{figure}
Figure~\ref{fig:real_blue_basketball_a} presents the posterior estimations of the mixture model parameters when the gravitational constant is fixed at $g_e^\ast=9.8$ and the intercept parameter $h_0$ is estimated. The posterior density of $\alpha$ is sharply concentrated near zero, indicating that the mixture model strongly favours the discrepancy-corrected component. This conclusion is also visible from the pointwise inclusion probabilities: $\hat p_i$ equal to one for all observation indices treats the discrepancy as active over the whole trajectory. 
The marginal posterior distribution of $h_0$ is concentrated around values close to $49$, which represents an overestimation relative to the nominal value suggested by the prior information. This behavior is explained by the joint evolution of $h_0$ and the discrepancy. Since $g_e$ is fixed, the model compensates a larger intercept by estimating a discrepancy that is negative at the beginning of the trajectory and positive toward the end. Hence, the data mainly identify the combined term $h_0+\delta(t)$, rather than $h_0$ and $\delta(t)$ separately.
The hyperparameters still reflect a regime in which the discrepancy explains most of the systematic deviation from the idealized trajectory. The posterior of $\gamma_\delta$ is concentrated toward $0.5$, corresponding to a smooth discrepancy process, which is consistent with the monotone shape of the estimated $\delta(x_i)$. The posterior draws of $k$ are concentrated near small values, while $\lambda^2$ is also close to zero with a few larger draws. This indicates that the model attributes most of the structured deviation to the GP discrepancy rather than to independent measurement noise. This behavior is coherent with the narrow credible bands for $\delta(x_i)$, but it also confirms the importance of monitoring the interaction between $\lambda^2$, $k$, and the discrepancy variance. Overall, when the mixture allocation is not thresholded, the discrepancy-corrected component can dominate globally and may absorb part of the information that would otherwise contribute to estimating $h_0$. We next apply the thresholded allocation mechanism of Section~\ref{sec:threshold} with $s=0.3$, so that the discrepancy-corrected component can be selected only where the current draw of $|\delta(x_i)|$ is large enough to be practically distinguishable from measurement noise. The results are shown in Figure~\ref{fig:real_blue_basketball_b}. The posterior mean of $\delta(t_i)$ is close to zero at the beginning of the trajectory and then increases smoothly, reaching larger positive values toward the end of the experiment. This behaviour is consistent with the physical interpretation of the data: 
initially, the free-fall model provides an accurate description of the motion. As the ball accelerates, the unmodeled drag force becomes increasingly significant, resulting in a departure from the model. The thresholded allocation changes the interpretation of the mixture. The posterior of $\alpha$ is no longer concentrated near zero, but is centred at an intermediate value, indicating that both components contribute to explaining the data. The pointwise allocation probabilities, $\hat p_i$, provide the intended local diagnostic: the early part of the trajectory is mainly assigned to the no-discrepancy component, while after a transition around $i \simeq 16$-$18$, the discrepancy-corrected component is selected almost surely. Thus, thresholding localizes where the discrepancy becomes necessary, instead of assigning the whole trajectory to the discrepancy-corrected model.
This local reallocation also explains the changes in the marginal posterior distributions of $h_0$, $\lambda^2$, and $k$. Under thresholding, small early discrepancies are no longer allowed to compensate for an inflated intercept. As a result, the posterior of $h_0$ moves down from the over-estimated value obtained in Figure~\ref{fig:real_blue_basketball_a} and becomes concentrated near $46.6$, much closer to the nominal initial height. At the same time, part of the small-scale variation previously absorbed by the always-active discrepancy is now attributed to the residual error and to the active part of the Gaussian-process discrepancy, leading to larger posterior values of $\lambda^2$ and a corresponding shift in $k$. These differences are therefore a direct consequence of the thresholded allocation, which separates negligible deviations from practically relevant discrepancy.

\section{Conclusion}\label{sec-conc}
This paper deals with computer code validation through the mixture-estimation approach of \citet{KKKMCPRJR2014}, in the setting where the code is linear in the calibration parameters (or can be reasonably approximated by a linear surrogate). The pure-code and the discrepancy-corrected models are embedded into an encompassing mixture, and validation is based on the posterior of the mixture weight $\alpha$; calibration and validation are carried out simultaneously.
On the methodological side, we showed that the parameter-sharing structure of the mixture allows noninformative priors on the shared parameters without jeopardizing posterior propriety, and we developed a Metropolis-within-Gibbs algorithm for the posterior inference of the mixture model parameters. On the empirical side, the simulation experiments show that when the data are generated from $\MF_0$ the posterior of $\alpha$ concentrates near 1, while under $\MF_1$ it concentrates near 0 except at very small correlation lengths, for which the GP discrepancy becomes indistinguishable from white noise and the decision becomes less sharp.
We also introduced a thresholded allocation mechanism that complements the global decision carried by $\alpha$ with a local diagnostic indicating where the discrepancy is effectively needed. The simulations show that this device is especially useful when the discrepancy is active only on part of the input domain, and that the quality of the local diagnostic depends on the confounding between the physical calibration and the discrepancy-once this confounding is reduced, the thresholded allocation map becomes much sharper. The real ball-drop application confirms the practical value of the approach: at the global level, $\alpha$ indicates that the nominal free-fall code is insufficient; at the local level, thresholding mechanism identifies precisely where the physical model remains adequate and where an additional correction is needed. Additional results, including an orthogonal Gaussian process construction aimed at further reducing confounding, are reported in the appendix. A natural continuation of this work would be to extend the method to richer nonlinear surrogate models.

%
%
%

\paragraph{Generative AI statement}
OpenAI's ChatGPT (GPT-5.6 Luna, 2026) was used for language editing and assistance with R code. The authors reviewed and validated all AI-assisted content.

\paragraph{Acknowledgements} 
Negar Soleimani acknowledges support from a scholarship awarded by the Ministère de l’Europe et des Affaires étrangères (MEAE), France.

  \bibliography{bibliography.bib}

  \newpage
  
\appendix{
\section*{Appendix}
\addcontentsline{toc}{section}{Appendix}

This appendix collects material that complements the main paper. Section~\ref{sup:propriety}
proves posterior propriety under the Jeffreys prior on the shared parameters
(Theorem~\ref{thm:propriety}). Section~\ref{sup:jeffreys} derives the Jeffreys
prior for the Gaussian linear regression component used throughout the paper.
Section~\ref{sup:conditionals} details the full conditional posterior
distributions of $\boldsymbol{\theta}$, $\lambda^2$, $k$ and $\alpha$ used in
the Metropolis-within-Gibbs sampler. Section~\ref{sup:ogp} gives the orthogonal
Gaussian process (OGP) construction together with the closed-form integrals
needed for its implementation on the linearized ball-drop model.
Section~\ref{sup:sim} reports additional simulation results,
namely (\ref{sup:scnI_III}) presents the other configurations that are referenced in Section 4.1 
of the main paper, (\ref{sup:ogp_sim}) inference under
the OGP prior when data are simulated from a classical GP, and
(\ref{sup:oracle}) an oracle identifiability diagnostic for the discrepancy hyper-parameters $(k,\gamma_\delta)$. Section~\ref{sup:real} reports the real-data results when the gravitational constant $g_e$ is left free, referenced
in Section~4.2 of the main paper. Finally, Section~\ref{sup:cgp-ogp-real} compares the classical and orthogonal GP priors on the real dataset, showing that the orthogonality constraint does not resolve the calibration--discrepancy confounding in this application and thus justifying the use of the classical GP prior in the main analysis.


\section{Propriety of the joint posterior distribution}\label{sup:propriety}
Consider the mixture model $\MF_\alpha$, for all $i\in\{1,\ldots,n\}$:
\begin{equation}\label{eq:1a}
\MF_\alpha: y_i\sim \alpha \ell_{\MF_0}(\pmb{\theta}, \lambda^2;y_i, x_i)+(1-\alpha)\ell_{\MF_1}(\pmb{\theta}, \lambda^2, \delta;y_i, x_i),
\end{equation}
introduced in Section~2 of the main paper, in which each observation is either assigned to the pure-code component $\MF_0$ or to the discrepancy-corrected component $\MF_1$. 

\begin{theorem}\label{thm:propriety}
Let $g:\mathcal{X}\to\mathbb{R}^{d}; d\geq1$ and $X=(x_1,\ldots,x_n)$ be a design matrix with distinct design points. Assume that the design matrix elements are scaled to the unit interval and $g(X)$ is full column rank $d$ with $n>d$. Under the weakly informative prior specification below
\begin{align*}
    \pi(\boldsymbol{\theta},\lambda^2)&\propto (\lambda^2)^{-(\nicefrac{d}{2}+1)}; \quad \alpha\sim\mathcal{B}(a_0,a_0); \quad \delta(\bx)\sim\mathcal{N}_n(0,\Sigma_\delta)~\text{with}~ \Sigma_\delta=\nicefrac{\lambda^2}{k}\text{Corr}_{\gamma_\delta}
\end{align*}
then the posterior distribution of the mixture model $\MF_\alpha$ is proper if 
\begin{itemize}
    \item[I.] $\Sigma_\delta$ has a standard stationary kernel and $\text{Corr}_{\gamma_\delta}$ is continuous with respect to
$\gamma_\delta$;
    \item[II.] $\gamma_\delta$ and $k$ are supported on compact sets and are assigned proper prior distributions;
    \item[III.] prior on $k$ satisfies $\int k^{n/2}\pi(k)\,\mathrm{d}k<\infty.$
\end{itemize}
\end{theorem}

\begin{proof}
For $Y=(y_1, y_2, \ldots, y_n)$, the marginal likelihood $m_{\MF_{\alpha}}(Y|X)$ is defined as 
\begin{align*}
    \int\left[\prod_{i=1}^n
    \left(
    \alpha \, \ell_{\mathcal{M}_0}(y_i \mid x_i, \theta, \lambda^2)
    +
    (1-\alpha)\, \ell_{\mathcal{M}_1}(y_i \mid x_i, \theta, \lambda^2, \delta)
    \right)
    \right]
    \pi(\pmb{\theta}, \lambda^2, \alpha, \delta, k, \gamma_\delta)
    \, \mathrm{d}\Theta,
\end{align*}

where
$\mathrm{d}\Theta = \mathrm{d}(\theta, \lambda^2, \alpha, \delta, k, \gamma_\delta)$. 
\small \begin{equation}\label{eq:18}
\int\prod_{i=1}^n\left[\alpha \nicefrac{\exp\left(-\frac{(y_i-g(x_i)\pmb{\theta})^2}{2\lambda^2} \right)}{(2\pi\lambda^2)^{1/2}}+ (1-\alpha)\nicefrac{\exp\left(-\frac{(y_i-g(x_i)\pmb{\theta}-\delta(x_i))^2}{2\lambda^2}\right)}{(2\pi\lambda^2)^{1/2}}\right]\pi(\pmb{\theta},\lambda,\alpha,\delta,k,\gamma_\delta)\mathrm{d}\Theta
\end{equation}
\normalsize
The marginal likelihood is a high-dimensional integral of a product-of-mixtures Gaussian kernel under an improper prior, and posterior propriety reduces to showing this integral is finite.
Let $\mathcal{S}=\{1, 2, \ldots, n \}$ and define, for $\nu = 0,1,\ldots,n$, $$\mathcal{A}^{(\nu)}=\{S\subset \mathcal{S} : |S|=\nu\}$$ the collection of all subsets of $\mathcal{S}$ of cardinality $\nu$, with
$|\mathcal{A}^{(\nu)}| = \binom{n}{\nu}$.
~If $\Omega=\{\mathcal{A}^{(1)}, \mathcal{A}^{(2)}, \ldots, \mathcal{A}^{(n-1)} \}$, then from $2^n=\sum_{\nu=0}^n\binom n\nu$, the total number of non-trivial allocations is therefore $|\Omega|=\sum_{\nu=1}^{n-1}|\mathcal{A}^{(\nu)}|=2^n-2$ ($\nu=0$ and $\nu=n$ correspond to pure components).
We also denote by $\mathcal{A}^{(\nu)}_j$, the $j-$th element of $\mathcal{A}^{(\nu)}$ 
and so the  $n-$fold binomial product in \eqref{eq:18} is then given by
\begin{align}\label{eq:20}
&\alpha^n \nicefrac{\exp\left(-\frac{1}{2\lambda^2}\left[\sum_{i=1}^n(y_i-g(x_i)\pmb{\theta})^2\right] \right)}{(2\pi\lambda^2)^{n/2}}\nonumber\\
&+\sum_{l=1}^{n-1}\alpha^{n-l}(1-\alpha)^l\sum_{j=1}^{\binom nl}\nicefrac{\exp\left(-\frac{1}{2\lambda^2}\left[\sum_{i\in \mathcal{S} \setminus\mathcal{A}^{(l)}_j}(y_i-g(x_i)\pmb{\theta})^2+\sum_{i\in \mathcal{A}^{(l)}_j}(y_i-g(x_i)\pmb{\theta}-\delta(x_i))^2\right] \right)}{(2\pi\lambda^2)^{n/2}}\nonumber\\
&+(1-\alpha)^n \nicefrac{\exp\left(-\frac{1}{2\lambda^2}\left[\sum_{i=1}^n(y_i-g(x_i)\pmb{\theta}-\delta(x_i))^2\right] \right)}{(2\pi\lambda^2)^{n/2}}.
\end{align}
\normalsize

The marginal likelihood is then computed by marginalizing out the parameters $\alpha, \pmb{\theta}, \lambda, \delta, \gamma_\delta$ and $k$ from \eqref{eq:20} starting first with integrating out the parameter $\delta$. 
The integral with respect to $\delta$ of the first term in \eqref{eq:20}
~$\int_{\delta}\pi(\delta)\mathrm{d}\delta=1$ since the prior is proper. For the second term, $l=|\mathcal{A}^{(l)}_j|$ 
~and denote by $Y_{\mathcal{A}^{(l)}_j}$: an $l-$components column vector of $y_i$ for all $i \in \mathcal{A}^{(l)}_j$. We also suppose that $g_{ \mathcal{A}^{(l)}_j}$ and $\delta_{ \mathcal{A}^{(l)}_j}$ indicate an $l\times d$ matrix of $g(x_i)$s and $l-$components column vector of $\delta(x_i)$s obtained for all $i \in \mathcal{A}^{(l)}_j$.  Since $\delta$ follows a Gaussian process centered on $n-$components vector of $0$ with covariance matrix $\Sigma_\delta=\nicefrac{\lambda^2}{k}\text{Corr}_{\gamma_\delta}$, then the GP restricted to finite indices $\{ i \in \mathcal{A}^{(l)}_j\}$ gives a multivariate normal distribution with the mean vector: $l-$components column vector of $0$ and $l\times l$ covariance matrix $\Sigma_{\mathcal{A}^{(l)}_j}=\nicefrac{\lambda^2}{k}\text{Corr}_{\gamma_\delta;\mathcal{A}^{(l)}_j}$ related to the corresponding submatrix.

~Using these notations, we can write 
\small $$\sum_{i\in \mathcal{A}^{(l)}_j}(y_i-g(x_i)\pmb{\theta}-\delta(x_i))^2=(Y_{\mathcal{A}^{(l)}_j}-g_{ \mathcal{A}^{(l)}_j}\pmb{\theta}-\delta_{ \mathcal{A}^{(l)}_j})^T(Y_{\mathcal{A}^{(l)}_j}-g_{\mathcal{A}^{(l)}_j}\pmb{\theta}-\delta_{\mathcal{A}^{(l)}_j})$$
\normalsize
when for any matrix $M$, $M^T$ indicates the transpose of $M$. If \small $\sum_{i=1}^n(y_i-g(x_i)\pmb{\theta}-\delta(x_i))^2=(Y-g\pmb{\theta}-\delta)^T(Y-g\pmb{\theta}-\delta)$ \normalsize where $g=g(X)$, $\delta=\delta(X)$ for simplicity and combining these expressions with the Gaussian prior distribution of $\delta$,
the two remaining terms in \eqref{eq:20} can be integrated with respect to
$\delta$ as follows:
\begin{equation}\label{eq:23}
\int_{\delta_{\mathcal{A}^{(l)}_j}}\left(\frac{|\text{Corr}_{\gamma_\delta;\mathcal{A}^{(l)}_j}|^{-1/2}}{(\nicefrac{2\pi \lambda^2}{k})^{\nicefrac{l}{2}}}\right)\exp\left(-\frac{1}{2\lambda^2}\left[
(\delta_{\mathcal{A}^{(l)}_j}-\pmb{\mu}_{j,l}^{\delta})^T\pmb{Q}_{j,l}(\delta_{\mathcal{A}^{(l)}_j}-\pmb{\mu}_{j,l}^{\delta})+\pmb{R}_{(j,l)}(\theta)
\right] \right)\mathrm{d}\delta_{\mathcal{A}^{(l)}_j}
\end{equation}
\normalsize
and
\begin{equation}\label{eq:24}
\int_{\delta}\left(\frac{|\text{Corr}_{\gamma_\delta}|^{-1/2}}{(\nicefrac{2\pi \lambda^2}{k})^{n/2}}\right)\exp\left(-\frac{1}{2\lambda^2}\left[
(\delta-\pmb{\mu}_{n}^{\delta})^T\pmb{Q}_n(\delta-\pmb{\mu}_{n}^{\delta})+\pmb{R}_n(\theta)
\right] \right)\mathrm{d}\delta
\end{equation}
\normalsize
where
\small \begin{align}\label{const:1}
\pmb{\mu}_{j,l}^{\delta}&=\pmb{Q}_{j,l}^{-1}(Y_{\mathcal{A}^{(l)}_j}-g_{ \mathcal{A}^{(l)}_j}\pmb{\theta})\nonumber\\
\pmb{Q}_{j,l}&=\mathds{I}_{l}+k\text{Corr}_{\gamma_\delta;\mathcal{A}^{(l)}_j}^{-1}\nonumber\\
\pmb{R}_{(j,l)}(\pmb{\theta})&=(Y_{\mathcal{A}^{(l)}_j}-g_{ \mathcal{A}^{(l)}_j}\pmb{\theta})^T
(\mathds{I}_{l}-\pmb{Q}_{j,l}^{-1})
(Y_{\mathcal{A}^{(l)}_j}-g_{ \mathcal{A}^{(l)}_j}\pmb{\theta})\nonumber\\
\pmb{\mu}_n^{\delta}&=\pmb{Q}_n^{-1}(Y-g\pmb{\theta})\nonumber\\
\pmb{Q}_n&=\mathds{I}_n+k\text{Corr}_{\gamma_\delta}^{-1}\nonumber\\
\pmb{R}_n(\pmb{\theta})&=(Y-g\pmb{\theta})^T(Y-g\pmb{\theta})-(Y-g\pmb{\theta})^T\pmb{Q}_n^{-1}(Y-g\pmb{\theta})
\end{align}
\normalsize
and $\mathds{I}_{l}$ indicates an identity matrix of size $l$. 
In order to solve the integral in \eqref{eq:23} and \eqref{eq:24}, we must prove that $\pmb{Q}_{j,l}$ and $\pmb{Q}_n$ are positive definite. For every non-zero column vector $\nu$ of $n$ real numbers, if we define $w=\text{Corr}_{\gamma_\delta;\mathcal{A}^{(l)}_j} \nu$, then 
$$w^T\text{Corr}_{\gamma_\delta;\mathcal{A}^{(l)}_j}^{-1}w=\nu^T\text{Corr}_{\gamma_\delta;\mathcal{A}^{(l)}_j}^T\nu.$$
Since the correlation matrix 
$\text{Corr}_{\gamma_\delta;\mathcal{A}^{(l)}_j}$ is symmetric and positive definite ($\text{Corr}_{\gamma_\delta;\mathcal{A}^{(l)}_j}^T=\text{Corr}_{\gamma_\delta;\mathcal{A}^{(l)}_j}$) then $w^T\text{Corr}_{\gamma_\delta;\mathcal{A}^{(l)}_j}^{-1}w=\nu^T\text{Corr}_{\gamma_\delta;\mathcal{A}^{(l)}_j}\nu>0$. 
Beside that, since $\nu^T\mathds{I}_{l}\nu=\sum_{i=1}^{l}\nu_i^2>0$
~and $k\in (0, 1)$ then
$\nu^T(\mathds{I}_{l}+k \text{Corr}_{\gamma_\delta;\mathcal{A}^{(l)}_j}^{-1})\nu=\nu^T\mathds{I}_{l}\nu+\nu^Tk\text{Corr}_{\gamma_\delta;\mathcal{A}^{(l)}_j}^{-1}\nu>0 $. 
We conclude then the matrix $\pmb{Q}_{j,l}$ and so $\pmb{Q}_n$ are both symmetric and positive definite. The two integrals  are then 
\begin{align}\label{eq:25}
\eqref{eq:23}
&=
k^{l/2}
|\mathrm{Corr}_{\gamma_\delta;\mathcal A_j^{(l)}}|^{-1/2}
|\pmb{Q}_{j,l}|^{-1/2}
\exp\left(-\frac{\pmb{R}_{(j,l)}(\pmb{\theta})}{2\lambda^2}\right),
\end{align}
\normalsize
and
\begin{align}\label{eq:26}
\eqref{eq:24}
&=
k^{n/2}
|\mathrm{Corr}_{\gamma_\delta}|^{-1/2}
|\pmb{Q}_n|^{-1/2}
\exp\left(-\frac{\pmb{R}_n(\theta)}{2\lambda^2}\right).
\end{align}
\normalsize
By replacing \eqref{eq:25} and \eqref{eq:26} in $m_{\MF_{\alpha}}(Y|X)$, we obtain
\small \begin{align}\label{eq:28}
m_{\MF_{\alpha}}(Y|X)&=\int\Bigg\{\alpha^n (2\pi\lambda^2)^{-n/2}\exp\left(-\frac{1}{2\lambda^2}\left[\sum_{i=1}^n(y_i-g(x_i)\pmb{\theta})^2\right] \right)\nonumber\\
&+\sum_{l=1}^{n-1}\alpha^{n-l}(1-\alpha)^l\sum_{j=1}^{\binom nl}(2\pi\lambda^2)^{-n/2}|\text{Corr}_{\gamma_\delta;\mathcal{A}^{(l)}_j}|^{-1/2}|\pmb{Q}_{j,l}|^{-\nicefrac{1}{2}}k^{\nicefrac{l}{2}}\nonumber\\
&\times \exp\left(-\frac{1}{2\lambda^2}\left[\sum_{i\in \mathcal{S}\setminus\mathcal{A}^{(l)}_j}(y_i-g(x_i)\pmb{\theta})^2+\pmb{R}_{(j,l)}(\pmb{\theta})\right] \right)\nonumber\\
&+(1-\alpha)^n(2\pi\lambda^2)^{-n/2}|\text{Corr}_{\gamma_\delta}|^{-1/2}|\pmb{Q}_n|^{-\nicefrac{1}{2}}k^{\nicefrac{n}{2}}\nonumber\\
&\times\exp\left(-\frac{1}{2\lambda^2}\left[\pmb{R}_n(\pmb{\theta})\right] \right)\Bigg\}\pi(\pmb{\theta}, \lambda^2, \alpha, k,\gamma_\delta)\mathrm{d}(\pmb{\theta}, \lambda^2, \alpha, k,\gamma_\delta)
\end{align}
\normalsize
Since $\pi(\pmb{\theta}, \lambda^2)\,\propto\,  (\lambda^2)^{-(d/2+1)}$  and focusing on integrating out $\pmb{\theta}$, then by completing the square in $\pmb{\theta}$, the integral of the first term in \eqref{eq:28} with respect to $\pmb{\theta}$ can be written as 
\begin{align}\label{eq:add}
&\int_{\pmb{\theta}} \exp\left(-\frac{1}{2\lambda^2}\left[(\pmb{\theta}-\pmb{\mu}_{(1,n)}^{\pmb{\theta}})^T(g^Tg)(\pmb{\theta}-\pmb{\mu}_{(1,n)}^{\pmb{\theta}})-\text{Const}\right] \right)(2\pi\lambda^2)^{-n/2}(\lambda^2)^{-(d/2+1)}\mathrm{d}\pmb{\theta}\nonumber\\
&=(2\pi)^{\nicefrac{(d-n)}{2}} (\lambda^2)^{-(n/2+1)}|g^Tg|^{-\nicefrac{1}{2}}\exp(-\frac{\text{Const}}{2\lambda^2}). 
\end{align}
where
\begin{align}\label{cons1n}
\pmb{\mu}_{(1,n)}^{\pmb{\theta}}&=(g^Tg)^{-1}g^TY\nonumber\\
\text{Const}&=Y^TY-Y^Tg(g^Tg)^{-1}g^TY
\end{align}
and the integration is with respect to a Gaussian kernel with precision $g^Tg/\lambda^2$. 

For the second integral with respect to $\pmb{\theta}$ in \eqref{eq:28}, we note that 
\small\begin{align}\label{QuadThet}
 \sum_{i\in \mathcal{S}\setminus\mathcal{A}^{(l)}_j}(y_i-g(x_i)\pmb{\theta})^2+\pmb{R}_{(j,l)}(\pmb{\theta})&=(Y_{\mathcal{S}\setminus\mathcal{A}^{(l)}_j}-g_{ \mathcal{S}\setminus\mathcal{A}^{(l)}_j}\pmb{\theta})^T(Y_{\mathcal{S}\setminus\mathcal{A}^{(l)}_j}-g_{\mathcal{S}\setminus\mathcal{A}^{(l)}_j}\pmb{\theta})\nonumber\\
&+(Y_{\mathcal A_j^{(l)}}-
g_{\mathcal A_j^{(l)}}\pmb\theta)^T
A_{j,l}
(Y_{\mathcal A_j^{(l)}}-
g_{\mathcal A_j^{(l)}}\pmb\theta).
\end{align} \normalsize
where $
A_{j,l}=\mathds{I}_l-\pmb{Q}_{j,l}^{-1}.$ The function $A_{j,l}$  is symmetric since both $\pmb{Q}_{j,l}$ and  $\pmb{Q}_{j,l}^{-1}$ are symmetric positive definite. However, in order to check the positive definiteness of $A_{j,l}$, we must verify whether the related eigenvalues are all positive. Because the correlation matrix $\text{Corr}_{\mathcal{A}^{(l)}_j}$ is positive definite, then there exists an orthogonal matrix $U$ of size $l\times l$ (i.e. $UU^T=\mathds{I}_{l}$) such that $\text{Corr}_{\gamma_\delta;\mathcal{A}^{(l)}_j}=U\Psi U^T$ when $\Psi$ is a diagonal matrix whose diagonal elements are the eigenvalues of the correlation matrix. This definition leads to rewrite the term $\text{Corr}_{\gamma_\delta;\mathcal{A}^{(l)}_j}^{-1}$ as follows
\begin{align}\label{ccccr}
\text{Corr}_{\gamma_\delta;\mathcal{A}^{(l)}_j}^{-1}
&=U\Psi^{-1} U^T
 \end{align}
where the inverse of the diagonal matrix $\Psi=\left[ \psi_{ii}\right]_{i=1}^{l}$ is obtained by replacing each element in the diagonal with its reciprocal as $\Psi^{-1}=\left[ \nicefrac{1}{\psi_{ii}}\right]_{i=1}^{l}$. On the other hands, $\mathds{I}_{l}=U\mathds{I}_{l} U^T$ and from \eqref{ccccr}, we then obtain
\begin{align*}
\mathds{I}_{l}-\pmb{Q}_{j,l}^{-1}&=\mathds{I}_{l}-\left(\mathds{I}_{l}+k\text{Corr}_{\gamma_\delta;\mathcal{A}^{(l)}_j}^{-1}\right)^{-1}\\
&=U\mathds{I}_{l} U^T-(U^T)^{-1}\left(\mathds{I}_{l} +\left[ \nicefrac{k}{ \psi_{ii}}\right]_{i=1}^{l} \right)^{-1}U^{-1}\\
&=U\mathds{I}_{l} U^T-U \left(\left[ 1+\nicefrac{k}{ \psi_{ii}}\right]_{i=1}^{l} \right)^{-1}U^T\\
&=U \left[ \nicefrac{k}{(\psi_{ii}+k) }\right]_{i=1}^{l}U^T.
 \end{align*}
where $k\in (0, 1)$, $\psi_{ii}>0$ and consequently $\mathds{I}_{l}-\pmb{Q}_{j,l}^{-1}$ has the positive eigenvalues and $A_{j,l}$ is positive definite.  Therefore, expanding both quadratic forms in \eqref{QuadThet} yields
\begin{align}\label{S14}
 \sum_{i\in \mathcal{S}\setminus\mathcal{A}^{(l)}_j}(y_i-g(x_i)\pmb{\theta})^2+\pmb{R}_{(j,l)}(\pmb{\theta})=
(\pmb\theta-\pmb\mu_{j,l}^{\pmb\theta})^T
M_{j,l}(k,\gamma_\delta)
(\pmb\theta-\pmb\mu_{j,l}^{\pmb\theta})
+C_{j,l}(k,\gamma_\delta)
\end{align}
where
\begin{align*}
\pmb\mu_{j,l}^{\pmb\theta}=M_{j,l}^{-1}&(k,\gamma_\delta)b_{j,l}; \quad \quad M_{j,l}(k,\gamma_\delta)=g_{\mathcal S\setminus\mathcal A_j^{(l)}}^Tg_{\mathcal S\setminus\mathcal A_j^{(l)}}
+g_{\mathcal A_j^{(l)}}^T
A_{j,l}g_{\mathcal A_j^{(l)}}\\
C_{j,l}(k,\gamma_\delta)&=Y_{\mathcal S\setminus\mathcal A_j^{(l)}}^TY_{\mathcal S\setminus\mathcal A_j^{(l)}}+Y_{\mathcal A_j^{(l)}}^TA_{j,l}Y_{\mathcal A_j^{(l)}}
-b_{j,l}^TM_{j,l}^{-1}(k,\gamma_\delta)b_{j,l}\\
&b_{j,l}=g_{\mathcal S\setminus\mathcal A_j^{(l)}}^TY_{\mathcal S\setminus\mathcal A_j^{(l)}}
+g_{\mathcal A_j^{(l)}}^TA_{j,l}Y_{\mathcal A_j^{(l)}}.
\end{align*}
To establish the positive definiteness of the matrix $M_{j,l}(k,\gamma_\delta)$, which is required for the Gaussian integral with respect to $\pmb{\theta}$ to be finite, let $v\in\mathbb{R}^{d}$ be arbitrary. Then
\begin{align}
v^{T}M_{j,l}(k,\gamma_\delta)v
&=v^{T}g_{\mathcal{S}\setminus\mathcal{A}^{(l)}_j}^{T}g_{\mathcal{S}\setminus\mathcal{A}^{(l)}_j}v
+v^{T}g_{\mathcal{A}^{(l)}_j}^{T}
\left(
\mathds{I}_{l}-\pmb{Q}_{j,l}^{-1}
\right)g_{\mathcal{A}^{(l)}_j}v\nonumber\\
&=\left\|
g_{\mathcal{S}\setminus\mathcal{A}^{(l)}_j}v
\right\|^{2}
+\left(g_{\mathcal{A}^{(l)}_j}v
\right)^{T}
\left(\mathds{I}_{l}-\pmb{Q}_{j,l}^{-1}
\right)
\left(g_{\mathcal{A}^{(l)}_j}v
\right).
\label{eq:Mjlquadratic}
\end{align}

Both terms on the right-hand side of \eqref{eq:Mjlquadratic} are nonnegative, hence $v^{T}M_{j,l}(k,\gamma_\delta)v \ge 0.$
Suppose now that $v^{T}M_{j,l}(k,\gamma_\delta)v=0.$ Since each term in \eqref{eq:Mjlquadratic} is nonnegative, both terms
must vanish separately; i.e. 
$
g_{\mathcal{S}\setminus\mathcal{A}^{(l)}_j}v=0;  g_{\mathcal{A}^{(l)}_j}v=0.
$
Combining these two relations gives
$$
gv=
\begin{pmatrix}
g_{\mathcal{S}\setminus\mathcal{A}^{(l)}_j}
\\[1mm]
g_{\mathcal{A}^{(l)}_j}
\end{pmatrix}
v=0.
$$

Since the full design matrix $g$ has full column rank $d$, its null space is trivial and consequently,
$gv=0$ implies
$v=0$, i.e., $\forall\, v\neq 0, v^{T}M_{j,l}(k,\gamma_\delta)v>0$. On the other hands, since $\mathds{I}_l-\pmb{Q}_{j,l}^{-1}$ is symmetric, so is $\left(g_{\mathcal{A}^{(l)}_j}v\right)^{T}\left(\mathds{I}_{l}-\pmb{Q}_{j,l}^{-1}\right)\left(g_{\mathcal{A}^{(l)}_j}v\right)$. Therefore both terms in the definition of $M_{j,l}(k,\gamma_\delta)$ are symmetric, so is $M_{j,l}(k,\gamma_\delta)$ and consequently $M_{j,l}(k,\gamma_\delta)$ is positive definite matrix. 
As a consequence, $M_{j,l}(k,\gamma_\delta)$ is invertible and all its eigenvalues are
strictly positive. On the other hands, $C_{j,l}(k,\gamma_\delta)$ can not be negative since the left-hand side of \eqref{S14} (sum of quadratic forms with $A_{j,l}>0$) is indeed nonnegative and $M_{j,l}(k,\gamma_\delta)>0$ on the right-hand side of \eqref{S14}. Furthermore, the equality $C_{j,l}(k,\gamma_\delta)=0$ holds if and only if there exists some parameter vector $\pmb\theta$ satisfying simultaneously 
$$
Y_{\mathcal S\setminus\mathcal A_j^{(l)}}=g_{\mathcal S\setminus\mathcal A_j^{(l)}}\pmb\theta; \quad \text{and} \quad A_{j,l}^{1/2}\left(Y_{\mathcal A_j^{(l)}}-g_{\mathcal A_j^{(l)}}\pmb\theta\right)=0.
$$
In other words, the data must fit the regression model exactly (taking into account the weighting by $A_{j,l}$). This event, however, has probability zero for continuous observation distribution we assumed and consequently, $C_{j,l}(k,\gamma_\delta)>0$ almost surely.
The Gaussian integral with respect to
$\pmb{\theta}$ is hence proper and evaluates to
\small\begin{align}\label{eq:30}
\int_{\mathbb{R}^{d}}
\exp\!\left(
-\frac{1}{2\lambda^{2}}
(\pmb{\theta}-\pmb{\mu}_{j,l})^{T}
M_{j,l}(k,\gamma_\delta)
(\pmb{\theta}-\pmb{\mu}_{j,l})+C_{j,l}(k,\gamma_\delta)
\right)
d\pmb{\theta}
&=\nonumber\\
(2\pi\lambda^{2})^{d/2}
|M_{j,l}(k,\gamma_\delta)|^{-1/2} &\exp\!\left(
-\frac{C_{j,l}(k,\gamma_\delta)}{2\lambda^2}
\right).
\end{align}\normalsize

The last integral with respect to $\pmb{\theta}$ in \eqref{eq:28} is given by
 \small \begin{align}\label{eq:31}
\int_{\pmb{\theta}}&\exp\left(-\frac{1}{2\lambda^2}\left[\pmb{R}_n(\theta)\right] \right)\mathrm{d}\pmb{\theta}\nonumber\\
&=\int_{\pmb{\theta}}\exp\left(-\frac{1}{2\lambda^2}\left[(\pmb{\theta}-\pmb{\mu}_{(2,n)}^{\pmb{\theta}})^Tg^T\left(\mathds{I}_n-\pmb{Q}_n^{-1} \right)g(\pmb{\theta}-\pmb{\mu}_{(2,n)}^{\pmb{\theta}})+F_n(k,\gamma_\delta)\right] \right)\mathrm{d}\pmb{\theta}\nonumber\\
&=(2\pi \lambda^2)^{\nicefrac{d}{2}}|g^T\left(\mathds{I}_n-\pmb{Q}_n^{-1} \right)g|^{-\nicefrac{1}{2}}\exp(-\nicefrac{F_n(k,\gamma_\delta)}{2\lambda^2})
\end{align}
\normalsize
where
\small \begin{align}
\pmb{\mu}_{(2,n)}^{\pmb{\theta}}&=\left[g^T\left(\mathds{I}_n-\pmb{Q}_n^{-1} \right)g\right]^{-1}g^T\left(\mathds{I}_n-\pmb{Q}_n^{-1} \right)Y\nonumber\\
F_n(k,\gamma_\delta)&=Y^T\left(\mathds{I}_n-\pmb{Q}_n^{-1} \right)Y-Y^T\left(\mathds{I}_n-\pmb{Q}_n^{-1} \right)g\left[g^T\left(\mathds{I}_n-\pmb{Q}_n^{-1} \right)g\right]^{-1}g^T\left(\mathds{I}_n-\pmb{Q}_n^{-1} \right)Y
\end{align}\normalsize
Because $\text{Corr}_{\gamma_\delta}^{-1}$ is symmetric positive definite, we can easily prove that 
 $g^T\left(\mathds{I}_n-\pmb{Q}_n^{-1} \right)g=g^T\left(\mathds{I}_n-\left(\mathds{I}_n+k\text{Corr}_{\gamma_\delta}^{-1} \right)^{-1} \right)g$
 is also a symmetric positive definite matrix (same proof as for $g^T_{\mathcal{A}^{(l)}_j}(\mathds{I}_{l}-\pmb{Q}_{j,l}^{-1})g_{\mathcal{A}^{(l)}_j}$).
By replacing \eqref{eq:add}, \eqref{eq:30} and \eqref{eq:31} in the marginal likelihood \eqref{eq:28} and focusing on integrating out $\lambda$, we obtain
\small \begin{align}\label{eq:32}
m_{\MF_{\alpha}}&(Y|X)= \int \alpha^n (2\pi )^{\nicefrac{(d-n)}{2}}|g^Tg|^{-\nicefrac{1}{2}}(\lambda^2)^{-(\nicefrac{n}{2}+1)}\exp(-\nicefrac{\text{Const}}{2\lambda^2})\pi(\alpha, k)\mathrm{d}(\lambda^2, \alpha, k)\nonumber\\
&+\int \sum_{l=1}^{n-1}\alpha^{n-l}(1-\alpha)^l\sum_{j=1}^{\binom nl}(2\pi)^{(d-n)/2}|\text{Corr}_{\gamma_\delta;\mathcal{A}^{(l)}_j}|^{-1/2}|\pmb{Q}_{j,l}|^{-\nicefrac{1}{2}}k^{\nicefrac{l}{2}}\nonumber\\
&\times 
|M_{j,l}(k,\gamma_\delta)|^{-1/2}(\lambda^2)^{-(\nicefrac{n}{2}+1)} \exp\!\left(-\frac{C_{j,l}(k,\gamma_\delta)}{2\lambda^2}
\right)\pi(\alpha, k,\gamma_\delta)\mathrm{d}(\lambda^2,\alpha, k,\gamma_\delta)\nonumber\\
&+\int(1-\alpha)^n(2\pi)^{\nicefrac{(d-n)}{2}}|\text{Corr}_{\gamma_\delta}|^{-1/2}|\pmb{Q}_n|^{-\nicefrac{1}{2}}k^{\nicefrac{n}{2}}|g^T\left(\mathds{I}_n-\pmb{Q}_n^{-1} \right)g|^{-\nicefrac{1}{2}}\nonumber\\
&\times(\lambda^2)^{-(\nicefrac{n}{2}+1)}\exp(-\nicefrac{F_n(k,\gamma_\delta)}{2\lambda^2})\pi(\alpha, k,\gamma_\delta)\mathrm{d}(\lambda^2,\alpha, k,\gamma_\delta)
\end{align}
\normalsize

The right-hand side of the inequality \eqref{eq:32} is integrable with respect to $\lambda^2$ provided that the quantities $\text{Const}$ and $F_n(k,\gamma_\delta)$ are positive.   The term $\text{Const}$ defined in \eqref{cons1n} can be rewritten as 
\begin{align}
\text{Const}&=Y^T\left[\mathds{I}_n-\mathbf{H}_g\right]Y
\end{align}
in which $\mathbf{H}_g=g(g^Tg)^{-1}g^T$ is the hat matrix associated with the linear model, i.e. the orthogonal projection matrix onto the column space of $g$. In particular, $\hat{Y}=\mathbf{H}_gY$ denotes the vector of fitted values \citep{DCHREW1978}. 
The matrix $\mathbf{H}_g$ satisfies the following classical properties: First, $\mathbf{H}_g$ is symmetric and idempotent,i.e. $ \mathbf{H}_g^T=\mathbf{H}_g,
    ~\mathbf{H}_g^2=\mathbf{H}_g$ 
    ; Second, since $g$ is a full rank matrix,  the matrix $\mathbf{H}_g$ is positive semidefinite and admits the spectral decomposition $\mathbf{H}_g=U_{\mathbf{H}_g}\Psi_{\mathbf{H}_g}U^T_{\mathbf{H}_g}$; And third, the eigenvalues of $\mathbf{H}_g$ denoted by $\psi^{\mathbf{H}_g}_{ii}=\left[\Psi_{\mathbf{H}_g}\right]_{ii}$ are either $0$ or $1$. Using this decomposition, we obtain       
\begin{align}\label{cons1n2}
\text{Const}&=Y^TU_{\mathbf{H}_g}\left[\mathds{I}_n-
\Psi_{\mathbf{H}_g}
\right]U^T_{\mathbf{H}_g}Y\nonumber\\
&=Y^TU_{\mathbf{H}_g}\operatorname{diag}\!\left(1-\psi_{11}^{\mathbf{H}_g},\ldots,1-\psi_{nn}^{\mathbf{H}_g}\right)U^T_{\mathbf{H}_g}Y\nonumber\\
&=\sum_{i=1}^n \left((1-\psi^{\mathbf{H}_g}_{ii})(\sum_{j=1}^n y_j u_{ji})^2\right)
\end{align}
where $u_{ji}$ denotes the $(j,i)$-th entry of $U_{\mathbf{H}_g}$. Since each coefficient $1-\psi_{ii}^{\mathbf{H}_g}$ is either $0$ or $1$, equation \eqref{cons1n2} shows that
$\text{Const}$ is nonnegative.
Moreover, with
$\text{Const}=Y^T(Y-\hat Y),$
 equality $\text{Const}=0$ holds if and only if $Y=\hat Y$, that is, if $Y$ belongs exactly to the column space of $g$.
Because the observations are generated from a continuous model, the event
$
Y\in\mathrm{Col}(g)
$
has probability zero. Therefore,
$
\mathbb P\!\left(\text{Const}>0\right)=1,
$
so that $\text{Const}$ is almost surely strictly positive.

For the quantity $F_n(k,\gamma_\delta)$, define $\Upsilon_n:=\mathds{I}_n-\pmb{Q}_n^{-1}$ to simplify the notation. Since
$
\pmb{Q}_n=\mathds{I}_n+k\,\mathrm{Corr}_{\gamma_\delta}^{-1},
$
with $k>0$ and $\mathrm{Corr}_{\gamma_\delta}^{-1}$ positive definite, it follows that $\pmb{Q}_n> \mathds{I}_n$, and $
\Upsilon_n=\mathds{I}_n-\pmb{Q}_n^{-1}> 0.
$
Hence, there exists a unique symmetric positive definite square root
$\Upsilon_n^{1/2}=U_{\Upsilon}\Psi_{\Upsilon_n}^{1/2}
U_{\Upsilon}^T$ such that $\Upsilon_n=\Upsilon_n^{1/2}\Upsilon_n^{1/2},$
when
$\Upsilon_n=U_{\Upsilon}\Psi_{\Upsilon_n}U_{\Upsilon}^T$
is the spectral decomposition of $\Upsilon_n$, with
$U_{\Upsilon}U_{\Upsilon}^T=U_{\Upsilon}^TU_{\Upsilon}=\mathds{I}_n$
and
$\Psi_{\Upsilon_n}=\operatorname{diag}\!\left(
\psi^{\Upsilon_n}_{11},\ldots,\psi^{\Upsilon_n}_{nn}
\right).$
The function $F_n(k,\gamma_\delta)$ can be therefore rewritten as
\begin{align}
F_n(k,\gamma_\delta)
&=
Y^T\Upsilon_n^{1/2}
\Bigg[
\mathds{I}_n
-
\Upsilon_n^{1/2}g
\left(
g^T\Upsilon_n g
\right)^{-1}
g^T\Upsilon_n^{1/2}
\Bigg]
\Upsilon_n^{1/2}Y.
\end{align}

If we define
\begin{align*}
    \mathbf P_{\Upsilon}
    =
    \Upsilon_n^{1/2}g
    \left(
    g^T\Upsilon_n g
    \right)^{-1}
    g^T\Upsilon_n^{1/2}.
\end{align*}

then, the matrix $\mathbf P_{\Upsilon}$ is the orthogonal projection matrix onto the column space of $\Upsilon_n^{1/2}g$. Therefore,
$
\mathds{I}_n-\mathbf P_{\Upsilon}
$
is symmetric positive semidefinite, implying that
$
F_n(k,\gamma_\delta)\ge 0.
$

Moreover, equality holds if and only if
$
\Upsilon_n^{1/2}Y\in\mathrm{Col}(\Upsilon_n^{1/2}g),
$
which is equivalent to $Y\in\mathrm{Col}(g)$ since $\Upsilon_n^{1/2}$ is invertible.
Because the observations arise from a continuous model, the event
$
Y\in\mathrm{Col}(g)
$
has probability zero. Consequently,
$
\mathbb P\!\left(F_n(k,\gamma_\delta)>0\right)=1.
$
Hence, $F_n(k,\gamma_\delta)$ is almost surely strictly positive. 

The integrals with respect to $\lambda$ are then, respectively, equal to  
$$
\frac{2^{(\nicefrac{n}{2})}\Gamma(\nicefrac{n}{2})}{\text{Const}^{\nicefrac{n}{2}}}; \quad \frac{2^{(\nicefrac{n}{2})}\Gamma(\nicefrac{n}{2})}{C_{j,l}(k,\gamma_\delta)^{\nicefrac{n}{2}}}; \quad
\frac{2^{(\nicefrac{n}{2})}\Gamma(\nicefrac{n}{2})}{F_n(k,\gamma_\delta)^{\nicefrac{n}{2}}}.
  $$
Since $\alpha\sim\mathcal{B}(a_0,a_0)$ with $a_0>0$, the integrals with
respect to $\alpha$ appearing in \eqref{eq:32} reduce to Beta-function ratios $$
\frac{\left[\Gamma(n+a_0)\Gamma(2a_0)\right]}{[\Gamma(a_0)\Gamma(n+2a_0)]}; \quad  
\frac{[\Gamma(n-l+a_0)\Gamma(l+a_0)\Gamma(2a_0)]}{[\Gamma(n+2a_0)\Gamma(a_0)^2]};\quad  \frac{[\Gamma(n+a_0)\Gamma(2a_0)]}{[\Gamma(a_0)\Gamma(n+2a_0)]}$$
These quantities are finite for all $a_0>0$, since the Gamma function is well-defined and finite on $(0,\infty)$. Moreover, from the Beta-function identity we can easily show that each ratio is uniformly bounded by $1$. This leads us to simplify the inequality \eqref{eq:32} as follows
\small \begin{align}\label{eq:33}
m_{\MF_{\alpha}}(Y|X)&\leq \int \frac{2^{(\nicefrac{d}{2})}\Gamma(\nicefrac{n}{2})}{\text{Const}^{\nicefrac{n}{2}}} \pi ^{\nicefrac{(d-n)}{2}}|g^Tg|^{-\nicefrac{1}{2}}\pi( k,\gamma_\delta)\mathrm{d}(k,\gamma_\delta)\nonumber\\%
&+\int \sum_{l=1}^{n-1}\sum_{j=1}^{\binom nl}\frac{\Gamma(\nicefrac{n}{2})|\text{Corr}_{\gamma_\delta;\mathcal{A}^{(l)}_j}|^{-1/2}|\pmb{Q}_{j,l}|^{-\nicefrac{1}{2}}k^{\nicefrac{l}{2}}}{C_{j,l}(k,\gamma_\delta)^{\nicefrac{n}{2}}}\pi^{\nicefrac{(d-n)}{2}}|g^T_{\mathcal{S}\setminus\mathcal{A}^{(l)}_j}g_{\mathcal{S}\setminus\mathcal{A}^{(l)}_j}|^{-\nicefrac{1}{4}}\nonumber\\
&\times |g^T_{\mathcal{A}^{(l)}_j}\left(\mathds{I}_{l}-\pmb{Q}_{j,l}^{-1}\right)g_{\mathcal{A}^{(l)}_j}|^{\nicefrac{-1}{4}}\pi(k,\gamma_\delta)\mathrm{d}(k,\gamma_\delta)\nonumber\\%
+\int\frac{2^{(\nicefrac{d}{2})}\Gamma(\nicefrac{n}{2})}{F_n(k,\gamma_\delta)^{\nicefrac{n}{2}}}&|\text{Corr}_{\gamma_\delta}|^{-1/2}|\pmb{Q}_n|^{-\nicefrac{1}{2}}k^{\nicefrac{n}{2}}|g^T\left(\mathds{I}_n-\pmb{Q}_n^{-1} \right)g|^{-\nicefrac{1}{2}}\pi^{\nicefrac{(d-n)}{2}}
\pi(k,\gamma_\delta)\mathrm{d}(k,\gamma_\delta)
\end{align}
\normalsize
We note here that in the right hand-side of \eqref{eq:33}, having proper priors on $k$ and $\gamma_\delta$ is necessary but not generally sufficient to conclude that it is finite. For instance, when $\gamma_\delta \to 0$, then $\text{Corr}_{\gamma_\delta}$ becomes nearly singular and $|\text{Corr}_{\gamma_\delta}|^{-1/2}\to \infty$ (the same occurs for $|\text{Corr}_{\gamma_\delta;\mathcal{A}_j^{(l)}}|^{-1/2}, |\pmb{Q}_n|^{-1/2}, |\pmb{Q}_{j,l}|^{-1/2}$). To avoid this degeneracy, we assume that the parameter $\gamma_\delta$ belongs to a compact set $[\underline{\gamma}, \bar{\gamma}]; 0<\underline{\gamma}< \bar{\gamma}\leq 1$ to ensure that the correlation matrices remain uniformly positive definite,
determinants stay bounded away from zero and all inverse determinants remain finite. 

Regarding the integrals with respect to $k$, the polynomial terms in $k$ such as $k^{l/2}$ and $k^{n/2}$ must be bounded to ensure that every integrand is bounded by a finite constant, i.e. the prior on $k$ must have finite moment of order $n/2$. On the other hands, when $k\to 0$ then $\pmb{Q}_n\to \mathds{I}_n$ and $|g^T\left(\mathds{I}_n-\pmb{Q}_n^{-1} \right)g|^{-\nicefrac{1}{2}}\to \infty$ (same for the term with $\mathds{I}_l-\pmb{Q}_{j,l}$ in the second integral). Thus, the second and third integrands in \eqref{eq:33} behaves like a ratio of quantities both tending to zero/infinity. Since, $|g^T\left(\mathds{I}_n-\pmb{Q}_n^{-1} \right)g|^{-\nicefrac{1}{2}}$ and $F_n(k,\gamma_\delta)^{-n/2}$ are order $k^{-d/2}$ and $k^{-n/2}$, respectively, and the integral also contains $k^{n/2}$, therefore $|g^T\left(\mathds{I}_n-\pmb{Q}_n^{-1} \right)g|^{-\nicefrac{1}{2}}F_n(k,\gamma_\delta)^{-n/2}k^{n/2}$ behaves like $k^{-d/2}$. So the integral is finite if $d/2<1$ and diverges at $k=0$ for all $d\geq 2$. Furthermore, the discrepancy prior parametrization $\delta(X)\sim \mathcal{GP}(0,\frac{\lambda^2}{k}\text{Corr}_{\gamma_\delta})$ becomes degenerate/infinite variance when $k=0$. 
Hence, we consider $k\in [\underline{k},1]$ to ensure no singularity for all $d>1$. 

By compactness and assigning proper prior distributions to $k$ and $\gamma_\delta$, every integrand appearing in \eqref{eq:33}
is consequently bounded by a finite constant times a proper density.
This proves that
\begin{align*}
    m_{\MF_\alpha}(Y\mid X)<\infty.
\end{align*}

 and the posterior distribution of the mixture model parameters is therefore proper. 
\end{proof}

\section{Jeffreys prior for the shared parameters $(\pmb{\theta},\lambda^2)$}\label{sup:jeffreys}
For the following Gaussian linear regression model
\begin{equation}\label{eq:supp-lm}
Y = g(X)\boldsymbol{\theta} + \boldsymbol{\varepsilon},
\qquad
\boldsymbol{\varepsilon}\sim\mathcal{N}\big(\mathbf{0},\lambda^2\mathds{I}_n\big),
\end{equation}
with $\boldsymbol{\theta}\in\mathbb{R}^d$ and $\lambda^2>0$
 the log-likelihood is
\begin{align*}
    \ell(\boldsymbol{\theta},\lambda^2) =
    -\tfrac{n}{2}\log(2\pi) - \tfrac{n}{2}\log(\lambda^2)
    -\tfrac{1}{2\lambda^2}(Y-g(X)\boldsymbol{\theta})^\top(Y-g(X)\boldsymbol{\theta}).
\end{align*}

For simplicity, we consider $\mathbf{r}=Y-g(X)\boldsymbol{\theta}$ and after differentiating,
\begin{align*}
    \frac{\partial^2\ell}{\partial\boldsymbol{\theta}\,\partial\boldsymbol{\theta}^\top}
    =-\frac{1}{\lambda^2}g(X)^\top g(X),
    \qquad
    \frac{\partial^2\ell}{\partial(\lambda^2)^2}
    =\frac{n}{2\lambda^4}-\frac{1}{\lambda^6}\mathbf{r}^\top\mathbf{r},
    \qquad
    \frac{\partial^2\ell}{\partial\boldsymbol{\theta}\,\partial\lambda^2}
    =-\frac{1}{\lambda^4}g(X)^\top\mathbf{r}.
\end{align*}

the Fisher
information matrix is block-diagonal,
\begin{equation}\label{eq:supp-fisher}
I(\boldsymbol{\theta},\lambda^2)=
\begin{pmatrix}
\displaystyle\frac{1}{\lambda^2}g(X)^\top g(X) & \mathbf{0}_{d\times 1}\\[4pt]
\mathbf{0}_{1\times d} & \displaystyle\frac{n}{2\lambda^4}
\end{pmatrix}\in\mathbb{R}^{(d+1)\times(d+1)}.
\end{equation}
Consequently,
\begin{align*}
    \det I(\boldsymbol{\theta},\lambda^2)
    = \det\!\Big(\tfrac{1}{\lambda^2}g(X)^\top g(X)\Big)\cdot\frac{n}{2\lambda^4}
    = \det(g(X)^\top g(X))\cdot\frac{n}{2}(\lambda^2)^{-(d+2)},
\end{align*}

and the Jeffreys prior is then
\begin{equation}\label{eq:supp-jeffreys}
\pi(\boldsymbol{\theta},\lambda^2)\;\propto\;
\frac{1}{(\lambda^2)^{d/2+1}},
\end{equation}
which is an improper-flat prior on $\boldsymbol{\theta}$ and an improper scale-invariant distribution on $\lambda^2$~\citep{OJBVDBS2001}.
In the ball-drop application with the two dimensional code 
$g(x)=(1,-x^2/2)$, we have
$\pi(\boldsymbol{\theta},\lambda^2)\propto 1/(\lambda^2)^{2}$.
We note here that when one of the calibration parameters is fixed in a given scenario, the Jeffreys prior is applied to the remaining free regression parameters.

\section{Full conditional posterior distributions}\label{sup:conditionals}

Given the latent allocations $\boldsymbol{\zeta}=(\zeta_1,\ldots,\zeta_n)^\top$,
the completed-data likelihood function is given as
\begin{equation}\label{eq:supp-compl-lik}
\mathcal{L}(\boldsymbol{\theta},\delta,\boldsymbol{\zeta},\alpha,\lambda^2; Y, X)
=\alpha^{n-\sum_{i=1}^n\zeta_i}(1-\alpha)^{\sum_{i=1}^n\zeta_i}
\prod_{i:\zeta_i=0}\mathcal{N}\!\big(y_i\mid g(x_i)\boldsymbol{\theta},\lambda^2\big)
\prod_{i:\zeta_i=1}\mathcal{N}\!\big(y_i\mid g(x_i)\boldsymbol{\theta}+\delta(x_i),\lambda^2\big).
\end{equation}
Let $m=\sum_{i=1}^{n}\zeta_i$ and denote by $\pmb{x}_m,\pmb{y}_m,\delta_m$ the restrictions of $X,Y,\delta(X)$ to $\{i:\zeta_i=1\}$, and by $\pmb{x}_{n-m},\pmb{y}_{n-m}$ the restrictions to $\{i:\zeta_i=0\}$. We derive in turn the full conditionals of $\boldsymbol{\theta}$, $\lambda^2$, $k$ and $\alpha$. The conditional distribution of $\delta$ is the standard Gaussian-process conditioning formula reported in the main paper, and the conditional of $\gamma_\delta$ is updated by a Metropolis--Hastings step.

\subsection{Full conditional of $\boldsymbol{\theta}$}\label{sup:cond-theta}

Under the prior $\pi(\boldsymbol{\theta}, \lambda^2)\propto \frac{1}{(\lambda^2)^{d/2+1}}$, the conditional posterior of $\boldsymbol{\theta}$ is proportional to,
\begin{align*}
    p(\boldsymbol{\theta}\mid\cdot)\;\propto\;
    \exp\!\left\{-\tfrac{1}{2\lambda^2}\sum_{\zeta_i=0}(y_i-g(x_i)\boldsymbol{\theta})^{2}\right\}
    \exp\!\left\{-\tfrac{1}{2\lambda^2}\sum_{\zeta_i=1}(y_i-g(x_i)\boldsymbol{\theta}-\delta(x_i))^{2}\right\}.
\end{align*}

Equivalently, the likelihood contribution for $\boldsymbol{\theta}$ can be written as
\begin{align*}\small
p(\boldsymbol{\theta}\mid\cdot)\;&\propto\;
\exp\!\left\{
-\frac{1}{2\lambda^2}
\left[
\left(\pmb{y}_{n-m}-g(\pmb{x}_{n-m})\boldsymbol{\theta}\right)^{\top}
\left(\pmb{y}_{n-m}-g(\pmb{x}_{n-m})\boldsymbol{\theta}\right)\right]\right\}\\
&\times
\exp\!\left\{
-\frac{1}{2\lambda^2}\left[
\left(\pmb{y}_m-\delta_m-g(\pmb{x}_m)\boldsymbol{\theta}\right)^{\top}
\left(\pmb{y}_m-\delta_m-g(\pmb{x}_m)\boldsymbol{\theta}\right)
\right]\right\}.
\end{align*}
Completing the square in $\boldsymbol{\theta}$ gives
\begin{align*}
    \small
    -\tfrac{1}{2\lambda^2}\left(\boldsymbol{\theta}^\top\!\Big[\underbrace{\big(g(\pmb{x}_{n-m})^\top g(\pmb{x}_{n-m})+g(\pmb{x}_m)^\top g(\pmb{x}_m)\big)}_{\displaystyle A}\Big]\boldsymbol{\theta}
    -2 \boldsymbol{\theta}^\top\!\Big[\underbrace{\big(g(\pmb{x}_{n-m})^\top \pmb{y}_{n-m}+g(\pmb{x}_m)^\top(\pmb{y}_m-\delta_m)\big)}_{\displaystyle B}\Big]\right).
\end{align*}

Hence the full conditional is multivariate Gaussian,
\begin{equation}\label{eq:supp-theta-full}
\boldsymbol{\theta}\mid Y,\delta,\lambda^2,\boldsymbol{\zeta}
\;\sim\;
\mathcal{N}_d\!\big(\hat{\boldsymbol{\mu}}_{\boldsymbol{\theta}},\hat{\Sigma}_{\boldsymbol{\theta}}\big), \quad \text{with}\quad \hat{\boldsymbol{\mu}}_{\boldsymbol{\theta}} = A^{-1}B,
\qquad
\hat{\Sigma}_{\boldsymbol{\theta}} = \lambda^2\,A^{-1}.
\end{equation}

In the ball-drop application, depending on the inferential scenario, one or both
components of $\boldsymbol{\theta}=(h_0,g_e)^\top$ may be fixed. In such cases,
the full conditional must be derived directly for the free parameter, rather
than sampling from the joint posterior and subsequently overwriting the fixed
component.

The computer model
$f(t,\boldsymbol{\theta})$ defined in the main paper in which $g_e$ is the gravitational parameter.
Writing the second column of the design as $c_i := -\tfrac{1}{2}t_i^{2}$, so that
the computer model reads $f(t_i,\boldsymbol{\theta}) = h_0 + g_e\,c_i$, the
scenario-specific updates are obtained as follows.

\paragraph{Case 1: $g_e$ fixed.}
If $g_e=g_{e}^{fix}$, only $h_0$ is updated ($d=1$) and
\begin{align*}
h_0 \mid . &\;\propto\;
\alpha^{n-m}(2\pi\lambda^2)^{-(n-m)/2}
\exp\left\{
-\frac{1}{2\lambda^2}\displaystyle\sum_{i:\zeta_i=0}(y_i-h_0-c_i g_{e}^{fix})^2
\right\}\\
&\times(1-\alpha)^{m}(2\pi\lambda^2)^{-m/2}
\exp\left\{
-\frac{1}{2\lambda^2}\displaystyle\sum_{i:\zeta_i=1}(y_i-h_0-c_i g_{e}^{fix}-\delta_i)^2
\right\}
(\lambda^2)^{-3/2}\\
\propto
\exp&\left\{
-\frac{1}{2\lambda^2}
\left[
(n-m) h_0^2
-2h_0\sum_{i:\zeta_i=0}(y_i-c_i g_{e}^{fix})
+
m h_0^2
-2h_0\sum_{i:\zeta_i=1}(y_i-c_i g_{e}^{fix}-\delta_i)
\right]
\right\}\\
&\propto
\exp\left(
-\frac{n}{2\lambda^2}
\left(
h_0-
\frac{1}{n}
\left(
\displaystyle
\sum_{i:\zeta_i=0}(y_i-c_i g_{e}^{fix})
+
\sum_{i:\zeta_i=1}(y_i-c_i g_{e}^{fix}-\delta_i)
\right)
\right)^2
\right)
\end{align*}
The conditional posterior distribution of $h_0|.$ is therefore Gaussian $\mathcal{N}(\hat{\mu}_{h_0}, \widehat{\mathrm{Var}}_{h_0})$ with

\begin{align*}
    \hat{\mu}_{h_0}
    =
    \frac{1}{n}
    \left(
    \displaystyle
    \sum_{i:\zeta_i=0}(y_i-c_i g_{e}^{fix})
    +
    \sum_{i:\zeta_i=1}(y_i-c_i g_{e}^{fix}-\delta_i)
    \right)
    \qquad \text{and}\qquad
    \displaystyle
    \widehat{\mathrm{Var}}_{h_0}
    =
    \frac{\lambda^2}{n}
\end{align*}

\paragraph{Case 2: $h_0$ fixed.}
\begin{align*}
g_{e} \mid . &\;\propto\;
\alpha^{(n-m)}(2\pi\lambda^2)^{-(n-m)/2}
\exp\left(
-\frac{1}{2\lambda^2}\displaystyle\sum_{i:\zeta_i=0}(y_i-h_0-c_i g_{e})^2
\right)\\
&\times
(1-\alpha)^m(2\pi\lambda^2)^{-m/2}
\exp\left(
-\frac{1}{2\lambda^2}\displaystyle\sum_{i:\zeta_i=1}(y_i-h_0-c_i g_{e}-\delta_i)^2
\right)(\lambda^2)^{-3/2}
\end{align*}
Letting $y_i' = y_i-h_0$,
\begin{align*}
g_{e} \mid . &\;\propto\;
\exp\left\{
-\frac{1}{2\lambda^2}
\left[
g_{e}^2
\left(
\sum_{i:\zeta_i=0}c_i^2+\sum_{i:\zeta_i=1}c_i^2
\right)
-
2g_{e}
\left(
\sum_{i:\zeta_i=0}c_i y_i'
+
\sum_{i:\zeta_i=1}c_i(y_i'-\delta_i)
\right)
\right]
\right\}\\
&=
\exp\left\{
-\frac{1}{2\lambda^2}
\left(
\sum_{i:\zeta_i=0}c_i^2+\sum_{i:\zeta_i=1}c_i^2
\right)
\left[
g_{e}-
\frac{
\displaystyle
\sum_{i:\zeta_i=0}c_i y_i'
+
\sum_{i:\zeta_i=1}c_i(y_i'-\delta_i)
}{
\displaystyle
\sum_{i:\zeta_i=0}c_i^2+\sum_{i:\zeta_i=1}c_i^2
}
\right]^2
\right\}
\end{align*}
The conditional posterior distribution of $g_e|.$ is indeed Gaussian $\mathcal{N}(\hat{\mu}_{g_e}, \widehat{\mathrm{Var}}_{g_e})$ where

\begin{align*}
    \displaystyle
\hat{\mu}_{g_e}
=
\frac{
\displaystyle
\sum_{i:\zeta_i=0}c_i(y_i-h_0)
+
\sum_{i:\zeta_i=1}c_i(y_i-h_0-\delta_i)
}{
\displaystyle
\sum_{i:\zeta_i=0}c_i^2+\sum_{i:\zeta_i=1}c_i^2
}
\quad \text{and}\quad
\displaystyle
\widehat{\mathrm{Var}}_{g_e}
=
\frac{\lambda^2}{
\displaystyle
\sum_{i:\zeta_i=0}c_i^2+\sum_{i:\zeta_i=1}c_i^2
}
\end{align*}

\paragraph{Case 3: both fixed.}
If both $g_e$ and $h_0$ are fixed, no Gibbs update is performed for $\boldsymbol{\theta}$.

\subsection{Full conditional of $\lambda^2$}\label{sup:cond-lambda}
Using the prior density
$\pi(\delta\mid\lambda^2,k,\gamma_\delta)\propto
(\lambda^2)^{-n/2}\exp\big[-\tfrac{k}{2\lambda^2}\delta^\top\mathrm{Corr}_{\gamma_\delta}^{-1}\delta\big]$
together with the Jeffreys prior
$\pi(\lambda^2)\propto(\lambda^2)^{-(d/2+1)}$,
\begin{align*}
    p(\lambda^2\mid\cdot)&\;\propto\;(\lambda^2)^{-n/2}
    \exp\!\left\{-\tfrac{1}{2\lambda^2}\Big[\textstyle\sum_{i:\zeta_i=0}(y_i-g(x_i)\boldsymbol{\theta})^2+\sum_{i:\zeta_i=1}(y_i-g(x_i)\boldsymbol{\theta}-\delta(x_i))^2\Big]\right\}\\
    &\times
    (\lambda^2)^{-n/2}
    \exp\!\big[-\tfrac{k}{2\lambda^2}\delta^\top\mathrm{Corr}_{\gamma_\delta}^{-1}\delta\big]\,
   (\lambda^2)^{-(d/2+1)}.
\end{align*}
Here $\boldsymbol{\theta}$ denotes the current value of the calibration parameter, with fixed
components kept at their prescribed values whenever such a constrained update is used.
Collecting the powers of $\lambda^2$ gives
\begin{align*}
    p(\lambda^2\mid\cdot)\;\propto\;(\lambda^2)^{-(n+d/2)-1}\exp\!\left\{-\tfrac{C}{\lambda^2}\right\},
\end{align*}
with
\begin{align*}
    C = \tfrac{1}{2}\!\left[\sum_{i:\zeta_i=0}(y_i-g(x_i)\boldsymbol{\theta})^2
    + \sum_{i:\zeta_i=1}(y_i-g(x_i)\boldsymbol{\theta}-\delta(x_i))^2
    + k\,\delta^\top\mathrm{Corr}_{\gamma_\delta}^{-1}\delta\right].
\end{align*}
This is the kernel of an inverse-gamma distribution, so
\begin{equation}\label{eq:supp-lambda-full}
\lambda^2\mid X,Y,\boldsymbol{\theta},\delta,\boldsymbol{\zeta},k,\gamma_\delta
\;\sim\;\mathcal{IG}\!\Big(n+\tfrac{d}{2},\;C\Big).
\end{equation}

\subsection{Full conditional of $k$}\label{sup:cond-k}
Under the prior $k\sim\mathcal{U}(\underline{k},1)$, using
$|\Sigma_\delta|=(\lambda^2/k)^n\,|\mathrm{Corr}_{\gamma_\delta}|$,
\begin{align*}
    p(\delta\mid X, k,\gamma_\delta,\lambda^2)
    \;=\;
    (2\pi)^{-n/2}\,\Big|\tfrac{\lambda^2}{k}\mathrm{Corr}_{\gamma_\delta}\Big|^{-1/2}
    \exp\!\Big[-\tfrac{k}{2\lambda^2}\delta^\top\mathrm{Corr}_{\gamma_\delta}^{-1}\delta\Big]
    \;\propto\;
    k^{n/2}\exp\!\Big[-\tfrac{k}{2\lambda^2}\delta^\top\mathrm{Corr}_{\gamma_\delta}^{-1}\delta\Big].
\end{align*}
Combining this with the prior $k\sim\mathcal{U}(\underline{k},1)$, whose density is
constant on $(\underline{k},1)$, gives
\begin{equation}\label{eq:supp-k-full}
p(k\mid X,\delta,\lambda^2,\gamma_\delta)
\;\propto\;
k^{(n/2+1)-1}\exp\!\Big[-\tfrac{1}{2\lambda^2}\,\delta^\top\mathrm{Corr}_{\gamma_\delta}^{-1}\delta\;k\Big]\,
\mathbf{1}_{(\underline{k},1)}(k).
\end{equation}
which is a Gamma$\big(n/2+1,\;\tfrac{1}{2\lambda^2}\delta^\top\mathrm{Corr}_{\gamma_\delta}^{-1}\delta\big)$
kernel truncated to $(\underline{k},1)$. Samples can be drawn either by inverse-CDF
transform on the truncated Gamma or by rejection sampling using an untruncated Gamma proposal.

\subsection{Full conditional of $\alpha$}\label{sup:cond-alpha}
Under the symmetric Beta prior $\alpha\sim\mathcal{B}(a_0,a_0)$,
and given the allocations $\boldsymbol{\zeta}$, the product of the prior and the completed likelihood 
$\alpha^{\sum_i\mathbf{1}\{\zeta_i=0\}}(1-\alpha)^{\sum_i\mathbf{1}\{\zeta_i=1\}}=\alpha^{n-m}(1-\alpha)^{m}$ contribute
\begin{equation}\label{eq:supp-alpha-full}
\alpha\mid\boldsymbol{\zeta}\;\sim\;
\mathcal{B}\!\big(n-m+a_0,\;m+a_0\big)
\end{equation}
used in the Metropolis-within-Gibbs algorithm of Section~3 
of the main paper.
\section{Orthogonal Gaussian process prior for the discrepancy}\label{sup:ogp}

As mentioned in Section~2 of the main paper, a fundamental difficulty in
Bayesian calibration is the confounding between the calibration parameters
$\boldsymbol{\theta}$ and the discrepancy function $\delta(X)$, since only
the sum $g(x)\boldsymbol{\theta}+\delta(x)$ is directly informed by the data.
To reduce this non-identifiability, \citet{plumlee2017} proposed the
\emph{orthogonal Gaussian process} (OGP) construction, which restricts the
discrepancy to the subspace orthogonal to the span of the computer-code
outputs, through the constraint
\begin{equation}\label{eq:supp-ortho}
\int_{\mathcal{X}} g(\xi)\,\delta(\xi)\,\mathrm{d}\xi = 0.
\end{equation}
This prevents the discrepancy from mimicking variations in the calibration
parameters. Letting $c(x,x')$ denote the covariance kernel of the
(non-orthogonal) GP prior on $\delta$, the OGP prior replaces it with the
orthogonalized kernel
\begin{equation}\label{eq:supp-cstar}
c^{*}(x,x') = c(x,x') - h(x)^\top H^{-1} h(x'),
\end{equation}
where
\begin{equation}\label{eq:supp-hH}
h(x) = \int_{\mathcal{X}} c(x,\xi)\,g(\xi)\,\mathrm{d}\xi,
\qquad
H = \iint_{\mathcal{X}\times\mathcal{X}} c(\xi',\xi)\,g(\xi)\,g(\xi')^\top\,\mathrm{d}\xi\,\mathrm{d}\xi'.
\end{equation}

\subsection{Closed-form expressions for $h(x)$ and $H$ in the ball-drop application}
\label{sec:closed-form-h-H}
 
The discrepancy process $\delta$ is a Gaussian process on $\mathcal{X}=[0,1]$ in the
rescaled variable
\begin{align*}
    \xi \;=\; \frac{t - t_{\min}}{a}, \qquad a \;=\; t_{\max} - t_{\min},
\end{align*}
with the exponential covariance kernel
\begin{align*}
    c(x,\xi) = \frac{\lambda^{2}}{k}\exp\!\left(-\frac{|x-\xi|}{\gamma_\delta}\right),\qquad\gamma_\delta>0,
\end{align*}
where $\sigma_{\delta}^{2} = \lambda^{2}/k$ is the marginal variance and
$\gamma_{\delta}$ is the correlation length (both on the rescaled $[0,1]$ scale).
The physical model, in normalized coordinates, is
\begin{align*}
    h(\xi,\boldsymbol{\theta}) \;=\; h_{0} - \tfrac{1}{2}\, g_e\,(a\xi + t_{\min})^{2},
    \qquad \boldsymbol{\theta} \;=\; (h_{0},\, g_e)^{\top},
\end{align*}
so that the design row in normalized coordinates is
\begin{align*}
    g(\xi) \;=\;
    \begin{pmatrix}
       1 \\[2pt]
       -\dfrac{1}{2}\,(a\xi + t_{\min})^{2}
    \end{pmatrix}.
\end{align*}
To evaluate~\eqref{eq:supp-cstar}, we compute both $h(x)$ and $H$ in closed forms and throughout, we write $\gamma$ for $\gamma_{\delta}$
and $E \;=\; e^{-1/\gamma}$ in order to simplify notations.

\subsubsection{Computation of $h(x)$}
 
\begin{equation}\label{eq:supp-h-expr}
h(x) = \frac{\lambda^{2}}{k}
\begin{pmatrix}
J(x,\gamma)\\[6pt]
-\dfrac{1}{2}\,\left[a^{2} I(x,\gamma) + 2a\,t_{\min}K(x,\gamma)+t_{\min}^{2}J(x,\gamma)\right]
\end{pmatrix},
\end{equation}
We define
\begin{align*}
    \left\{
    \begin{aligned}
      J(x,\gamma) &= \int_0^1 e^{-|x-\xi|/\gamma}\,\mathrm{d}\xi,\\
      K(x,\gamma) &= \int_{0}^{1} \xi\, e^{-|x-\xi|/\gamma}\, d\xi,\\
      I(x,\gamma) &= \int_{0}^{1} \xi^{2}\, e^{-|x-\xi|/\gamma}\, d\xi.
    \end{aligned}
    \right.
\end{align*}

\subsubsection*{Computation of $J(x,\gamma)$}

\textbf{Left piece.} Substitute $u = x - \xi$:
\begin{align*}
    \int_{0}^{x} e^{-(x-\xi)/\gamma}\, d\xi
    \;=\; \int_{0}^{x} e^{-u/\gamma}\, du
     \;=\; \gamma\bigl(1 - e^{-x/\gamma}\bigr).
\end{align*}
 
\textbf{Right piece.} Substitute $u = \xi - x$:
\begin{align*}
    \int_{x}^{1} e^{-(\xi-x)/\gamma}\, d\xi
    \;=\; \int_{0}^{1-x} e^{-u/\gamma}\, du
    \;=\; \gamma\bigl(1 - e^{-(1-x)/\gamma}\bigr).
\end{align*}
 
\textbf{Sum.}
\begin{align*}
    J(x,\gamma) \;=\; \gamma\bigl(2 - e^{-x/\gamma} - e^{-(1-x)/\gamma}\bigr).
\end{align*}

\subsubsection*{Computation of $K(x,\gamma)$}

\begin{align*}
    \int_{0}^{1} \xi\, e^{-|x-\xi|/\gamma}\, d\xi
    \;=\; \int_{0}^{x} \xi\, e^{-(x-\xi)/\gamma}\, d\xi
    \;+\; \int_{x}^{1} \xi\, e^{-(\xi-x)/\gamma}\, d\xi.
\end{align*}
 
\textbf{Left piece $K_{L}(x)$.}
\begin{align*}
    K_{L}(x) \;=\; e^{-x/\gamma}\int_{0}^{x}\xi\, e^{\xi/\gamma}\, d\xi.
\end{align*}

Integration by parts with $u = \xi$, $dv = e^{\xi/\gamma}\, d\xi$, so $v = \gamma\, e^{\xi/\gamma}$:
\begin{align*}
    \int_{0}^{x}\xi\, e^{\xi/\gamma}\, d\xi
    \;=\; \bigl[\gamma\xi\, e^{\xi/\gamma}\bigr]_{0}^{x} - \gamma\int_{0}^{x} e^{\xi/\gamma}\, d\xi
    \;=\; \gamma x\, e^{x/\gamma} - \gamma^{2}\bigl(e^{x/\gamma} - 1\bigr).
\end{align*}
 which deduce
\begin{align*}
    K_{L}(x) \;=\; \gamma x - \gamma^{2} + \gamma^{2}\, e^{-x/\gamma}.
\end{align*}
\textbf{Right piece $K_{R}(x)$.}
\begin{align*}
    K_{R}(x) \;=\; \int_{0}^{1-x}(u+x)\, e^{-u/\gamma}\, du
    \;=\; \int_{0}^{1-x} u\, e^{-u/\gamma}\, du \;+\; x\int_{0}^{1-x} e^{-u/\gamma}\, du.
\end{align*}
The second integral is $\gamma\bigl(1 - e^{-(1-x)/\gamma}\bigr)$. For the first, integration by parts:
\begin{align*}
    \int_{0}^{1-x} u\, e^{-u/\gamma}\, du
    \;=\; \bigl[-\gamma u\, e^{-u/\gamma}\bigr]_{0}^{1-x} + \gamma\int_{0}^{1-x} e^{-u/\gamma}\, du
    \;=\; -\gamma(1-x)\, e^{-(1-x)/\gamma} + \gamma^{2}\bigl(1 - e^{-(1-x)/\gamma}\bigr).
\end{align*}
which leads to 
\begin{align*}
    K_{R}(x) \;=\; \gamma x + \gamma^{2} - \gamma(1+\gamma)\, e^{-(1-x)/\gamma}
\end{align*}

\textbf{Sum.}
\begin{align*}
    K(x, \gamma) \;=\; 2\gamma x \;+\; \gamma^{2}\, e^{-x/\gamma} \;-\; \gamma(1+\gamma)\, e^{-(1-x)/\gamma}.
\end{align*}

\subsubsection*{Computation of $I(x,\gamma)$}

\begin{align*}
    \int_{0}^{1} \xi^{2}\, e^{-|x-\xi|/\gamma}\, d\xi
\;=\; \int_{0}^{x} \xi^{2}\, e^{-(x-\xi)/\gamma}\, d\xi
\;+\; \int_{x}^{1} \xi^{2}\, e^{-(\xi-x)/\gamma}\, d\xi.
\end{align*}

\textbf{Left piece $I_{L}(x)$.} Integrate by parts twice.
First, $u = \xi^{2}$, $dv = e^{\xi/\gamma}d\xi$:
\begin{align*}
    \int_{0}^{x}\xi^{2} e^{\xi/\gamma}\, d\xi
    \;=\; \gamma x^{2} e^{x/\gamma} - 2\gamma\int_{0}^{x}\xi\, e^{\xi/\gamma}\, d\xi.
\end{align*}

\begin{align*}
    \int_{0}^{x}\xi^{2} e^{\xi/\gamma}\, d\xi
    \;=\; \gamma x^{2} e^{x/\gamma} - 2\gamma^{2} x\, e^{x/\gamma} + 2\gamma^{3} e^{x/\gamma} - 2\gamma^{3}.
\end{align*}
 
\begin{align*}
    I_{L}(x) \;=\; \gamma x^{2} - 2\gamma^{2} x + 2\gamma^{3} - 2\gamma^{3}\, e^{-x/\gamma}.
\end{align*}

\textbf{Right piece $I_{R}(x)$.} Substitute $u = \xi - x$ and expand $(u+x)^{2}$:
\begin{align*}
    I_{R}(x)
    \;=\; \int_{0}^{1-x} u^{2}\, e^{-u/\gamma}\, du
    \;+\; 2x\int_{0}^{1-x} u\, e^{-u/\gamma}\, du
    \;+\; x^{2}\int_{0}^{1-x} e^{-u/\gamma}\, du.
\end{align*}

The last two have already been computed. For the first, integration by parts with $u' = u^{2}$:
\begin{align*}
    \int_{0}^{1-x} u^{2}\, e^{-u/\gamma}\, du
    \;=\; -\gamma(1-x)^{2} e^{-(1-x)/\gamma} + 2\gamma\int_{0}^{1-x} u\, e^{-u/\gamma}\, du.
\end{align*}

\begin{align*}
    I_{R}(x) \;=\; \gamma x^{2} + 2\gamma^{2} x + 2\gamma^{3} - \bigl(\gamma + 2\gamma^{2} + 2\gamma^{3}\bigr)\, e^{-(1-x)/\gamma}.
\end{align*}

\textbf{Sum.} 
\begin{align*}
    I(x,\gamma) \;=\; 2\gamma x^{2} + 4\gamma^{3} - 2\gamma^{3}\, e^{-x/\gamma} - \gamma\bigl(1 + 2\gamma + 2\gamma^{2}\bigr)\, e^{-(1-x)/\gamma}.
\end{align*}
 
\subsubsection{Computation of $H(\gamma_{\delta})$}
$H$ is the $2\times 2$ matrix
\begin{align*}
    H(\gamma_{\delta})
    \;=\; \int_{0}^{1}\!\!\int_{0}^{1} g(x)\, e^{-|x-\xi|/\gamma}\, g(\xi)^{\top}\, dx\, d\xi.
\end{align*}
 
By symmetry of the kernel, $H_{12} = H_{21}$, so three distinct entries must be computed.
The full covariance-scaled matrix is $(\lambda^{2}/k)\, H(\gamma_{\delta})$.

For each entry, we perform the inner integral over $\xi$ first; the result is
$J(x,\gamma)$, $K(x,\gamma)$, or $I(x,\gamma)$. The remaining (outer) integral over $x$ is then a
single integral on $[0,1]$.

\subsubsection*{Entry $H_{11}$}
 
With $g_{1}(\xi) = 1$,
\begin{align*}
    H_{11}
    \;=\; \int_{0}^{1}\Bigl[\underbrace{\int_{0}^{1} e^{-|x-\xi|/\gamma}\, d\xi}_{=\, J(x,\gamma)}\Bigr]\, dx
    \;=\; \int_{0}^{1} J(x,\gamma)\, dx.
\end{align*}

We integrate $J(x,\gamma) = 2\gamma - \gamma e^{-x/\gamma} - \gamma e^{-(1-x)/\gamma}$ piece by piece:
\begin{itemize}
\item $\displaystyle \int_{0}^{1} 2\gamma\, dx \;=\; 2\gamma$.
\item $\displaystyle \int_{0}^{1} e^{-x/\gamma}\, dx \;=\; \gamma - \gamma E$, so $\displaystyle -\gamma\int = -\gamma^{2} + \gamma^{2} E$.
\item $\displaystyle \int_{0}^{1} e^{-(1-x)/\gamma}\, dx \;=\; \gamma - \gamma E$ (by $u = 1-x$), so $\displaystyle -\gamma\int = -\gamma^{2} + \gamma^{2} E$.
\end{itemize}
Summing:
\begin{align*}
    H_{11} \;=\; 2\gamma - 2\gamma^{2} + 2\gamma^{2}\, E.
\end{align*}

\subsubsection*{Entry $H_{12} = H_{21}$}
 
With $g_{1}(x) = 1$ and $g_{2}(\xi) = -\tfrac{1}{2}(a\xi + t_{\min})^{2}$, the inner integral is
\begin{align*}
    \int_{0}^{1} (a\xi + t_{\min})^{2}\, e^{-|x-\xi|/\gamma}\, d\xi
    \;=\; a^{2}\, I(x,\gamma) \;+\; 2 a t_{\min}\, K(x,\gamma) \;+\; t_{\min}^{2}\, J(x,\gamma).
\end{align*}

Hence
\begin{align*}
    H_{12} \;=\; -\tfrac{1}{2}\Bigl[a^{2}\int_{0}^{1} I(x,\gamma)\, dx \;+\; 2 a t_{\min}\int_{0}^{1} K(x,\gamma)\, dx \;+\; t_{\min}^{2}\int_{0}^{1} J(x,\gamma)\, dx\Bigr].
\end{align*}

The last integral is $H_{11}$, already computed. The other two are computed below.
 
\paragraph{Integral $\int_{0}^{1} K(x,\gamma)\, dx$.}
With $K(x,\gamma) = 2\gamma x + \gamma^{2} e^{-x/\gamma} - \gamma(1+\gamma)\, e^{-(1-x)/\gamma}$:
\begin{itemize}
\item $\displaystyle \int_{0}^{1} 2\gamma x\, dx \;=\; \gamma$.
\item $\displaystyle \gamma^{2}\int_{0}^{1} e^{-x/\gamma}\, dx \;=\; \gamma^{2}(\gamma - \gamma E) \;=\; \gamma^{3} - \gamma^{3} E$.
\item $\displaystyle -\gamma(1+\gamma)\int_{0}^{1} e^{-(1-x)/\gamma}\, dx \;=\; -\gamma^{2}(1+\gamma)(1 - E) \;=\; -\gamma^{2} - \gamma^{3} + (\gamma^{2} + \gamma^{3})\, E$.
\end{itemize}
Summing:
\begin{align*}
    \int_{0}^{1} K(x,\gamma)\, dx \;=\; \gamma - \gamma^{2} + \gamma^{2}\, E.
\end{align*}
 
\paragraph{Integral $\int_{0}^{1} I(x,\gamma)\, dx$.}
With $I(x,\gamma) = 2\gamma x^{2} + 4\gamma^{3} - 2\gamma^{3} e^{-x/\gamma} - \gamma(1 + 2\gamma + 2\gamma^{2})\, e^{-(1-x)/\gamma}$:
\begin{itemize}
\item $\displaystyle \int_{0}^{1} 2\gamma x^{2}\, dx \;=\; \tfrac{2\gamma}{3}$.
\item $\displaystyle \int_{0}^{1} 4\gamma^{3}\, dx \;=\; 4\gamma^{3}$.
\item $\displaystyle -2\gamma^{3}\int_{0}^{1} e^{-x/\gamma}\, dx \;=\; -2\gamma^{3}(\gamma - \gamma E) \;=\; -2\gamma^{4} + 2\gamma^{4} E$.
\item $\displaystyle -\gamma(1 + 2\gamma + 2\gamma^{2})\int_{0}^{1} e^{-(1-x)/\gamma}\, dx
\;=\; -\gamma^{2}(1 + 2\gamma + 2\gamma^{2})(1 - E)$\\
\phantom{xx}$\;=\; -\gamma^{2} - 2\gamma^{3} - 2\gamma^{4} + (\gamma^{2} + 2\gamma^{3} + 2\gamma^{4})\, E$.
\end{itemize}
Summing:
\begin{align*}
    \int_{0}^{1} I(x,\gamma)\, dx \;=\; \tfrac{2\gamma}{3} - \gamma^{2} + 2\gamma^{3} - 4\gamma^{4} + \bigl(\gamma^{2} + 2\gamma^{3} + 4\gamma^{4}\bigr)\, E.
\end{align*}
 
\paragraph{Assembly of $H_{12}$.}
\begin{align*}
    \begin{aligned}
    H_{12} \;=\; -\tfrac{1}{2}\Bigl[\;
    & a^{2}\Bigl(\tfrac{2\gamma}{3} - \gamma^{2} + 2\gamma^{3} - 4\gamma^{4} + (\gamma^{2} + 2\gamma^{3} + 4\gamma^{4})\, E\Bigr) \\
    +\;& 2 a t_{\min}\Bigl(\gamma - \gamma^{2} + \gamma^{2}\, E\Bigr) \\
    +\;& t_{\min}^{2}\Bigl(2\gamma - 2\gamma^{2} + 2\gamma^{2}\, E\Bigr)
    \;\Bigr].
    \end{aligned}
\end{align*}

(When $t_{\min} = 0$, this collapses to $-\tfrac{a^{2}}{2}\int_{0}^{1} I(x,\gamma)\, dx$, recovering the
expression for the un-shifted design.)
 
\subsubsection*{Entry $H_{22}$}
 
With $g_{2}(x) = -\tfrac{1}{2}(ax + t_{\min})^{2}$, the inner integral is again
$a^{2} I(x,\gamma) + 2at_{\min} K(x,\gamma) + t_{\min}^{2} J(x,\gamma)$, so
\begin{align*}
    H_{22} \;=\; \tfrac{1}{4}\int_{0}^{1}(ax + t_{\min})^{2}\Bigl[a^{2} I(x,\gamma) + 2at_{\min} K(x,\gamma) + t_{\min}^{2} J(x,\gamma)\Bigr]\, dx.
\end{align*}

We will use the following elementary integrals:
\begin{align*}
\int_{0}^{1} x\, e^{-x/\gamma}\, dx &\;=\; \gamma^{2} - \gamma E - \gamma^{2} E,
& \int_{0}^{1} x\, e^{-(1-x)/\gamma}\, dx &\;=\; \gamma - \gamma^{2} + \gamma^{2} E, \\
\int_{0}^{1} x^{2}\, e^{-x/\gamma}\, dx &\;=\; 2\gamma^{3} - \gamma E - 2\gamma^{2} E - 2\gamma^{3} E,
& \int_{0}^{1} x^{2}\, e^{-(1-x)/\gamma}\, dx &\;=\; \gamma - 2\gamma^{2} + 2\gamma^{3} - 2\gamma^{3} E.
\end{align*}

\paragraph{Integral $\int_{0}^{1} x J(x,\gamma)\, dx$.}
\begin{itemize}
\item $\displaystyle \int 2\gamma x\, dx \;=\; \gamma$.
\item $\displaystyle -\gamma \int x e^{-x/\gamma}\, dx \;=\; -\gamma^{3} + \gamma^{2} E + \gamma^{3} E$.
\item $\displaystyle -\gamma \int x e^{-(1-x)/\gamma}\, dx \;=\; -\gamma^{2} + \gamma^{3} - \gamma^{3} E$.
\end{itemize}
Summing (non-$E$: $\gamma - \gamma^{2}$; $E$ coeff: $\gamma^{2}$):
\begin{align*}
    \int_{0}^{1} x J(x,\gamma)\, dx \;=\; \gamma - \gamma^{2} + \gamma^{2}\, E.
\end{align*}
 
\paragraph{Integral $\int_{0}^{1} x^{2} J(x,\gamma)\, dx$.}
\begin{itemize}
\item $\displaystyle \int 2\gamma x^{2}\, dx \;=\; \tfrac{2\gamma}{3}$.
\item $\displaystyle -\gamma \int x^{2} e^{-x/\gamma}\, dx \;=\; -2\gamma^{4} + \gamma^{2} E + 2\gamma^{3} E + 2\gamma^{4} E$.
\item $\displaystyle -\gamma \int x^{2} e^{-(1-x)/\gamma}\, dx \;=\; -\gamma^{2} + 2\gamma^{3} - 2\gamma^{4} + 2\gamma^{4} E$.
\end{itemize}
Summing:
\begin{align*}
    \int_{0}^{1} x^{2} J(x,\gamma)\, dx \;=\; \tfrac{2\gamma}{3} - \gamma^{2} + 2\gamma^{3} - 4\gamma^{4} + (\gamma^{2} + 2\gamma^{3} + 4\gamma^{4})\, E.
\end{align*}
 
\paragraph{Integral $\int_{0}^{1} x K(x,\gamma)\, dx$.}
\begin{itemize}
\item $\displaystyle \int 2\gamma x^{2}\, dx \;=\; \tfrac{2\gamma}{3}$.
\item $\displaystyle \gamma^{2}\int x e^{-x/\gamma}\, dx \;=\; \gamma^{4} - \gamma^{3} E - \gamma^{4} E$.
\item $\displaystyle -\gamma(1+\gamma)\int x e^{-(1-x)/\gamma}\, dx
\;=\; -\gamma^{2}(1-\gamma^{2}) - \gamma^{3}(1+\gamma)\, E$
\\ \phantom{xx} $\;=\; -\gamma^{2} + \gamma^{4} - (\gamma^{3} + \gamma^{4})\, E$.
\end{itemize}
Summing:
\begin{align*}
    \int_{0}^{1} x K(x,\gamma)\, dx \;=\; \tfrac{2\gamma}{3} - \gamma^{2} + 2\gamma^{4} - (2\gamma^{3} + 2\gamma^{4})\, E.
\end{align*}
 
\paragraph{Integral $\int_{0}^{1} x^{2} K(x,\gamma)\, dx$.}
\begin{itemize}
\item $\displaystyle \int 2\gamma x^{3}\, dx \;=\; \tfrac{\gamma}{2}$.
\item $\displaystyle \gamma^{2}\int x^{2} e^{-x/\gamma}\, dx \;=\; 2\gamma^{5} - \gamma^{3} E - 2\gamma^{4} E - 2\gamma^{5} E$.
\item $\displaystyle -\gamma(1+\gamma)\int x^{2} e^{-(1-x)/\gamma}\, dx
\;=\; -\gamma^{2} + \gamma^{3} - 2\gamma^{5} + (2\gamma^{4} + 2\gamma^{5})\, E$,
\\ using $(1+\gamma)(\gamma - 2\gamma^{2} + 2\gamma^{3}) = \gamma - \gamma^{2} + 2\gamma^{4}$.
\end{itemize}
\begin{align*}
    \int_{0}^{1} x^{2} K(x,\gamma)\, dx \;=\; \tfrac{\gamma}{2} - \gamma^{2} + \gamma^{3} - \gamma^{3}\, E.
\end{align*}
 
\paragraph{Integral $\int_{0}^{1} x I(x,\gamma)\, dx$.}
\begin{itemize}
\item $\displaystyle \int 2\gamma x^{3}\, dx \;=\; \tfrac{\gamma}{2}$.
\item $\displaystyle \int 4\gamma^{3} x\, dx \;=\; 2\gamma^{3}$.
\item $\displaystyle -2\gamma^{3}\int x e^{-x/\gamma}\, dx \;=\; -2\gamma^{5} + 2\gamma^{4} E + 2\gamma^{5} E$.
\item $\displaystyle -\gamma(1 + 2\gamma + 2\gamma^{2})\int x e^{-(1-x)/\gamma}\, dx
\;=\; -\gamma^{2} - \gamma^{3} + 2\gamma^{5} - (\gamma^{3} + 2\gamma^{4} + 2\gamma^{5})\, E$,
\\ using $(1 + 2\gamma + 2\gamma^{2})(\gamma - \gamma^{2}) = \gamma + \gamma^{2} - 2\gamma^{4}$.
\end{itemize}
Summing:
\begin{align*}
    \int_{0}^{1} x I(x,\gamma)\, dx \;=\; \tfrac{\gamma}{2} - \gamma^{2} + \gamma^{3} - \gamma^{3}\, E.
\end{align*}
 
\paragraph{Integral $\int_{0}^{1} x^{2} I(x,\gamma)\, dx$.}
\begin{itemize}
\item $\displaystyle \int 2\gamma x^{4}\, dx \;=\; \tfrac{2\gamma}{5}$.
\item $\displaystyle \int 4\gamma^{3} x^{2}\, dx \;=\; \tfrac{4\gamma^{3}}{3}$.
\item $\displaystyle -2\gamma^{3}\int x^{2} e^{-x/\gamma}\, dx
\;=\; -4\gamma^{6} + 2\gamma^{4} E + 4\gamma^{5} E + 4\gamma^{6} E$.
\item $\displaystyle -\gamma(1 + 2\gamma + 2\gamma^{2})\int x^{2} e^{-(1-x)/\gamma}\, dx
\;=\; -\gamma^{2} - 4\gamma^{6} + (2\gamma^{4} + 4\gamma^{5} + 4\gamma^{6})\, E$,
\\ using $(1 + 2\gamma + 2\gamma^{2})(1 - 2\gamma + 2\gamma^{2}) = 1 + 4\gamma^{4}$.
\end{itemize}
Summing:
\begin{align*}
    \int_{0}^{1} x^{2} I(x,\gamma)\, dx \;=\; \tfrac{2\gamma}{5} - \gamma^{2} + \tfrac{4\gamma^{3}}{3} - 8\gamma^{6} + (4\gamma^{4} + 8\gamma^{5} + 8\gamma^{6})\, E.
\end{align*}
 
\begin{align*}
    \begin{aligned}
       H_{22} \;=\; \tfrac{1}{4}\Bigl[\;
       & a^{4}\Bigl(\tfrac{2\gamma}{5} - \gamma^{2} + \tfrac{4\gamma^{3}}{3} - 8\gamma^{6} + (4\gamma^{4} + 8\gamma^{5} + 8\gamma^{6})\, E\Bigr) \\
       +\;& 4 a^{3} t_{\min}\Bigl(\tfrac{\gamma}{2} - \gamma^{2} + \gamma^{3} - \gamma^{3}\, E\Bigr) \\
       +\;& 2 a^{2} t_{\min}^{2}\Bigl[\Bigl(\tfrac{2\gamma}{3} - \gamma^{2} + 2\gamma^{3} - 4\gamma^{4} + (\gamma^{2} + 2\gamma^{3} + 4\gamma^{4})\, E\Bigr) \\[-2pt]
       & \hphantom{2 a^{2} t_{\min}^{2}\Bigl[\,} + 2\Bigl(\tfrac{2\gamma}{3} - \gamma^{2} + 2\gamma^{4} - (2\gamma^{3} + 2\gamma^{4})\, E\Bigr)\Bigr] \\
       +\;& 4 a t_{\min}^{3}\Bigl(\gamma - \gamma^{2} + \gamma^{2}\, E\Bigr) \\
       +\;& t_{\min}^{4}\Bigl(2\gamma - 2\gamma^{2} + 2\gamma^{2}\, E\Bigr)
       \;\Bigr].
    \end{aligned}
\end{align*}

(When $t_{\min} = 0$, only the first line survives, giving
$H_{22} = \tfrac{a^{4}}{4}\int_{0}^{1} x^{2} I(x,\gamma)\, dx$, recovering the expression for the
un-shifted design.)

\section{Additional simulation results}\label{sup:sim}

Recall that Section~4.1 reports four inferential settings obtained by combining two modeling choices: the use of thresholding versus no thresholding and fixing $g_e=g_e^\ast$ versus estimating $g_e$ within the MCMC algorithm. All experiments are conducted on datasets simulated from $\MF_1$, where the discrepancy is inactive on the first part of the input domain and becomes active afterwards.
The four settings considered are:
no thresholding with $(g_e,h_0)$ estimated;
no thresholding with $g_e$ fixed at its true value (main paper);
thresholding with $(g_e,h_0)$ estimated;
thresholding with $g_e$ fixed (main paper).
Sections~\ref{sup:scnI_III}--\ref{sup:oracle} report the complete results for the cases with $(g_e,h_0)$ estimated and thresholding, additional experiments under the OGP prior, and an oracle identifiability diagnostic for $(k,\gamma_\delta)$. Throughout this section, we use the same simulation setup as in Section~4.1: $\boldsymbol{\theta}^\ast=(46.45,9.8)^\top$, ${\lambda^2}^\ast=0.01$, $k^\ast=0.1$, $n=45$ design points on the normalized time instants of the \emph{Blue Basketball} experiment, and $50$ replicated datasets per configuration.

\subsection{Additional diagnostics for the simulation under $\MF_0$}\label{sup:m0_extra}

This subsection collects two complementary diagnostics for the experiment of Section~4.1 of the main paper in which $50$ datasets of size $n=45$ are simulated from $\MF_0$
~and inference is carried out under the mixture model $\MF_\alpha$. The main paper
reports the posterior summaries of $(\boldsymbol{\theta},\lambda^2,\alpha)$; we report here (i)~a posterior-predictive check against the simulated trajectories
and (ii)~the empirical coverage of central credible intervals for
$(g_e,h_0,\lambda^2)$.

\paragraph{Posterior-predictive check.}
 
\begin{figure}[h!]
    \centering
    \begin{minipage}{0.45\textwidth}
        \centering
        \includegraphics[width=\linewidth]{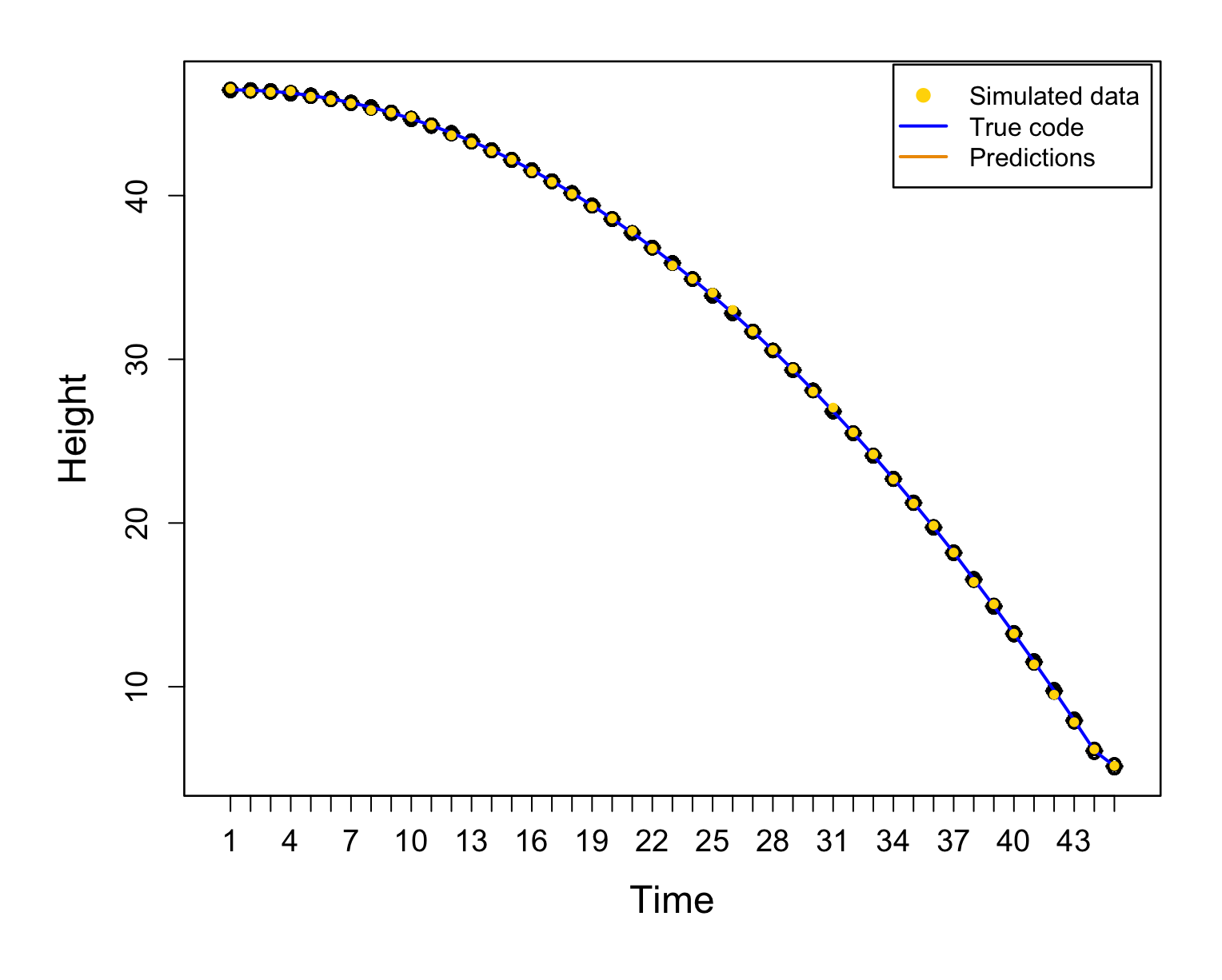}
    \end{minipage}
    \hfill
    \begin{minipage}{0.45\textwidth}
        \centering
        \includegraphics[width=\linewidth]{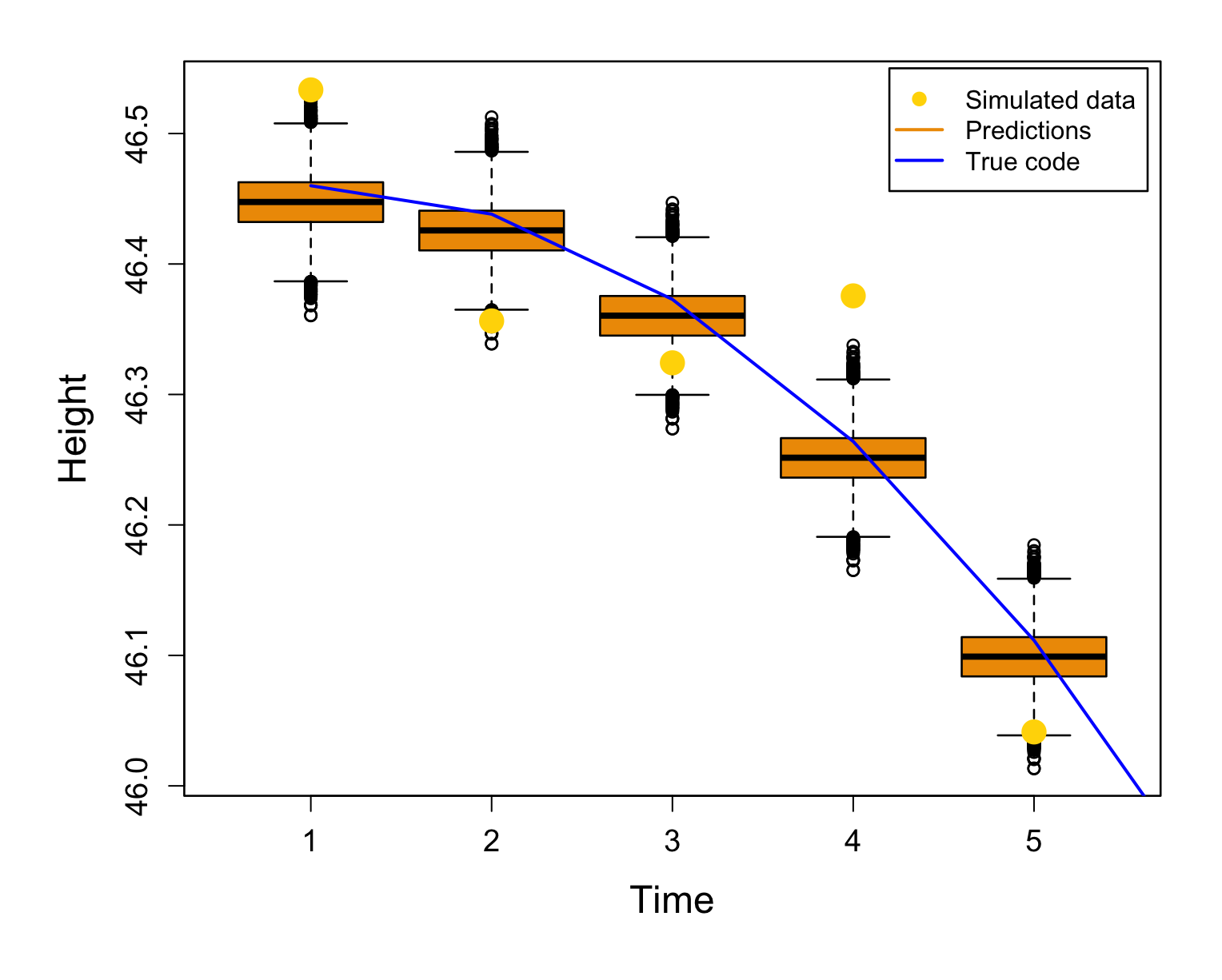}
    \end{minipage}
    \caption{Comparison between a simulated dataset and the posterior predictive distribution under the mixture model $\MF_\alpha$ when the data come from $\MF_0$. Left: full input domain. Right: zoomed view on the first five design points. In this experiment, the variance parameter $\lambda^2$ is endowed with the Jeffreys prior.}
    \label{fig:compare-classic}
\end{figure}
Figure~\ref{fig:compare-classic} compares the simulated observations with the posterior predictive distribution. The left panel shows that the posterior predictive mean closely tracks the true code output across the entire input domain, with no visible systematic deviation. The right panel provides a zoomed view of the first five input points. The true code values lie near the
centre of the posterior predictive distributions, and the spread of the boxplots reflects only observational variability. No systematic bias or structural discrepancy is observed.

\paragraph{Empirical coverage of credible intervals.}
To assess the calibration accuracy of the posterior distributions obtained from the MCMC algorithm, we evaluate the empirical coverage of central credible intervals across the $50$ simulated datasets. For each dataset and each parameter $\phi \in \{g_e, h_0,\lambda^2\}$, we construct a central posterior credible interval from the empirical quantiles of the MCMC chain and record whether the true value lies inside it.

For each parameter $\phi \in \{g_e,h_0,\lambda^2\}$ and each simulated dataset $v=1,\ldots,V$ (with $V=50$), we compute a central posterior credible interval $C_{\phi,v}^{(\ell)}=\bigl[L_{\phi,v}^{(\ell)},\,U_{\phi,v}^{(\ell)}\bigr],$ at level $\ell \in \{0.50,0.95\}$.
The bounds $L_{\phi,v}^{(\ell)}$ and $U_{\phi,v}^{(\ell)}$ are computed from the post burn-in MCMC sample of $\phi$ as the empirical quantiles of order $(1-\ell)/2$ and $(1+\ell)/2$, respectively; that is, $L_{\phi,v}^{(0.95)}$ and $U_{\phi,v}^{(0.95)}$ correspond to the empirical $2.5\%$ and $97.5\%$ quantiles of the posterior sample, while $L_{\phi,v}^{(0.50)}$ and $U_{\phi,v}^{(0.50)}$ correspond to the empirical $25\%$ and $75\%$ quantiles. The empirical coverage reported in Table~\ref{tab:ci_coverage} is defined as $
\mathrm{Coverage}_{\phi}^{(\ell)}
=
\frac{1}{V}\sum_{v=1}^{V}
\mathbf{1}\!\left\{\phi^\star \in C_{\phi,v}^{(\ell)}\right\},$
that is, the proportion of simulated datasets for which the true parameter value $\phi^\star$ belongs to the corresponding posterior credible interval. The mean CI length is
$
\mathrm{Mean\ CI\ length}_{\phi}^{(\ell)}
=
\frac{1}{V}\sum_{v=1}^{V}
\left(U_{\phi,v}^{(\ell)}-L_{\phi,v}^{(\ell)}\right),
$
and summarizes the average width of the posterior credible intervals across replications. Finally, letting  $\hat{\phi}_v=\mathbb{E}(\phi\mid Y_v)$ denote the posterior mean of the parameter $\phi \in \{g_e,h_0,\lambda^2\}$ computed from the post burn-in MCMC sample of dataset $v$, the mean absolute bias is computed as 
$
\mathrm{Mean\ abs.\ bias}_{\phi}
=
\frac{1}{V}\sum_{v=1}^{V}\left|\hat{\phi}_v-\phi^\star\right|.
$
Since this quantity is based on posterior means rather than on credible intervals, it does not depend on the confidence level, which explains why the same value is reported for the $50\%$ and $95\%$ rows.
  \begin{table}[h!]
\centering
\begin{tabular}{lcccc}
\hline
Parameter & CI level & Coverage & Mean CI length & Mean abs.\ bias \\
\hline
$g_e$        & $50\%(95\%)$ & $42\%(96\%)$ & 0.0151(0.0446) & 0.0092(0.0092) \\
$h_0$      & $50\%(95\%)$ & $48\%(90\%)$ & 0.0295(0.0874) & 0.0188(0.0188) \\
$\lambda^2$& $50\%(95\%)$ & $48\%(96\%)$ & 0.0028(0.0086) & 0.0017(0.0017) \\
\hline
\end{tabular}
\caption{Empirical coverage, average length and mean absolute bias of central credible intervals for $g_e$, $h_0$ and $\lambda^2$ across 50 simulated datasets.}
\label{tab:ci_coverage}
\end{table}
Table~\ref{tab:ci_coverage} shows that the empirical coverage of all three parameters is close to the nominal levels, with a mild under-coverage at the $50\%$ level ($42\%$--$48\%$) and coverage between $90\%$ and $96\%$ at the $95\%$ level. The mean absolute bias is small for $g_e$, $h_0$ and $\lambda^2$ alike, indicating accurate recovery of both the calibration parameters and the noise variance. 
The slight under-coverage at the $50\%$ level is consistent with the well-known difficulty of fully calibrating variance-related uncertainty in GP discrepancy models, where uncertainty on scale parameters tends to be underestimated.

\begin{figure}[h!]
  \centering
  \begin{subfigure}{0.28\textwidth}
    \centering
    \includegraphics[width=\linewidth]{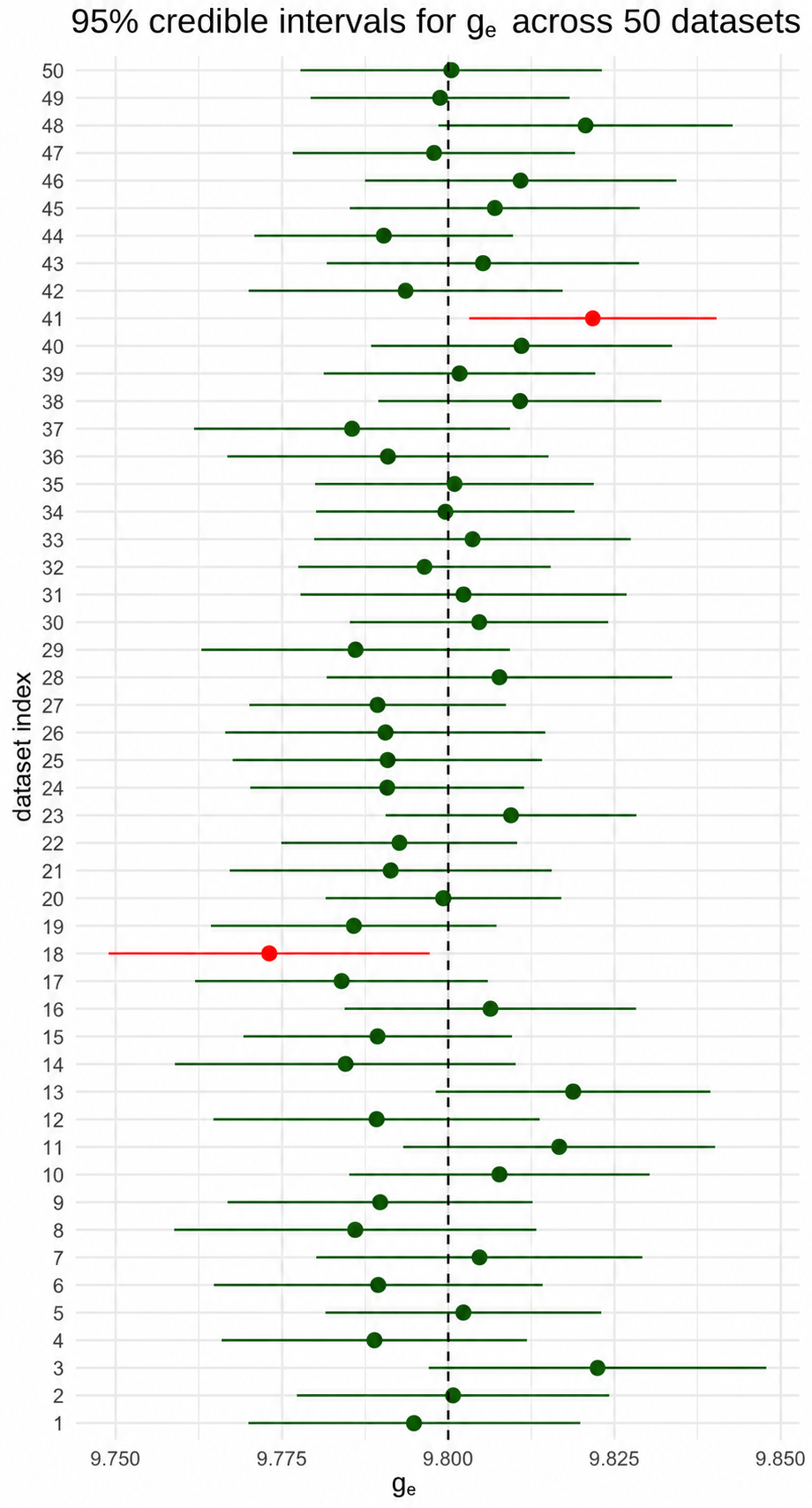}
  \end{subfigure}
  \hfill
  \begin{subfigure}{0.28\textwidth}
    \centering
    \includegraphics[width=\linewidth]{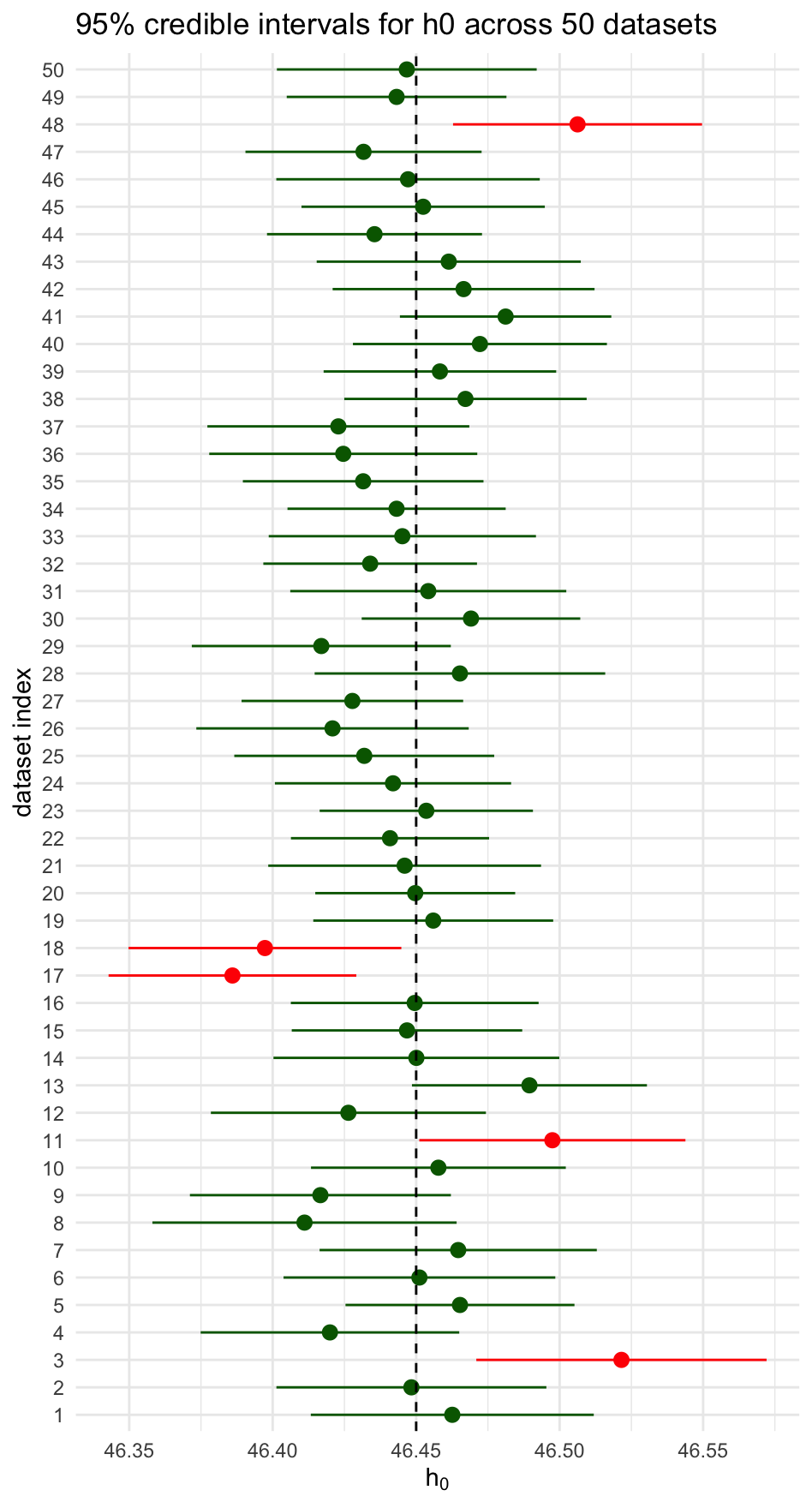}
  \end{subfigure}
  \hfill
  \begin{subfigure}{0.28\textwidth}
    \centering
    \includegraphics[width=\linewidth]{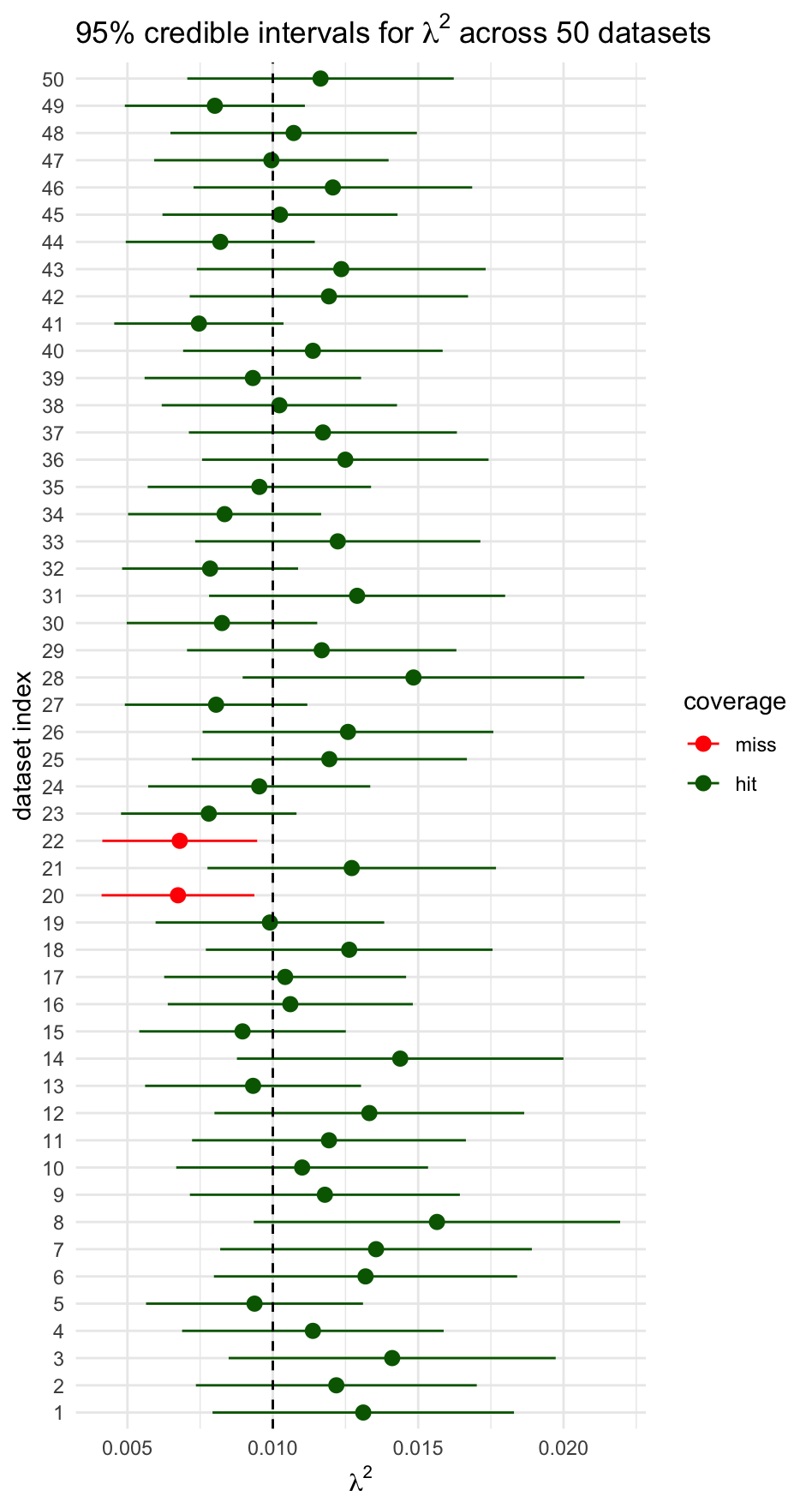}
  \end{subfigure}
  \caption{Posterior $95\%$ credible intervals for the calibration parameters $g_e$ (left), $h_0$ (middle) and the noise variance $\lambda^2$ (right) across the $50$ replicated datasets simulated under $\MF_0$ (classical Gaussian process). Each horizontal segment corresponds to one dataset; dots indicate the posterior medians and segments represent the $95\%$ equal-tailed credible intervals. The vertical dashed line marks the true parameter value used to generate the data, and intervals are coloured according to whether they cover the truth.}
  \label{fig:CI50_all_1}
\end{figure}

Figure~\ref{fig:CI50_all_1} shows that, for the calibration parameters $g_e$ and $h_0$, the credible intervals are generally centred close to the true values and cover them in nearly all replications. The few deviations appear to be isolated rather than systematic, which is consistent with stable calibration of the physical parameters under $\MF_0$. For $\lambda^2$, the intervals that miss the truth tend to do so on the same side of the dashed line, pointing to a mild systematic component in the recovered noise variance rather than to a coverage deficiency as such. Such sensitivity of variance components is common in hierarchical calibration settings, and it motivates interpreting inference on $\lambda^2$ more cautiously than inference on the structural parameters.
Overall, the figure supports accurate recovery of the physical parameters under $\MF_0$.

\subsection{Sensitivity to the allocation mechanism and to the sample size under $\MF_0$}\label{sup:m0_sensitivity}

This subsection reports a sensitivity analysis of the experiment of Section 4.1 of the main paper under $\MF_0$ (Figure~1 
in the main text), in which $50$ datasets of size $n=45$ are simulated from $\MF_0$ and inference is carried out under the mixture model $\MF_\alpha$ with the standard allocation mechanism. Two factors are varied independently and then jointly: (i)~the allocation mechanism between $\MF_0$ and $\MF_1$, replaced by its threshold-based version that prevents negligible posterior discrepancy values $\delta(x_i)$ from being assigned to the discrepancy component, and (ii)~the sample size, increased from $n=45$ to $n=100$. All other elements of the simulation design — priors, MCMC settings, true parameter values $(\boldsymbol{\theta}^\ast,\lambda^{2\ast})$, and the $50$ replicated datasets per scenario — are identical to those of the main-text experiment, so that the differences observed across Figures~\ref{fig:theta_alpha-classic_45_Seuil}-\ref{fig:theta_alpha-classic_100_Seuil} can be attributed unambiguously to these two factors.

\begin{figure}[h!]
    \centering
    \includegraphics[width=0.7\textwidth, height=0.4\textheight]{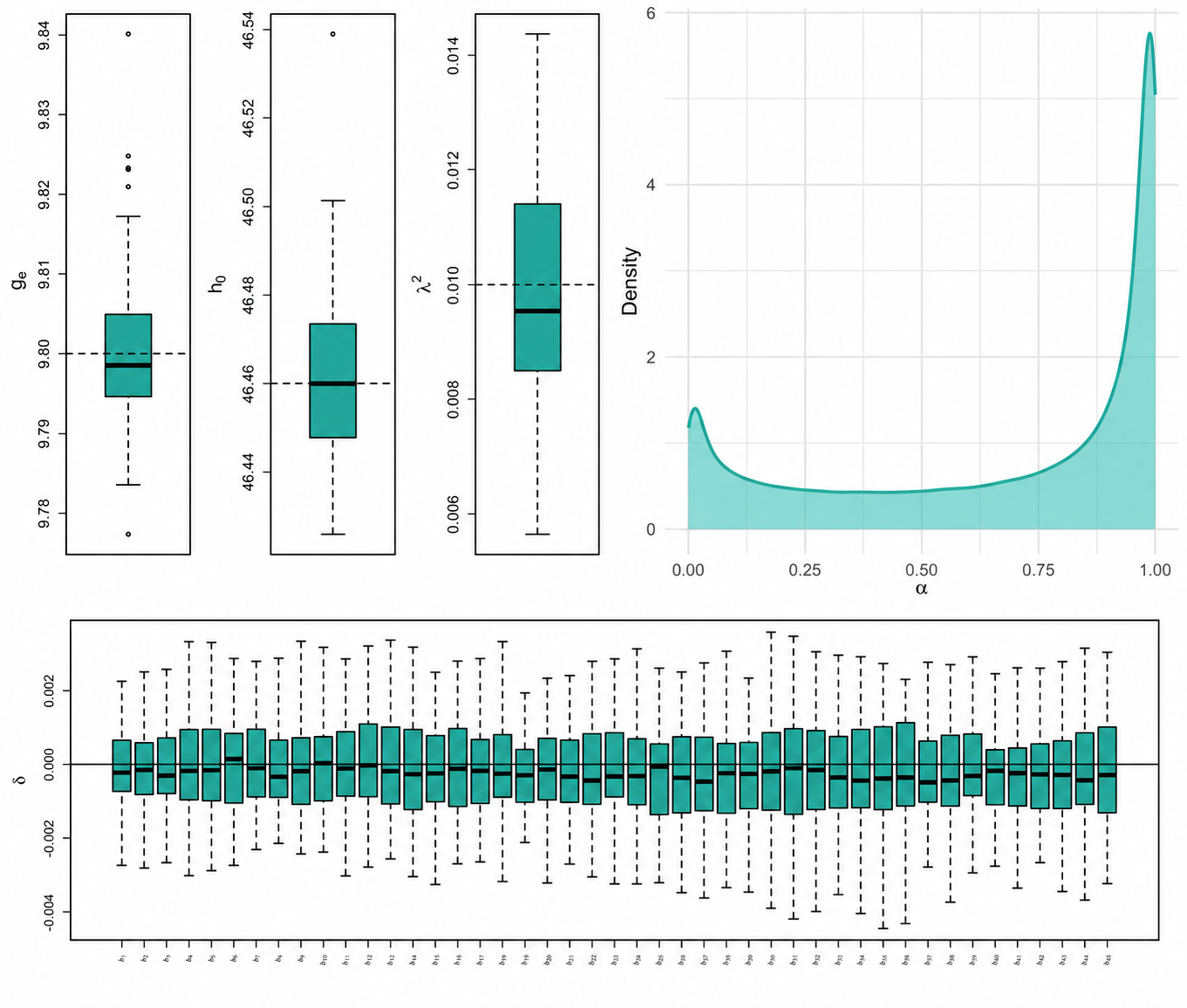}
    \caption{Results for $50$ datasets of size $n=45$ simulated under $\MF_0$ using the threshold-based allocation mechanism. (Top left): distributions of the posterior means of the model parameters compared to the true values, shown by horizontal dashed lines. (Top right): pooled posterior density of $\alpha$ across the $50$ datasets. (Bottom): distributions of the posterior means of the discrepancy values $\delta(x_i)$ across the $50$ datasets. The threshold mechanism prevents negligible discrepancy values from driving the allocation toward the discrepancy component.}
    \label{fig:theta_alpha-classic_45_Seuil}
\end{figure} 

\begin{figure}[h!]
    \centering
    \includegraphics[width=0.7\textwidth, height=0.4\textheight]{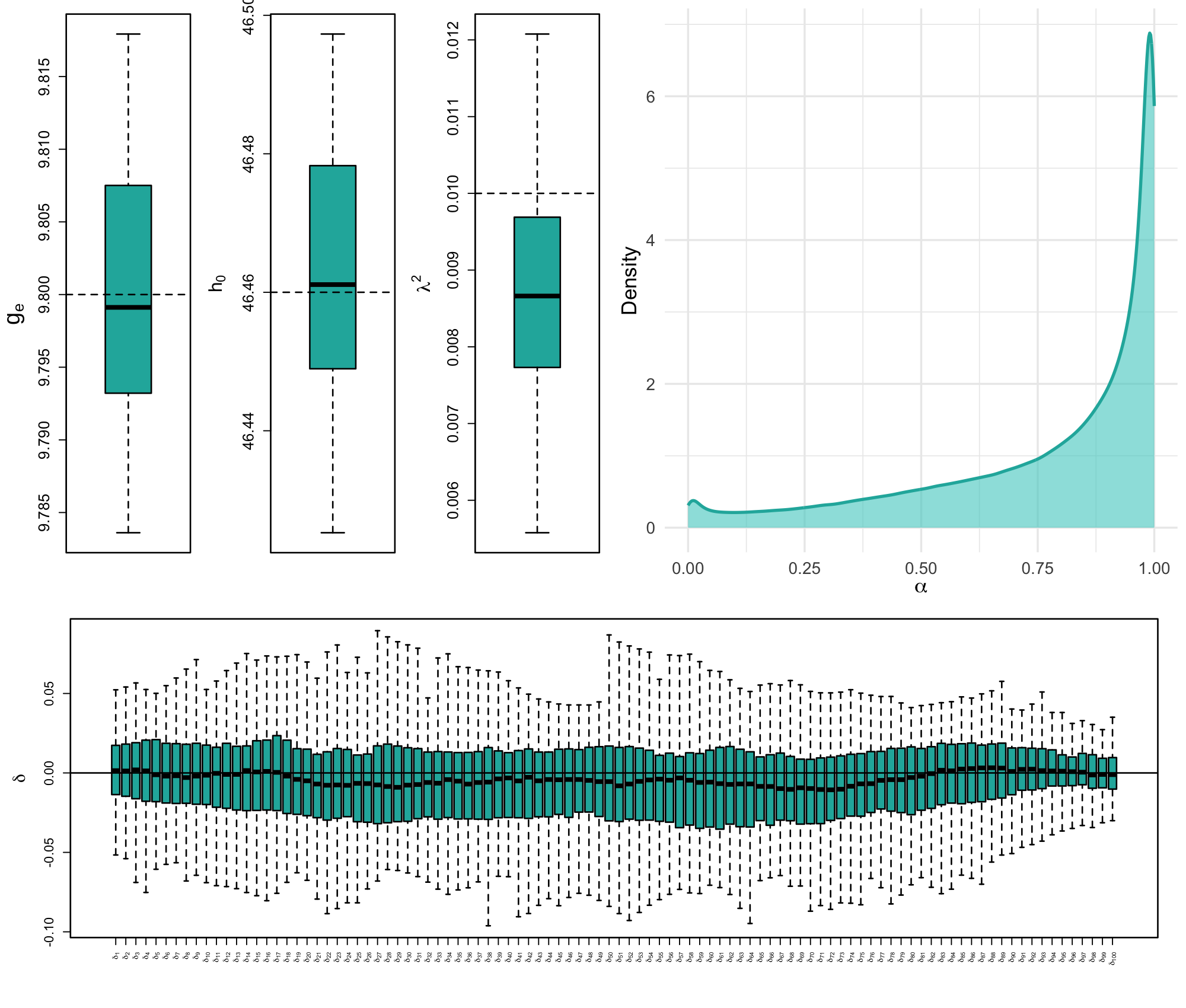}
     \caption{Results for $50$ datasets of size $n=100$ simulated under $\MF_0$. (Top left): distributions of the posterior means of the model parameters compared to the true values, shown by horizontal dashed lines. (Top right): pooled posterior density of $\alpha$ across the $50$ datasets. (Bottom): distributions of the posterior means of the discrepancy values $\delta(x_i)$ across the $50$ datasets.}
    \label{fig:theta_alpha-classic_100}
\end{figure} 

\begin{figure}[h!]
    \centering
    \includegraphics[width=0.7\textwidth, height=0.4\textheight]{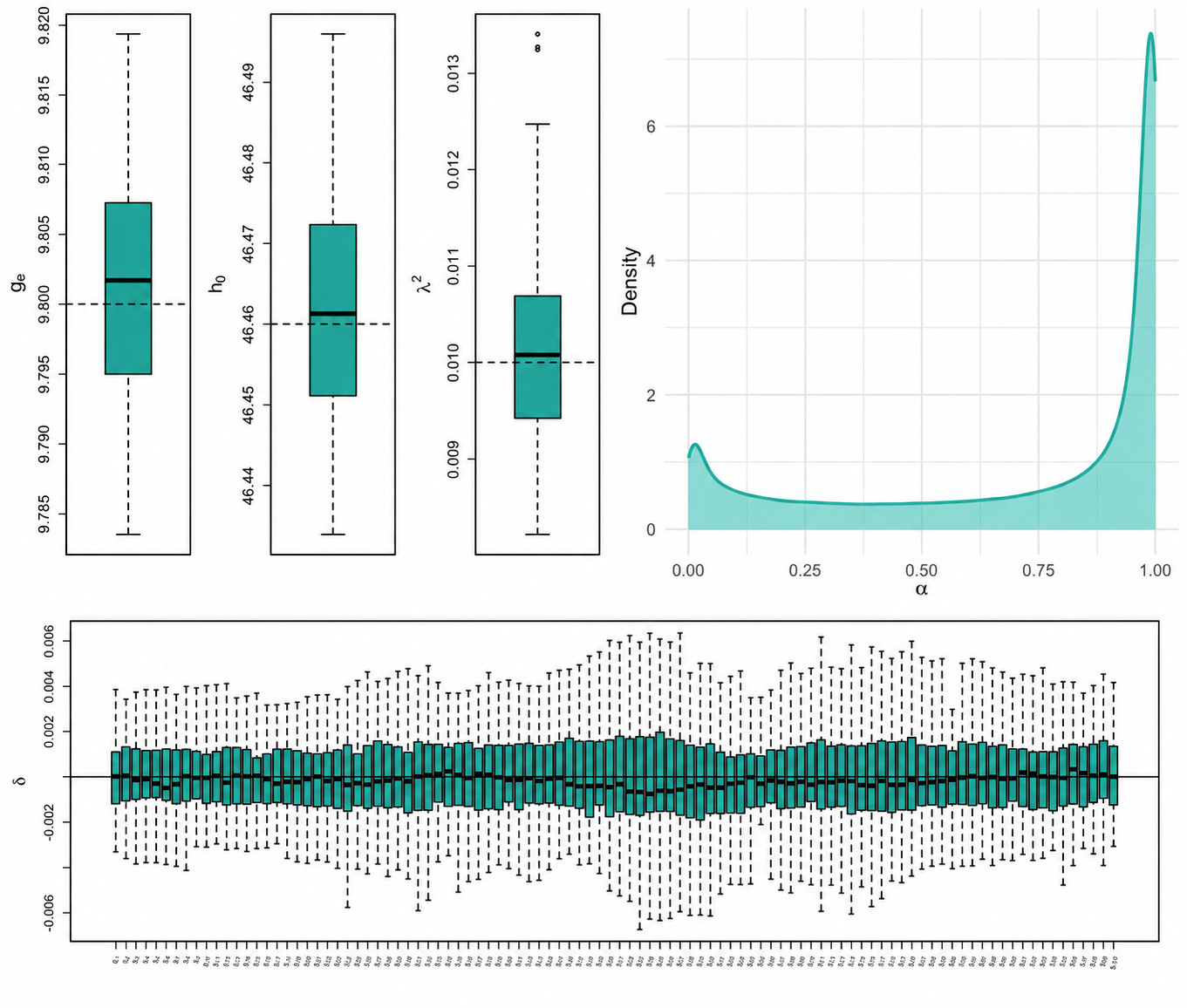}
    \caption{Results for $50$ datasets of size $n=100$ simulated under $\MF_0$ using the threshold-based allocation mechanism. (Top left): distributions of the posterior means of the model parameters compared to the true values, shown by horizontal dashed lines. (Top right): pooled posterior density of $\alpha$ across the $50$ datasets. (Bottom): distributions of the posterior means of the discrepancy values $\delta(x_i)$ across the $50$ datasets. Compared with the standard allocation mechanism, the thresholded version gives a more conservative interpretation of small posterior discrepancy values.}
    \label{fig:theta_alpha-classic_100_Seuil}
\end{figure}
Overall, the sensitivity analysis confirms that the proposed approach is robust to both the choice of the allocation mechanism and the sample size, and that its two natural refinements — thresholding and additional data — act jointly to improve the calibration and the model-selection behaviour of $\MF_\alpha$ when the data are generated under $\MF_0$.

\subsection{More simulation study results}\label{sup:scnI_III}
 
\paragraph{ Baseline mixture inference without thresholding, with free $(g_e,h_0)$.}
We fit the original mixture model on the same synthetic datasets as in
Section~4.1, Figure 3 of the main paper, but without any thresholding and without fixing $g_e$: both $g_e$ and $h_0$ are estimated from the data, and latent allocations are driven only by the global mixture mechanism through $\alpha$.
For a representative replicated dataset we also examine the pointwise
posterior inclusion probabilities 
$\hat p_i = \P\bigl(\zeta_i = 1 \mid \pmb{y}, \pmb{X}, \pmb{\theta}, \delta, \lambda^2, k, \gamma_\delta, \alpha\bigr)$.

\paragraph{Thresholded allocations with $(g_e,h_0)$ free.}
We now use exactly the same thresholding rule as in the main text (Figure 3, threshold $s=0.3$), but we let both $g_e$ and $h_0$ be estimated from the data. The comparison with the thresholded, $g_e$-fixed case of the main paper therefore isolates the effect of the calibration-discrepancy confounding under a fixed thresholding mechanism.

\begin{figure}[h!]
        \centering
        \includegraphics[width=0.5\textwidth]{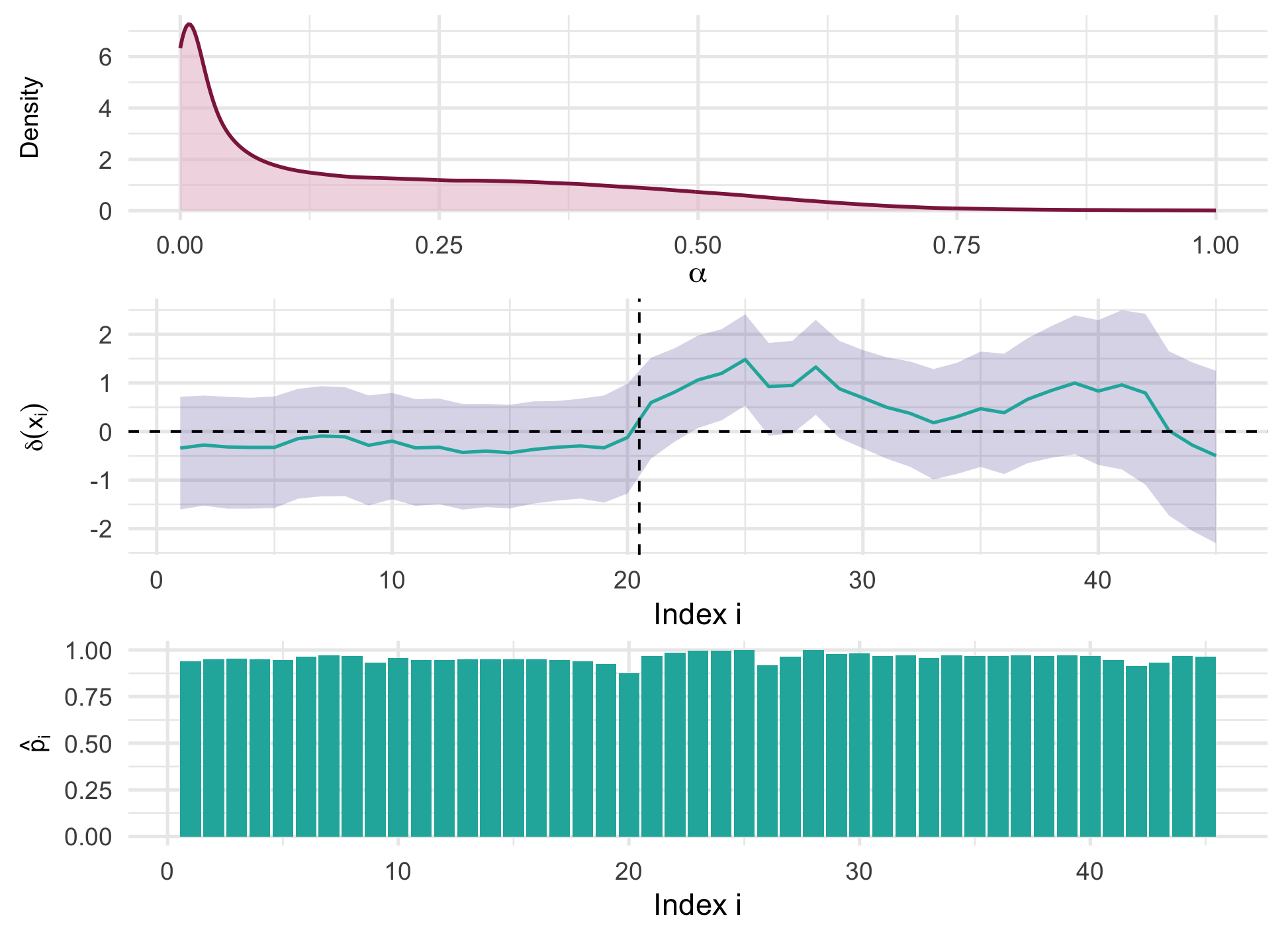}\includegraphics[width=0.5\textwidth]{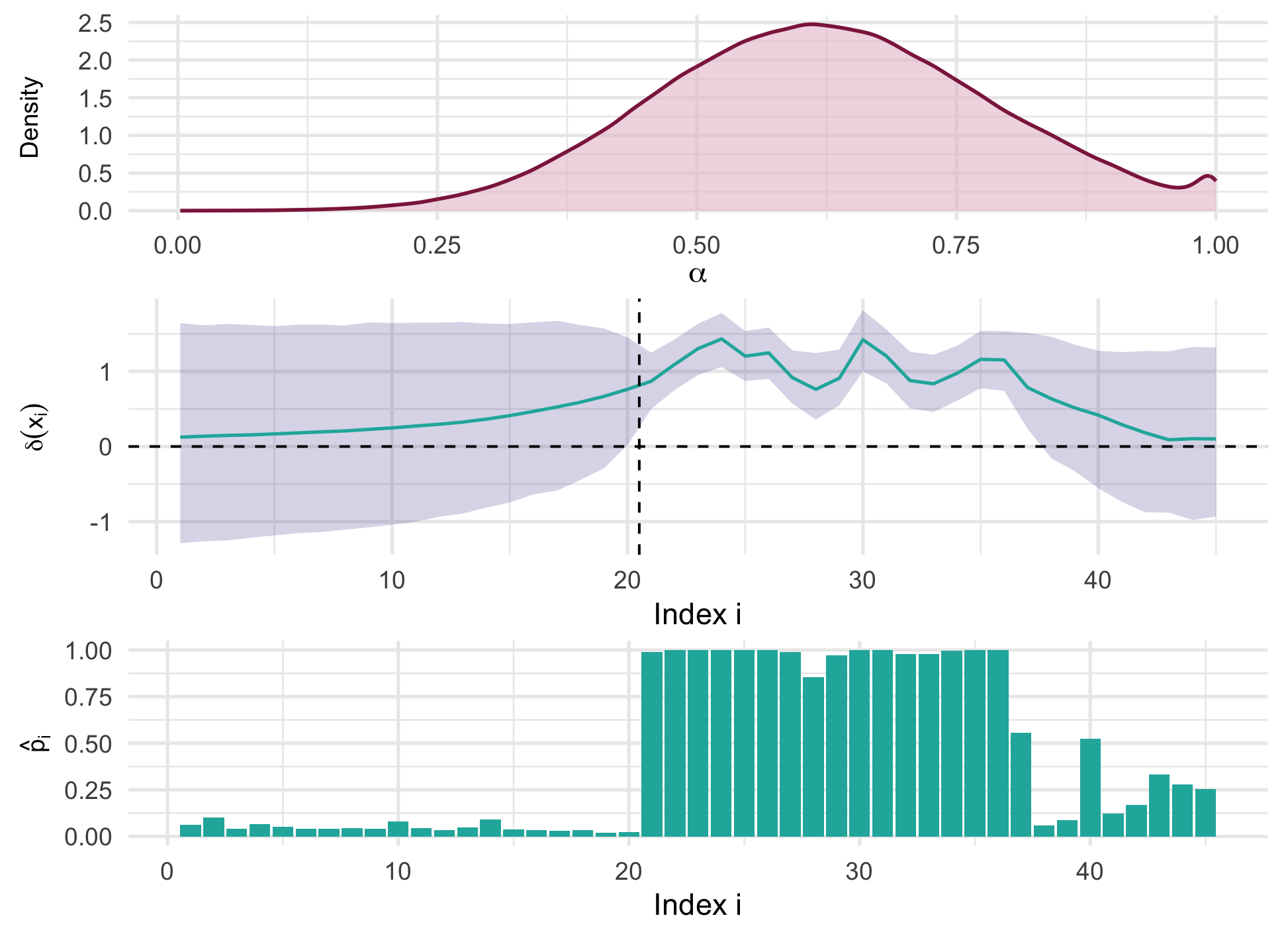}
        \caption{\small Bayesian inference : (Left panel) without thresholding and (Right panel) with thresholding, when all of the parameters $g_e, h_0,\delta,\lambda^2,\alpha$ are estimated. (Top) Pooled posterior density of $\alpha$. (Middle) Posterior means of $\delta(x_i)$ with $95\%$ credible bands plotted for one representative dataset. The vertical dashed line marks the true activation boundary at $i=n_0+0.5$. (Bottom) Pointwise posterior inclusion probabilities $\hat p_i$.}
        \label{fig:seuil_senario1_3}
\end{figure}

Figure~\ref{fig:seuil_senario1_3} left shows that without thresholding the model detects a discrepancy somewhere, but does not localize it sharply: the posterior mean of $\delta(x_i)$ does shift upward after the true boundary, but the credible
bands still overlap zero near the transition, and the pointwise inclusion
probabilities remain high across the whole input domain including the
region where the true discrepancy is negligible. Once the
discrepancy-corrected component is globally preferred, the standard mixture
model tends to use it almost everywhere. This is precisely the setting in
which a threshold-based local diagnostic is useful, and the corresponding case in the main paper (thresholding with $g_e$ fixed) shows the improvement obtained by combining thresholding with a
fixed $g_e$.

Figure~\ref{fig:seuil_senario1_3} right shows that thresholding on its own-without also
reducing the calibration-discrepancy confounding-yields a much less
contrasted local diagnostic than the thresholded, $g_e$-fixed case of the main paper. The posterior of $\alpha$ is no longer concentrated on small values: its mass spreads over intermediate values
and part of it even moves above $1/2$. The posterior mean of $\delta(x_i)$ is
shrunk strongly towards zero across most of the trajectory, and the pointwise
inclusion probabilities $\hat p_i$ remain low almost everywhere, with only a
mild increase near the true change-point. In other words, once $g_e$ is free,
part of the systematic lack-of-fit that should be attributed to $\delta(X)$
is re-expressed as a shift in the calibration parameter, the threshold is
triggered less often, and the local allocation pattern becomes harder to read.

Taken together, these two experiments confirm the interpretation given in the main paper: sharp local detection of the active region requires \emph{both} the thresholded allocation rule \emph{and} a reduction of the calibration--discrepancy confounding. Neither ingredient alone is sufficient. The thresholded, $g_e$-fixed case of the main paper combines the two and produces the
clearest local allocation map.
 
\subsection{Inference under the OGP prior when data are simulated from a classical GP}\label{sup:ogp_sim}
 
In this experiment, we evaluate the robustness of the OGP discrepancy model
introduced in Section~\ref{sup:ogp} under a mild prior misspecification. For
each value of the true correlation length
$\gamma_\delta^\ast\in\{0.01,0.1,0.2,\ldots,0.9\}$, we generate $50$ datasets
of size $n=45$ from
$y_i = g(x_i)\boldsymbol{\theta}^\ast+\delta^\ast(x_i)+\varepsilon_i$,
with the discrepancy $\delta^\ast$ simulated from the \emph{classical}
(non-orthogonal) GP with exponential correlation kernel, while \emph{inference}
is performed under the OGP prior.

\begin{figure}[h!]
\centering
\includegraphics[width=0.9\textwidth]{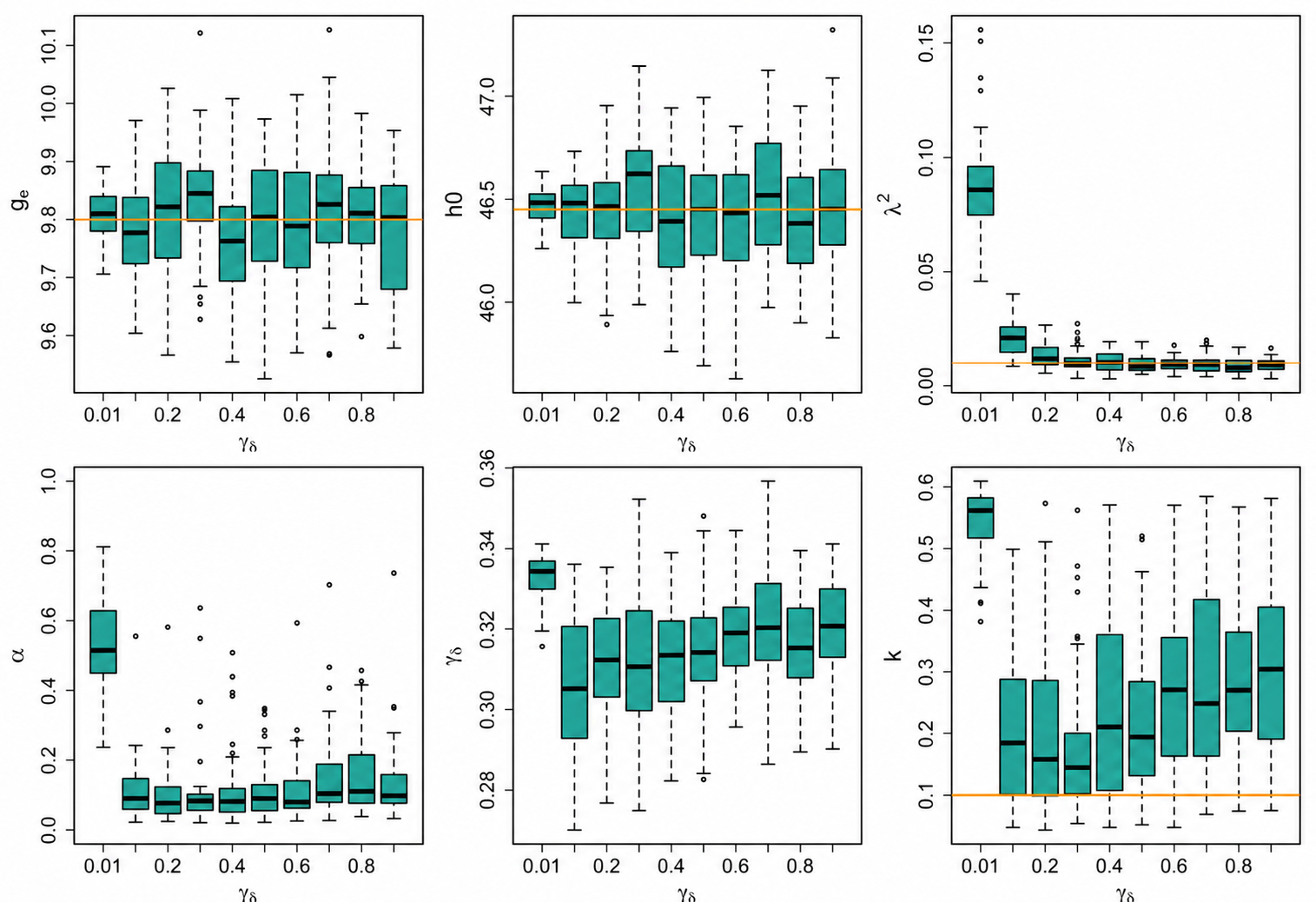}
\caption{\textbf{Inference under the OGP prior, data simulated from a classical GP.}
For each $\gamma_\delta^\ast\in\{0.01,0.1,\ldots,0.9\}$ and $50$ replicated datasets of
size $n=45$, boxplots show the posterior means across replications for
$g_e$, $h_0$, $\lambda^2$, $\alpha$, $\gamma_\delta$ and $k$. Red dashed horizontal
lines indicate the true values. The MCMC is run for $10{,}000$ iterations with a burn-in of $2{,}000$.}
\label{fig:ogp_inf}
\end{figure}

The aim of this experiment is to check that adopting the OGP prior does not degrade inference when the discrepancy is in fact generated by a classical GP, i.e.\ under a mild prior misspecification. The results
(Figure~\ref{fig:ogp_inf}) are qualitatively similar to those obtained under the classical GP prior in the main paper. The calibration parameters $(g_e,h_0)$ are recovered stably around their true values across the whole range of $\gamma_\delta^\ast$. The posterior of the mixture weight $\alpha$ concentrates near zero once $\gamma_\delta^\ast\gtrsim 0.1$, correctly favouring the discrepancy-corrected component, and the $\gamma_\delta^\ast=0.01$ case remains a stress test in which the GP discrepancy becomes nearly
indistinguishable from additional white noise: there $\alpha$ shifts upward and becomes more dispersed, while $\lambda^2$ and $k$ inflate accordingly. The variance-related hyper-parameters $(\lambda^2,k,\gamma_\delta)$ display the usual amplitude range trade-off but remain centred in plausible regions.

\subsection{Oracle identifiability diagnostic for $(k,\gamma_\delta)$}\label{sup:oracle}
 
A well-known challenge in Bayesian calibration with Gaussian-process discrepancy
is that variance-related hyper-parameters are often weakly identified
from finite datasets: similar likelihood values can be obtained by trading
off measurement noise against discrepancy magnitude, and by compensating the
discrepancy amplitude with its correlation range
\citep{kennedy2001,brynjarsdottir2014}.
This is consistent with
what we observe under $\MF_1$ in the main paper, where $(g_e,h_0)$ remain stable
while the dispersion of $(k,\gamma_\delta)$ is markedly larger.
 
To disentangle intrinsic identifiability limitations from potential
implementation issues, we run an \emph{oracle} diagnostic in which inference
is performed \emph{conditionally on the realized discrepancy}: for each
replicated dataset we condition on the simulated discrepancy path $\delta^\ast$
and treat it as known throughout the run. We also remove mixture-allocation
uncertainty by assigning all observations to the discrepancy-corrected
component ($\zeta_i\equiv 1$), so that the information about
$(k,\gamma_\delta)$ comes solely from the GP model for $\delta(X)$.
We compare three configurations:
\begin{itemize}
\item \textbf{Scenario A:} $\delta$ fixed; both $k$ and $\gamma_\delta$ are estimated.
\item \textbf{Scenario B:} $\delta$ fixed; $k=k^\ast$ fixed; only $\gamma_\delta$ is estimated.
\item \textbf{Scenario C:} $\delta$ fixed; $\gamma_\delta=\gamma_\delta^\ast$ fixed; only $k$ is estimated.
\end{itemize}
 
\begin{figure}[h!]
\centering
\begin{subfigure}[t]{0.32\textwidth}
  \centering
  \includegraphics[width=\linewidth]{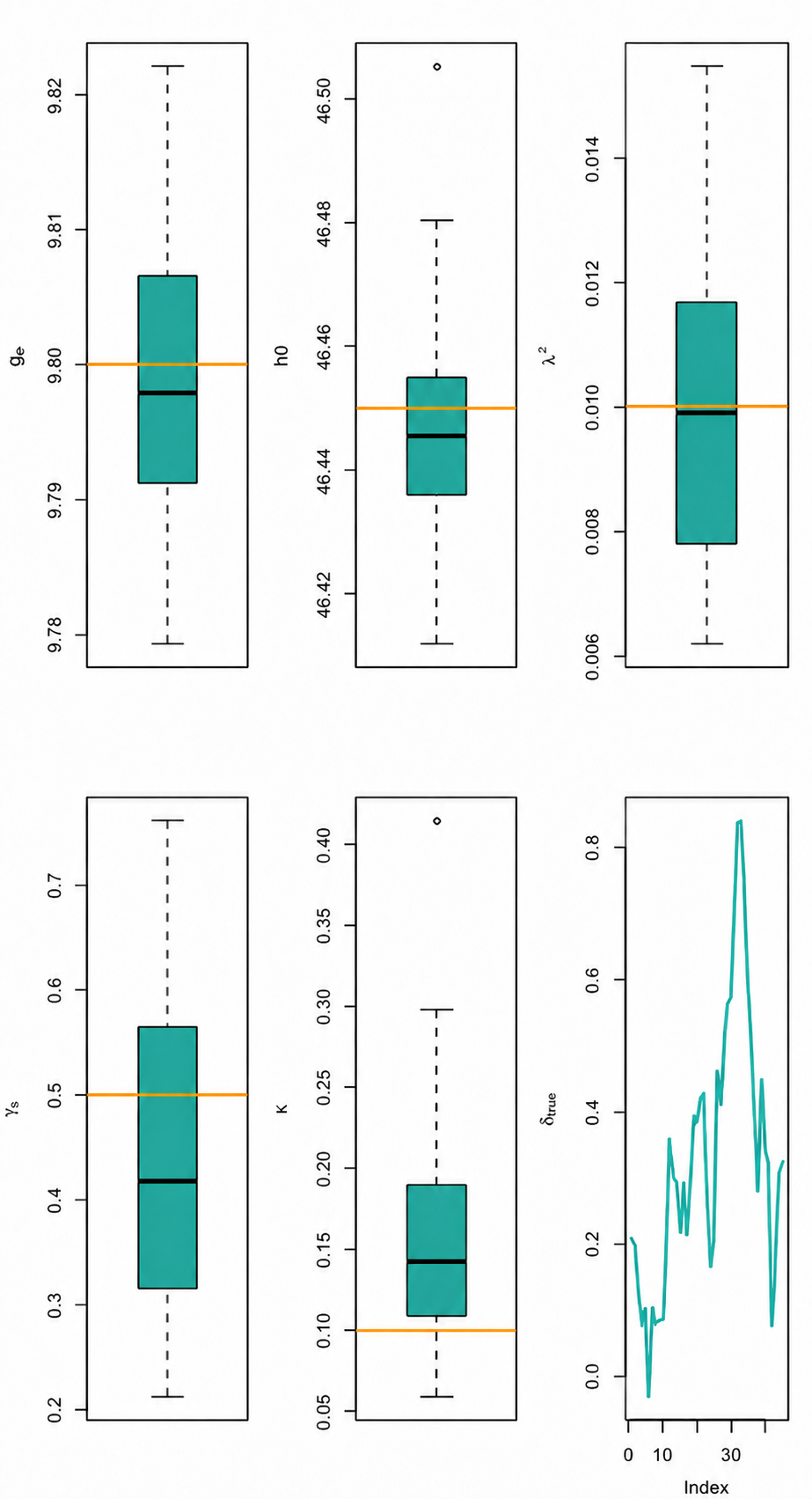}
  \caption{Scenario~A: $k$, $\gamma_\delta$ estimated.}
\end{subfigure}\hfill
\begin{subfigure}[t]{0.32\textwidth}
  \centering
  \includegraphics[width=\linewidth]{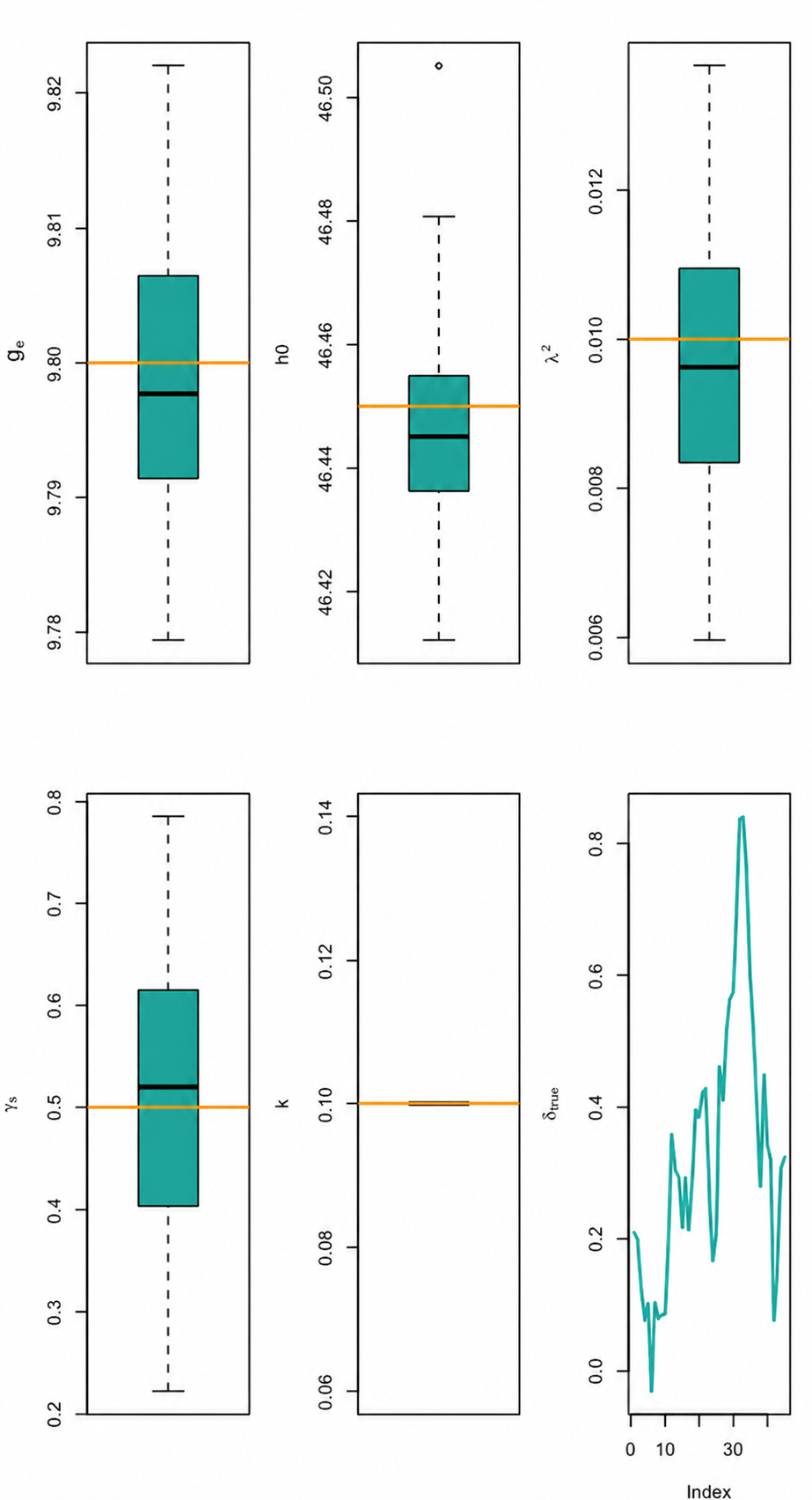}
  \caption{Scenario~B: $k=k^\ast$ fixed; $\gamma_\delta$ estimated.}
\end{subfigure}\hfill
\begin{subfigure}[t]{0.32\textwidth}
  \centering
  \includegraphics[width=\linewidth]{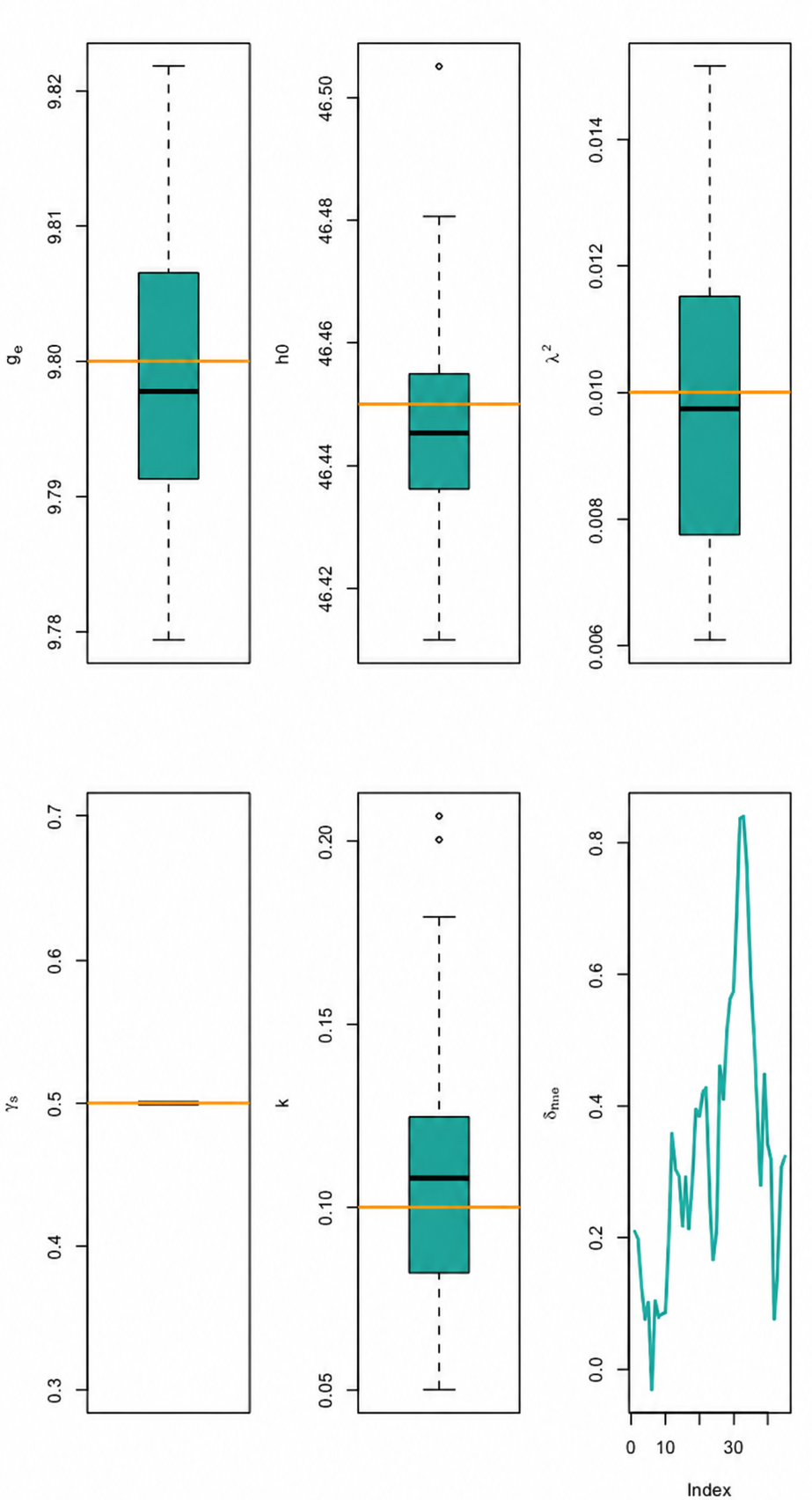}
  \caption{Scenario~C: $\gamma_\delta=\gamma_\delta^\ast$ fixed; $k$ estimated.}
\end{subfigure}
\caption{\textbf{Oracle diagnostic conditional on the true discrepancy.}
Posterior-mean summaries across $50$ replicated datasets (boxplots) for
$(g_e,h_0,\lambda^2,\gamma_\delta,k)$ under three oracle configurations. The
discrepancy path $\delta^\ast$ is held fixed during inference; horizontal lines
indicate the true values.}
\label{fig:oracle}
\end{figure}

Figure~\ref{fig:oracle} shows that the calibration parameters $(g_e,h_0)$ are
recovered stably across the three configurations, and $\lambda^2$ is reasonably
calibrated once $\delta$ is fixed: the main numerical difficulty observed in
the full model is \emph{not} driven by the updates for $(g_e,h_0,\lambda^2)$. The
picture differs for the discrepancy hyper-parameters. In Scenario~A, the
posterior means of $(k,\gamma_\delta)$ exhibit noticeable spread and tend to
move in opposite directions: replications with smaller inferred
$\gamma_\delta$ are often associated with larger values of $k$, and conversely.
This is the typical ``range vs.\ amplitude'' confounding in GP discrepancy
models, here expressed through $\sigma_\delta^2=\lambda^2/k$: over a finite
design, rather different $(k,\gamma_\delta)$ pairs can induce similar prior
densities for the realized $\delta^\ast$ and lead to comparable fits.
Scenarios~B and~C make this confounding explicit. When $k$ is fixed at
$k^\ast$, the posterior means of $\gamma_\delta$ concentrate tightly around
$\gamma_\delta^\ast$; conversely, when $\gamma_\delta$ is fixed, the posterior
means of $k$ concentrate around $k^\ast$. Taken together, these oracle results
point to a genuine identifiability limitation among variance-related
components rather than to a coding issue in the MCMC implementation, and this
limitation concerns mainly the discrepancy hyper-parameters. Crucially, it
does not undermine the model-choice objective, which is driven by the
posterior behaviour of the mixture weight $\alpha$ in the main experiments.

\section{Additional real-data results: Blue Basketball with $g_e$ free}\label{sup:real}

In Section~4.2 of the main paper, the real-data analysis of the \emph{Blue
Basketball} dataset uses $g_e$ fixed at its nominal value $9.8\,\mathrm{m/s^2}$
in order to reduce the calibration--discrepancy confounding. For completeness,
we report here the companion experiment in which $g_e$ is left free and is
estimated jointly with $h_0$, $\lambda^2$, $k$, $\gamma_\delta$, $\alpha$ and
the discrepancy process $\delta(X)$. We use the same classical GP prior on
$\delta$ and the same unthresholded allocation rule as in the first paragraph
of Section~4.2.

\begin{figure}[h!]
    \centering
    \begin{subfigure}[t]{0.49\textwidth}
        \centering
        \includegraphics[width=\linewidth]{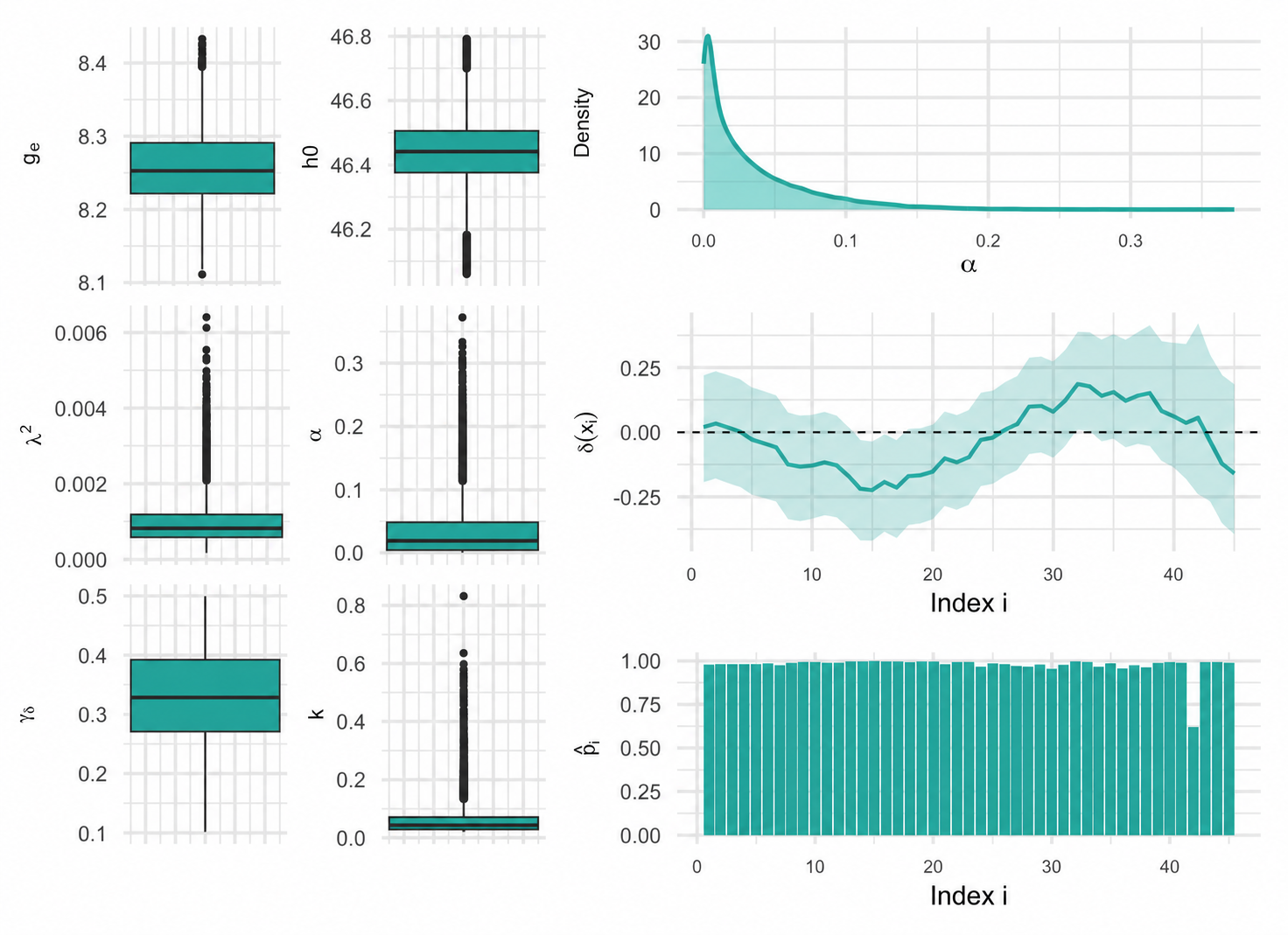}
        \caption{Unthresholded mixture.}
        \label{fig:real_blue_basketball_aa}
    \end{subfigure}\hfill
    \begin{subfigure}[t]{0.49\textwidth}
        \centering
        \includegraphics[width=\linewidth]{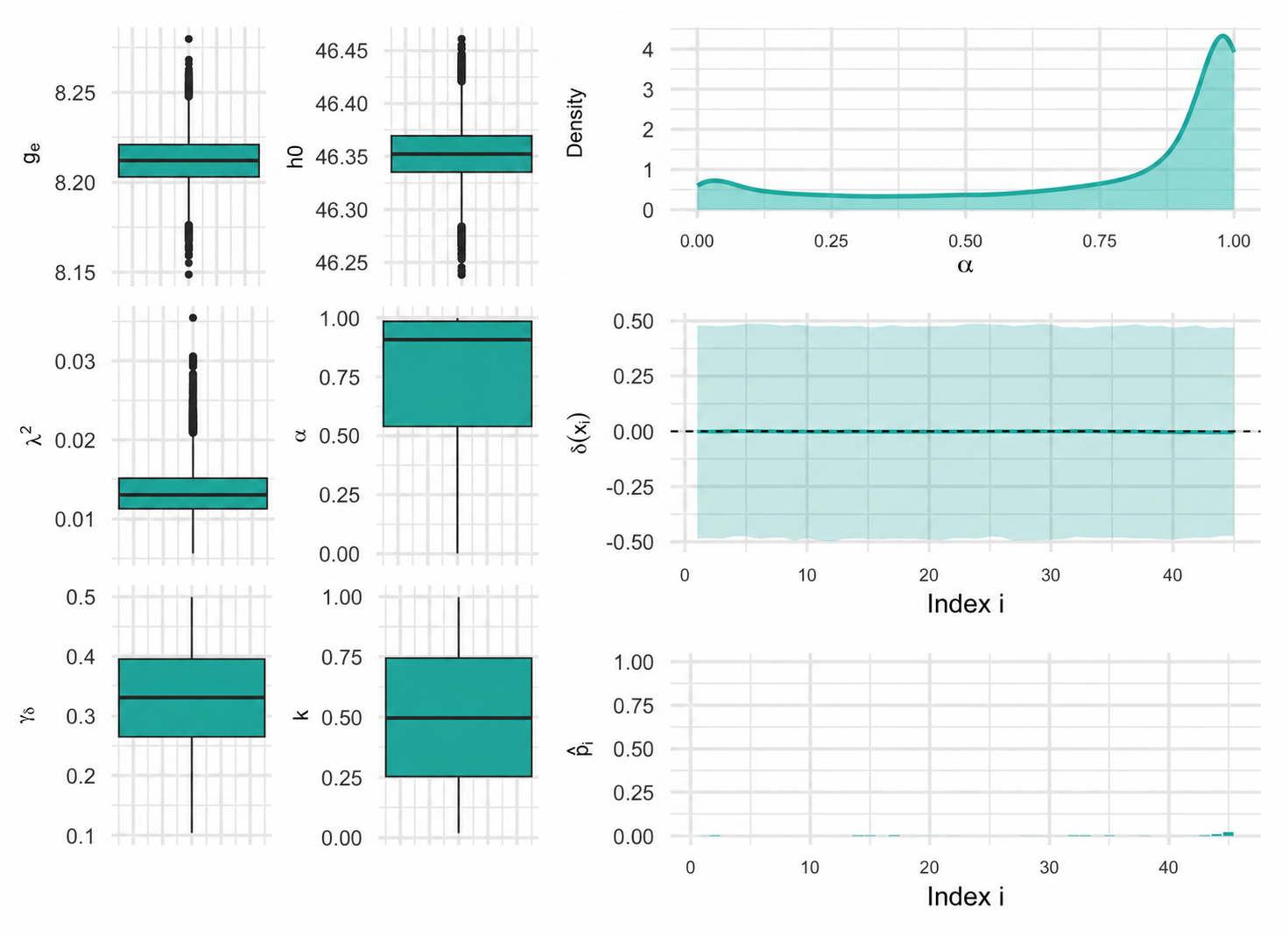}
        \caption{Thresholded mixture, $s=0.3$.}
        \label{fig:real_blue_basketball_ba}
    \end{subfigure}
    \caption{\small Real data application (Blue Basketball), $g_e$ estimated jointly with the other parameters. Posterior summaries: (Left-Boxplots) Posterior distribution of the MCMC draws; (Top-Right-panel) Posterior density of $\alpha$ samples; (Middle-Right-panel) Posterior mean (solid-line) of $\delta(x_i)$ with pointwise $95\%$ credible bands; (Bottom-Right-panel) Pointwise posterior inclusion probabilities $\hat p_i$.}
    \label{fig:real_blue_basketballa}
\end{figure}

Figure~\ref{fig:real_blue_basketballa}a is qualitatively similar to the analogous fit with
$g_e$ fixed reported in the main paper. The posterior density of $\alpha$ is
concentrated near zero, so the mixture assigns most of its mass to the
discrepancy-corrected component. The posterior summary of $\delta(x_i)$ stays
away from zero across much of the domain, reflecting that the real trajectory
exhibits mild but coherent departures from the idealized free-fall model. The
posterior mass of $(\lambda^2,k)$ concentrates near small values, reflecting
the well-known confounding between measurement noise and discrepancy amplitude:
a good fit can often be obtained by reducing the noise variance while letting
the GP absorb the remaining variability. Fixing $g_e$ in the main paper makes
this confounding more explicit (since systematic lack-of-fit can no longer be
absorbed by a shift in $g_e$) but does not qualitatively change the global
conclusion: at the dataset level, the mixture prefers the
discrepancy-corrected component; at the local level (after thresholding, with
$g_e$ fixed, Section~4.2 of the main paper), the discrepancy-corrected component
is selected only beyond a clear transition point along the trajectory.

Figure~\ref{fig:real_blue_basketballa}b shows that the posterior of the gravitational constant concentrates around $g_e\approx 8.21$, with $h_0\approx 46.35$, that is, a slightly reduced effective gravity that absorbs a part of the systematic positive bias caused by air drag. With these calibration values, the data-versus-physics residuals lie within $[-0.50, 0.50]$ across the entire trajectory, so all of them fall below the threshold level $s=0.3$.
 
The thresholding device then has a strong consequence: at every index $i$, the local evidence is judged compatible with $\MF_0$ and the pointwise inclusion probabilities are essentially zero, $\hat p_i \approx 0$ for all $i$. The posterior of $\alpha$ concentrates near one, in agreement with the fact that the data are now globally allocated to the no-discrepancy component. Since no observation is allocated to $\MF_1$, the discrepancy $\delta(x_i)$ is no longer informed by the likelihood and its posterior collapses to the prior: its mean is essentially zero across the whole domain, with symmetric pointwise $95\%$ credible bands of approximately
$\pm 0.5$, uniform in $i$. The hyperparameters $(\lambda^2,k,\gamma_\delta)$ are similarly close to their prior under this regime.
 
This experiment isolates the role of the $g_e$--$\delta$ confounding on real data. When $g_e$ is allowed to vary, the calibration parameter alone explains enough of the systematic departure from a vacuum trajectory that the local residuals fall below the threshold; the thresholded mixture then turns off the discrepancy term entirely and provides no localized information.

\section{Comparison of classical GP and orthogonal GP priors on real data}\label{sup:cgp-ogp-real}
Section~\ref{sup:ogp} introduces the orthogonal Gaussian process prior on $\delta(X)$, which imposes the constraint $\langle \delta, g(x)\boldsymbol{\theta} \rangle = 0$ in order 
to prevent the discrepancy from absorbing the contribution of the calibration parameters. A natural question is whether this constraint effectively resolves the identifiability issue between 
$\boldsymbol{\theta}$ and $\delta(X)$ on real data. To investigate this, we run the mixture model on the Blue Basketball dataset with all parameters free ($g_e$ estimated jointly with the other parameters, no thresholding) under the orthogonal GP prior, and compare it with the classical-GP fit already reported in Figure~\ref{fig:real_blue_basketball}. Figure~\ref{fig:cgp-ogp-comparison} shows the orthogonal-GP result; the two analyses produce essentially identical posterior summaries across all parameters. In particular:
\begin{itemize}
    \item the posterior of $\alpha$ concentrates near the same values 
    in both cases, indicating that the mixture reaches the same 
    conclusion about the relative weight of the 
    discrepancy-corrected component;
    \item the posterior mean of $\delta(x_i)$ and the pointwise 
    $95\%$ credible bands are visually indistinguishable between the 
    two priors;
    \item the pointwise inclusion probabilities $\hat p_i$ exhibit 
    the same pattern, with uniform allocation across observations.
\end{itemize}

\begin{figure}[h!]
    \centering
    \includegraphics[width=0.6\textwidth]{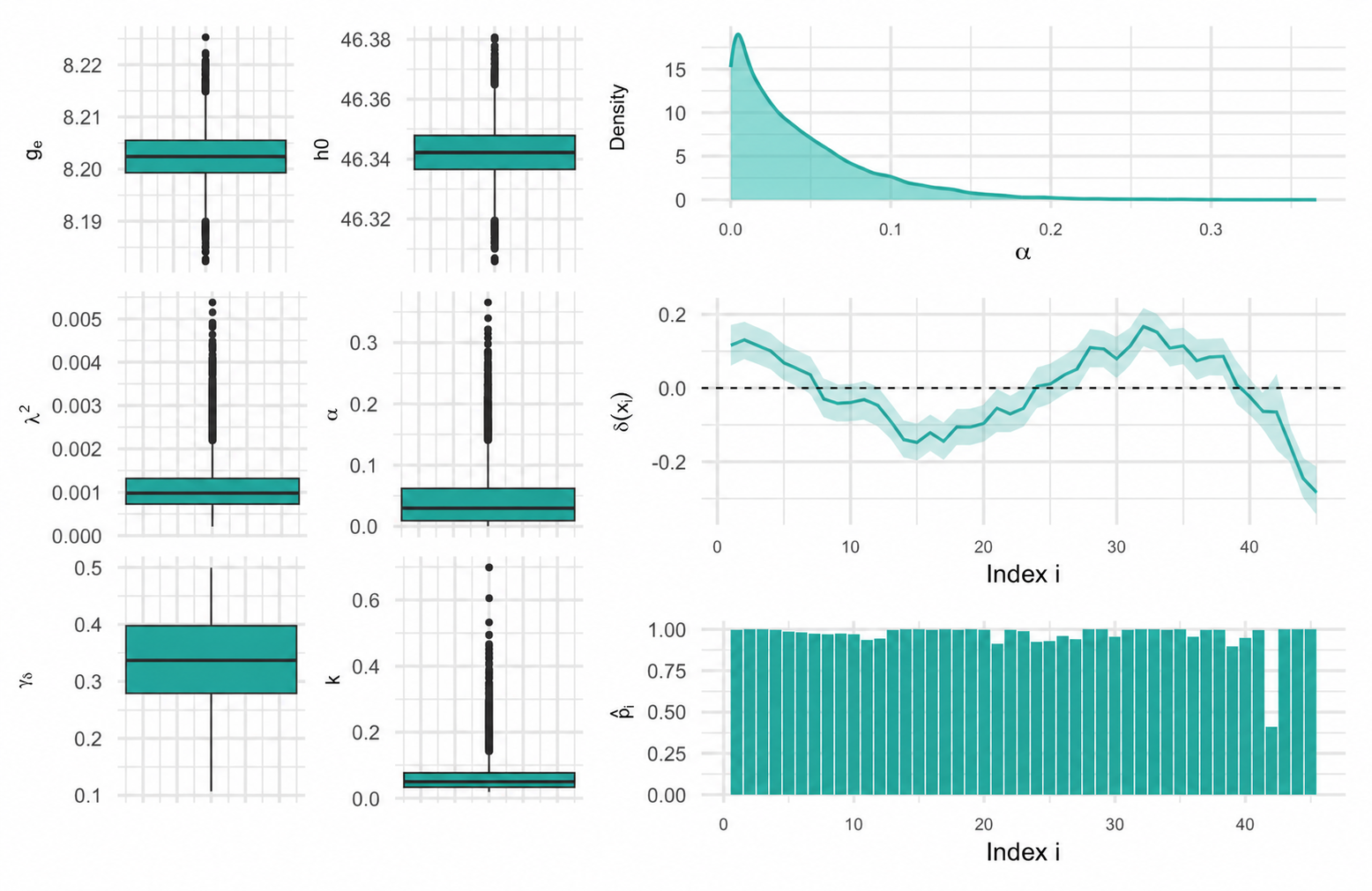}
    \caption{\textbf{Orthogonal GP prior on the Blue Basketball real dataset, $g_e$ free, no thresholding.}
    From left to right: marginal posterior boxplots of the calibration and discrepancy hyperparameters; posterior density of $\alpha$; posterior mean of $\delta(x_i)$ with pointwise $95\%$ credible bands; pointwise inclusion probabilities $\hat p_i$. The corresponding classical-GP fit is shown in Figure~\ref{fig:real_blue_basketball_a}; the two are essentially identical, showing that the orthogonality constraint does not resolve the identifiability issue between $\boldsymbol{\theta}$ and $\delta(X)$ in this example.}
    \label{fig:cgp-ogp-comparison}
\end{figure}

This result shows that in this application the orthogonality 
constraint does not provide additional identifiability for the 
calibration parameters compared to the classical GP. The confounding 
between $\boldsymbol{\theta}$ and $\delta(X)$ is driven primarily 
by the linear structure of the physical model 
$g(x)\boldsymbol{\theta}$ and the flexibility of the GP, rather than 
by the choice of prior covariance kernel. Both approaches reach the 
same empirical conclusion: the mixture systematically allocates mass 
to the discrepancy-corrected component $\MF_1$, and the discrepancy 
function exhibits a coherent positive trend along the trajectory. 
This comparison justifies the use of the classical GP prior in the 
main analysis, as the simpler specification achieves identical 
inferential conclusions while preserving computational convenience.

}

\end{document}